\documentclass[11pt]{article}

\usepackage[margin=1in]{geometry}
\usepackage{amsmath,amssymb,amsthm}
\usepackage{graphicx}
\usepackage{tikz}
\usetikzlibrary{decorations.pathreplacing}
\usepackage{booktabs}
\usepackage{longtable}
\usepackage{makecell}
\usepackage[section]{placeins}
\usepackage{flafter}  
\usepackage{needspace}          
\usepackage{etoolbox}
\newif\ifsubsecbarrier \subsecbarriertrue
\pretocmd{\subsection}{\ifsubsecbarrier\FloatBarrier\fi}{}{}
\pretocmd{\paragraph}{\FloatBarrier}{}{}
\usepackage{hyperref}
\usepackage[round]{natbib}
\usepackage{multirow}
\usepackage{array}
\usepackage{xcolor}
\usepackage{enumitem}
\usepackage{multicol}
\usepackage{caption}
\usepackage{textcomp}
\usepackage{microtype}

\hypersetup{
  colorlinks=true,
  linkcolor=blue!70!black,
  citecolor=blue!70!black,
  urlcolor=blue!70!black
}

\newtheorem{proposition}{Proposition}
\newtheorem{assumption}{Assumption}
\newtheorem{remark}{Remark}

\newcommand{\E}{\mathbb{E}}
\newcommand{\Var}{\mathrm{Var}}
\newcommand{\Cov}{\mathrm{Cov}}
\newcommand{\R}{\mathbb{R}}
\newcommand{\calL}{\mathcal{L}}
\newcommand{\calU}{\mathcal{U}}
\newcommand{\iid}{\stackrel{\text{iid}}{\sim}}
\newcommand{\plim}{\text{plim}}
\DeclareMathOperator{\clip}{clip}

\title{\textbf{When Do Surrogate Metrics Work?\\A Finite-Sample Comparison Under Realistic Failure Modes}}

\author{
  Zihao Chen \\
  University of Washington \\
  \texttt{zchen05@uw.edu}
}

\date{September 2026}

\begin{document}

\maketitle

\begin{abstract}
Online experiments must often be evaluated before long-term outcomes mature. Under rolling enrollment, these outcomes are observed only for early enrollees, while short-term surrogates are available for everyone. We compare seven estimators across eleven data-generating processes, spanning partial mediation, drift, outcome sparsity, and enrollment-time labeling, with up to $R = 2{,}000$ replications over more than 500 method-by-scenario cells. We find a sharp robustness--efficiency tradeoff: the surrogate index delivers large efficiency gains when surrogacy holds but its coverage collapses under violations, while PPI-family methods stay asymptotically valid under random labeling at smaller gains. We give the finite-sample variance of PPI++ in the all-units parameterization for a fixed predictor, a joint asymptotic distribution for the two estimators under cross-fitting, and a Hausman-type estimator-disagreement diagnostic, then quantify the detection-damage gap: in the partial-mediation design, where we locate both edges, a band of violations destroys surrogate-index coverage yet is too small to detect on most datasets. On the 64{,}000-customer Hillstrom experiment the diagnostic rarely flags a violation that biases the surrogate index, and a Cauchy-kernel hybrid of the two estimators inherits 10.7\% relative bias; on the 14-million-user Criteo experiment the violation is detected, and subsampling traces detection turning on with scale as damage persists. PPI++ has limits: at a rare-conversion $n = 30{,}000$ Criteo subsample its empirical coverage is 85.0\%. We recommend prespecifying PPI++ with the exact variance as the primary analysis under random labeling and adequate labeled outcome counts, reading the diagnostic as a warning, not a certificate, and treating the surrogate index as a sensitivity analysis.

\end{abstract}


\newpage
\setcounter{tocdepth}{2}

\section{Introduction}
\label{sec:introduction}

Online controlled experiments (A/B tests) are the primary tool for causal inference at technology companies, and each launch decision rests on an estimated effect on a ``north-star'' outcome such as long-term retention, cumulative revenue, or lifetime value. These outcomes may need 30, 60, or 90 days to mature, while decisions are due sooner. Units usually enroll on a rolling basis, so at any candidate decision date the outcome has matured for the earliest enrollees and is still censored for everyone else. Choosing when to decide is therefore choosing what fraction of outcomes to observe: waiting is costly, and deciding early means acting on the matured fraction plus short-term signals observed for every unit. This paper asks when that second option yields valid inference and how much precision it gains.

Consider a retailer running an email-marketing experiment. Customers are randomized into campaign and control on a rolling basis over a four-week enrollment window, and the primary outcome is conversion, a purchase within two weeks of exposure. Funnel signals (a website visit, an add-to-cart event, a first-day click) arrive within a day or two. At a review in week 2.8, conversion windows have closed for roughly the first week of enrollees, two sevenths of the customers enrolled by that date, while every enrolled customer has funnel signals. The team can wait until week six, when every conversion window has closed, or call the experiment now from the matured conversions and everyone's funnel signals (Figure~\ref{fig:timeline}); the Hillstrom analysis of Section~\ref{sec:empirical} instantiates this structure on real data. A streaming service testing a recommendation algorithm against 90-day retention, with users onboarded over twelve weeks, faces the same choice at a month-four review: retention has matured for about a third of the sample, while first-week engagement is available for everyone.

The analysis at a review conditions on the units enrolled by the review date, and the review date fixes the matured share $\pi_L$ among them (Equation~\eqref{eq:piL-calendar}); we feature $\pi_L = 0.20$, a review at week 2.5 in the e-commerce example. Every coverage statement refers to a single pre-committed review date: calling the experiment at the first significant look of a continuously monitored test adds optional-stopping distortions, which anytime-valid variants of prediction-powered inference address \citep{kilian2025anytime}. The labeled fraction can also be a budget choice, when the outcome is expensive rather than slow (a survey, human labeling, or purchases matched from offline records), and it can then be large. We vary $\pi_L$ from 0.05 to 1.00: small values are the calendar regime, and moderate-to-large values are where a plug-in variance paired with a coefficient below its optimum undercovers and the overlap correction matters most. In both readings the outcome must be defined for every randomized unit (a conversion indicator, or revenue per user including zeros); metrics defined only for converters raise a selection problem rather than a missing-label problem (Section~\ref{sec:setup-notation}).

Two regimes look similar but rest on different assumptions. In the \emph{same-experiment delayed-label} regime (Regime~A), all units belong to one randomized experiment and the outcome has matured for a subset, as in both examples. In the \emph{historical-cohort} regime (Regime~B), the decision is due before any outcome from the current experiment has matured, so the prediction model must be trained on a prior cohort and transported, which requires comparable $(X, S)$ distributions, a stable conditional outcome model $\E[Y \mid S, X]$, and similar treatment policies and treatment effects on the surrogate. Our formal results are stated for Regime~A. DGP~4 (Section~\ref{sec:dgps}) probes one Regime~B failure, a historical index that no longer transports. Regime~B introduces additional transport assumptions, so the Regime~A guarantees do not automatically extend to it, and the prediction model should be retrained on current labels whenever possible.

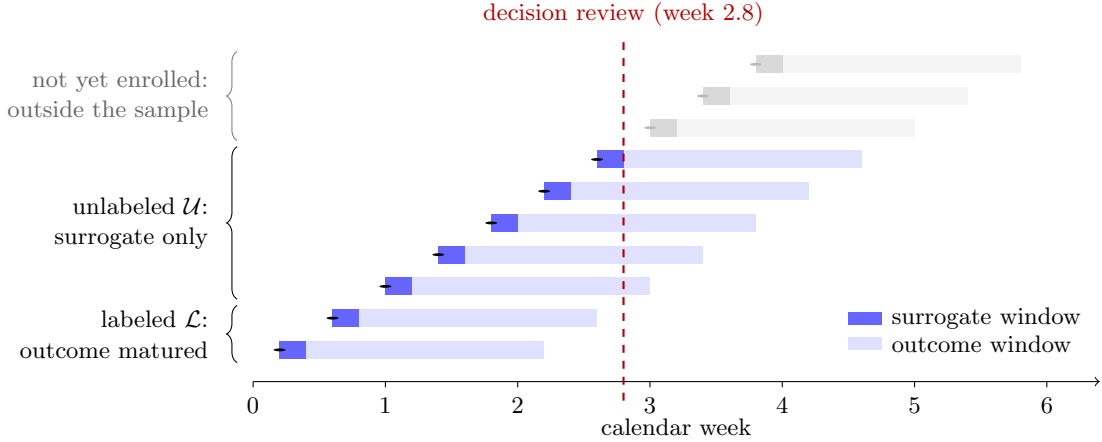
\begin{figure}[!ht]
\centering
\begin{tikzpicture}[x=1.75cm, y=0.42cm]
  \draw[->] (0,0) -- (6.4,0);
  \node[below] at (3.2,-0.85) {\footnotesize calendar week};
  \foreach \w in {0,1,2,3,4,5,6} { \draw (\w,0) -- (\w,-0.18) node[below] {\footnotesize \w}; }
  \foreach [count=\i] \e in {0.2, 0.6, 1.0, 1.4, 1.8, 2.2, 2.6} {
    \fill[blue!12] (\e,\i-0.28) rectangle (\e+2,\i+0.28);
    \fill[blue!60] (\e,\i-0.28) rectangle (\e+0.2,\i+0.28);
    \fill (\e,\i) circle (0.045);
  }
  \foreach [count=\i from 8] \e in {3.0, 3.4, 3.8} {
    \fill[gray!8] (\e,\i-0.28) rectangle (\e+2,\i+0.28);
    \fill[gray!30] (\e,\i-0.28) rectangle (\e+0.2,\i+0.28);
    \fill[gray!55] (\e,\i) circle (0.045);
  }
  \draw[dashed, thick, red!70!black] (2.8,-0.6) -- (2.8,10.9);
  \node[above, red!70!black] at (2.8,10.9) {\footnotesize decision review (week 2.8)};
  \draw[decorate,decoration={brace,mirror,amplitude=4pt}] (-0.12,2.4) -- (-0.12,0.6) node[midway,left=7pt,align=right] {\footnotesize labeled $\calL$:\\[-2pt] \footnotesize outcome matured};
  \draw[decorate,decoration={brace,mirror,amplitude=4pt}] (-0.12,7.4) -- (-0.12,2.6) node[midway,left=7pt,align=right] {\footnotesize unlabeled $\calU$:\\[-2pt] \footnotesize surrogate only};
  \draw[decorate,decoration={brace,mirror,amplitude=4pt},gray] (-0.12,10.4) -- (-0.12,7.6) node[midway,left=7pt,align=right,text=gray!80!black] {\footnotesize not yet enrolled:\\[-2pt] \footnotesize outside the sample};
  \fill[blue!60] (4.5,1.8) rectangle (4.75,2.2); \node[right] at (4.75,2.0) {\footnotesize surrogate window};
  \fill[blue!12] (4.5,1.0) rectangle (4.75,1.4); \node[right] at (4.75,1.2) {\footnotesize outcome window};
\end{tikzpicture}
\caption{The timing problem in the e-commerce example. Each row is a cohort of customers, ordered by rolling enrollment over four weeks (dots mark exposure). Funnel surrogates (dark) arrive within days; the conversion window (light) takes two weeks. At the week-2.8 review seven of the ten cohorts have enrolled, every enrolled cohort's surrogate window has closed, and the earliest two have matured; those two, two sevenths of the enrolled cohorts, form the labeled set $\calL$. Enrolled but unmatured cohorts contribute surrogates only, and the three later cohorts (gray) have not enrolled and are outside the review-time sample. Regime~A fits the prediction model on $\calL$; Regime~B arises when the review precedes every outcome window, so the model must come from a prior cohort.}
\label{fig:timeline}
\end{figure}

This timing problem has motivated a substantial and increasingly interconnected literature on \textit{surrogate metrics}: short-term proxy outcomes that correlate with the long-term target. Three traditions bear on this question, though recent work has shown they are converging.

The \textit{statistical surrogacy literature} includes the clinical-trial framework of \citet{prentice1989surrogate}; subsequent econometric work develops surrogate-index methods for estimating long-term treatment effects. \citet{athey2019surrogate} introduced the surrogate index for imputing long-term outcomes from short-term proxies, and \citet{chen2023semiparametric} develop the corresponding semiparametric efficiency theory. Under the conditional-mean invariance $\E[Y |T, S, X] = \E[Y |S, X]$ of Assumption~\ref{asm:surrogacy}, the index uses the full sample rather than the labeled subset and gains efficiency accordingly. When treatment and long-term outcomes are not jointly observed, the conditional-mean restriction cannot be checked directly in the available data. Transporting the outcome model to a new experiment requires additional assumptions. The invariance fails whenever treatment affects the outcome through channels that bypass the surrogate, which in practice is common.

\textit{Prediction-powered inference} (PPI, \citealp{angelopoulos2023prediction}; PPI++, \citealp{angelopoulos2023ppipp}; recalibrated PPI, \citealp{ji2025predictions}) takes a different approach: rather than trusting the prediction model, it uses labeled data to correct for whatever systematic error the model introduces. The model-agnostic guarantee is both the strength and limitation of PPI: at the population-optimal tuning parameter its asymptotic variance is no larger than that of the labeled-only baseline, but the gain can be negligible in finite samples, and efficiency is bounded by prediction quality and conservative tuning. \citet{mozer2026ppi} shows that PPI++ is structurally equivalent to the classical difference estimator \citep{cassel1976results}, connecting prediction-powered inference to decades of survey sampling theory. \citet{ji2025predictions} formally unify the surrogate index and PPI as endpoints of a bias--variance spectrum, establishing the theoretical foundation that our simulations characterize empirically.

A third strand from the \textit{tech-industry experimentation literature} has produced practical tools for surrogate selection and combination. \citet{tripuraneni2024choosing} construct composite proxies from historical experiments, calibrating weights across multiple surrogate metrics. \citet{richardson2025pareto} develop Pareto-optimal proxy metrics. \citet{duan2021online} provide practical guidelines for surrogate metric use at scale. Industry methods construct proxies using historical experiments and explicit optimization objectives. Their targets and assumptions differ from the same-experiment treatment-effect estimand studied here.

These literatures are converging theoretically, but each developed its tools in relative isolation: the surrogacy literature offers efficiency under strong assumptions, PPI robustness with limited efficiency, and the industry literature proxy-construction tools aimed at different targets. A systematic finite-sample comparison of the estimators under the failure modes that arise in practice (partial mediation, heterogeneous surrogate quality, temporal drift, outcome sparsity, nonlinear relationships, and antagonistic surrogates) has not, to our knowledge, been undertaken.

Pre-experiment variance reduction such as CUPED \citep{deng2013improving} and its machine-learning extensions \citep{poyarkov2016boosted} uses \textit{pre-treatment} covariates and is complementary to the \textit{post-treatment} surrogate methods studied here; surrogate methods add value beyond a CUPED-adjusted baseline (Section~\ref{sec:results-valid}). The contributions, in decreasing order of what we believe is new, are a benchmark, a quantification of when violations are statistically invisible, and practitioner tools that apply established machinery to this setting:
\begin{enumerate}[leftmargin=*,itemsep=2pt]
  \item \textit{A systematic finite-sample benchmark.} Across seven estimators, eleven data-generating processes (DGPs), and 514 method-by-scenario cells in 93 simulation scenarios (primary configurations of Methods 0 to 5; the ablations of Appendix~\ref{app:ablations} are not counted), with up to $R = 2{,}000$ replications per cell, the robustness--efficiency tradeoff is sharp. On the valid-surrogate DGPs~1 and~5, surrogate-based methods deliver large efficiency gains at small labeled fractions but lose coverage under moderate violations, while PPI-family methods stay valid under MCAR labeling at modest gains that shrink as the labeled fraction grows (Section~\ref{sec:results-valid}). The tradeoff persists across the prediction model (OLS and gradient-boosted trees, GBT), the number of surrogates, and mild nonlinearity. The PPI side needs random labeling and adequate labeled outcome counts in each arm: it fails when the treatment effect drifts over the enrollment window (DGP~11, Section~\ref{sec:results-enrollment}) and when labeled conversions are rare (Section~\ref{sec:criteo}).
  \item \textit{The detection-damage gap, located.} Specification pretests have low power against exactly the local violations that break inference \citep{leeb2005model, guggenberger2010hausman}; practice needs the location of that blind spot. At $\pi_L = 0.20$ we calibrate the band of violations that destroys surrogate-index coverage yet escapes an estimator-disagreement diagnostic on most datasets: $\rho \in [0.09, 0.37]$ at $n = 10{,}000$ and $[0.03, 0.16]$ at $n = 100{,}000$ (Figure~\ref{fig:detection-damage}). Measured in the direct effect $\delta$, both edges shrink between the two sample sizes by ratios consistent with root-$n$ scaling, the damage edge less precisely located, so scale relocates the band rather than closing it; we claim the band for DGP~2, where we locate both edges, not universally (Section~\ref{sec:detection-damage}). Two industrial funnels show its two sides by analogy, not as points on the DGP~2 boundary. On Hillstrom (64{,}000 customers) the surrogate index carries 24.9\% bias and covers 7.2\% while the diagnostic rejects on 11\% of splits; on Criteo (13.98 million users) the same within-funnel violation is detected, and subsampling traces detection turning on across two orders of magnitude of $n$ while the damage persists. A Criteo-calibrated semi-synthetic testbed, a multi-surrogate extension, and the LaLonde negative control complete the evidence (Section~\ref{sec:empirical}).
  \item \textit{Practitioner tooling from established machinery.} (a)~The exact finite-sample variance of PPI++ in the all-units parameterization $\hat{\tau}_Y^{\calL} + \lambda(\hat{\tau}_{\hat{Y}}^{\mathrm{all}} - \hat{\tau}_{\hat{Y}}^{\calL})$ at a fixed predictor and coefficient: the difference-estimator variance in that notation, which keeps the overlap covariance an independent-samples plug-in drops. This is bookkeeping for a reparameterization, not a defect in \citet{angelopoulos2023ppipp}, whose labeled-unlabeled form assigns the same variance (Section~\ref{sec:corrected-variance}). (b)~A joint asymptotic distribution for the surrogate index and PPI++ when the index is the learned linear one and every unit is scored out of fold (Proposition~\ref{prop:joint-si-ppi}). It establishes first-order validity of the implemented OLS procedure, carries the first-stage uncertainty a fixed-predictor calculation omits, and supplies the standard error of the SI--PPI++ estimator-disagreement diagnostic, a \citet{hausman1978specification}-type test strictly weaker than a test of surrogacy (Sections~\ref{sec:properties} and~\ref{sec:surrogacy-test}). (c)~A Cauchy-kernel hybrid, $w = c/(c + T_n^2)$, in the frequentist model-averaging family of \citet{hjort2003frequentist}, whose local-asymptotic risk (Proposition~\ref{prop:limit-risk}) disciplines the tuning constant (Table~\ref{tab:limit-minimax-c}). No adaptive weighting we study escapes the gap, and on Hillstrom the hybrid inherits 10.7\% relative bias, so it is exploratory rather than a replacement for PPI++ (Appendix~\ref{app:hybrid}).
  \item \textit{A prespecified primary analysis.} Under MCAR labeling with adequate labeled outcome counts in each arm, the evidence supports PPI++ with the exact variance as the default primary analysis when surrogacy is uncertain, with the diagnostic as a warning rather than a certification of the surrogate-index interval (Section~\ref{sec:recommendations}).
\end{enumerate}

\paragraph{Organization.} Section~\ref{sec:methods} sets out the estimators and the formal results, Section~\ref{sec:simulation} the eleven DGPs (Table~\ref{tab:dgps}) and the implementation of each design family (Table~\ref{tab:implementation}), Section~\ref{sec:results} the simulation evidence, and Section~\ref{sec:empirical} the public-data illustrations. Section~\ref{sec:recommendations} gives the workflow the results support and Section~\ref{sec:discussion} the limitations. The appendices hold the hybrid estimator, the proofs, the diagnostic's power grid, the Cross-PPI comparison, and the sensitivity analyses, ablations, and per-DGP results behind the main-text claims.

\section{Setup and Methods}
\label{sec:methods}

\subsection{Notation, Estimands, and Sampling Regime}
\label{sec:setup-notation}

\paragraph{Sampling regime (Regime~A).} All $n$ units belong to a single randomized experiment with $T_i \iid \text{Bernoulli}(p)$ for a known $p \in (0,1)$; the surrogate $S_i$ is observed for all units, and the primary outcome $Y_i$ for a subset chosen by a missing-completely-at-random (MCAR) labeling mechanism (Assumption~\ref{asm:mcar}). In this regime PPI++'s correction term has expectation zero at a fixed predictor, and the surrogate index, PPI++, and the hybrid share a common probability space. Missing-at-random (MAR) labeling does not by itself deliver that centering for the unweighted correction; the MAR variant of DGP~5 stresses it. The simulations set $p = 0.5$; the public-data analyses do not, and no formal result requires it.

\paragraph{Units, treatment, and outcomes.} For each unit $i = 1, \ldots, n$ we observe a pre-treatment covariate vector $X_i \in \R^p$, a surrogate $S_i \in \R$ (for all units), and a primary outcome $Y_i \in \R$ (for a labeled subset): in the e-commerce example a website visit and conversion, in the streaming example first-week engagement and 90-day retention. The main analysis takes $S_i$ scalar; DGP~7 and the multi-surrogate testbed of Section~\ref{sec:semisynthetic} study several surrogates.

The $n$ units are partitioned into a \textit{labeled set} $\calL$ with $|\calL| = n_L$ units where both $(S_i, Y_i)$ are observed, and an \textit{unlabeled set} $\calU$ with $|\calU| = n_U = n - n_L$ units where only $S_i$ is observed. We write $n_{L,t} = |\{i \in \calL : T_i = t\}|$ for the number of labeled units in treatment arm $t \in \{0, 1\}$. The labeled fraction $\pi_L = n_L / n$ is a key design parameter, where $n$ is the number of units enrolled by the review date. Under the calendar interpretation of Section~\ref{sec:introduction}, $\pi_L$ is pinned down by the review date: in the e-commerce example, with uniform enrollment over weeks 0 to 4 and a two-week outcome window, a review at week $t$ yields the matured share among units enrolled by the review,
\begin{equation}
  \pi_L(t) = \frac{t-2}{t} \;\; (2 \leq t \leq 4), \qquad \pi_L(t) = \frac{t-2}{4} \;\; (4 \leq t \leq 6),
  \label{eq:piL-calendar}
\end{equation}
the second branch applying once enrollment has ended at week 4, so $\pi_L = 0.20$ corresponds to a review at week 2.5 and $\pi_L = 1$ to week 6. Under the label-budget interpretation, $\pi_L$ is a cost choice.

Under the Neyman--Rubin potential outcomes framework, each unit has potential outcomes $(S_i(0), S_i(1), Y_i(0), Y_i(1))$, and we assume SUTVA (stable unit treatment value assumption: no interference between units and no hidden versions of treatment). The primary estimand is the average treatment effect on the primary outcome:
\begin{equation}
  \tau = \E[Y_i(1) - Y_i(0)].
  \label{eq:ate}
\end{equation}
We also define the surrogate ATE $\tau_S = \E[S_i(1) - S_i(0)]$. A prediction model $\hat{f}: \R^{1+p} \to \R$, estimated from the labeled set or auxiliary data, maps surrogate values and covariates to predicted primary outcomes: $\hat{Y}_i = \hat{f}(S_i, X_i)$.

We make two formal assumptions. The first governs the labeling mechanism; the second defines the surrogacy condition that some methods require.

\begin{assumption}[MCAR labeling]
\label{asm:mcar}
Membership in $\calL$ is independent of $(T_i, S_i, X_i, Y_i)$, so outcomes are missing completely at random. The labeled set is drawn either by independent $\mathrm{Bernoulli}(\pi_L)$ indicators or as a uniformly random subset of fixed size $\lfloor \pi_L n \rfloor$, which is how our simulations draw it; every result below holds under either scheme. The covariates enter the statement because the prediction model uses them. We relax this in DGP~5 (Section~\ref{sec:dgps}).
\end{assumption}

Under the calendar interpretation, where the labeled units are the earliest enrollees, Assumption~\ref{asm:mcar} needs scrutiny. Randomization makes enrollment time independent of $T_i$, but MCAR also requires enrollment time to be unrelated to $(X_i, S_i, Y_i)$: a stable user mix and a stable outcome process across the enrollment window. Seasonality, novelty effects, and drift in the enrolling population break this. DGP~11 labels the earliest-enrolled cohort under dialed drift (Section~\ref{sec:results-enrollment}), DGP~5's MAR variant stresses labeling through the surrogate, and DGP~4 stresses a surrogate--outcome relationship learned on a historical cohort. The public-data analyses of Section~\ref{sec:empirical} mask labels at random, so MCAR holds there by construction.

A second scope boundary separates three kinds of missingness. A well-defined outcome may be unobserved because its window has not closed or its measurement was not funded; that is the setting of Assumption~\ref{asm:mcar}. A well-defined outcome may be observed with probability depending on post-treatment observables, as in DGP~5's MAR variant, where labeled-only and PPI++ lose their design-based guarantee while AIPW with a correct observation model survives: across $q$, relative bias runs from $-4.7\%$ to $-8.5\%$ for labeled-only, $-0.1\%$ to $-2.1\%$ for PPI++, and $-1.6\%$ to $+0.7\%$ for AIPW. Finally, some outcomes exist only for a treatment-dependent subpopulation: average order value or post-purchase satisfaction exist only for customers who convert, and conversion is what the campaign moves. That is not a missing-label problem. The all-users contrast is undefined when $Y_i(t)$ does not exist for non-converters, comparing converters across arms conditions on a post-treatment variable, and reweighting on pre-treatment covariates would need an assumption at least as strong as principal ignorability. Principal stratification \citep{frangakis2002principal}, trimming bounds \citep{lee2009training}, or, most practically, an outcome defined for every randomized unit (a conversion indicator; revenue per user including zeros) are the remedies. Every outcome in this paper is of that per-randomized-unit form.

\begin{assumption}[Statistical surrogacy: conditional-mean invariance across treatment states]
\label{asm:surrogacy}
The conditional mean of the primary outcome given the surrogate and covariates is invariant across treatment states. Formally, there is a function $m(\cdot, \cdot)$ such that
\begin{equation}
  \E[Y_i(1) |S_i(1) = s, X_i = x] = \E[Y_i(0) |S_i(0) = s, X_i = x] = m(s, x)
  \label{eq:surrogacy-pot}
\end{equation}
for all $(s, x)$ in the support of the relevant conditional distributions, together with the overlap condition that this support is shared across treatment arms. Equivalently, in the observed data, $\E[Y_i |T_i, S_i, X_i] = \E[Y_i |S_i, X_i]$.
\end{assumption}

Assumption~\ref{asm:surrogacy} is a statistical condition on observed conditional means, and causal full mediation does not by itself deliver it. Full mediation, $Y_i(t,s) = Y_i(s)$, implies Assumption~\ref{asm:surrogacy} only without surrogate--outcome confounding. With a latent $U$ independent of $T$, let $S = T + U$ and $Y = S + U$: there is no direct effect of $T$ on $Y$ holding $S$ fixed, yet $U = S - T$ is recoverable from the conditioning set, so $\E[Y \mid S, T] = 2S - T$ and invariance fails. Identification from surrogate data therefore needs either full mediation with no unmeasured mediator--outcome confounding \citep{frangakis2002principal, vanderweele2013surrogate} or the statistical condition imposed directly. All our DGPs draw $\varepsilon_{S,i}$ and $\varepsilon_{Y,i}$ independently, so the two coincide in the simulations; in the field they can come apart, and surrogacy means the statistical condition throughout.

Assumption~\ref{asm:surrogacy} identifies the surrogate index (Method~2) when $\hat{f}$ is consistent for $m(s, x)$. When it fails (DGPs 2 and 3), the surrogate index incurs the bias of Proposition~\ref{prop:si-bias}; in DGP~4 it holds within the experiment, and the surrogate index is biased because its historical fit does not transport. PPI++ (Method~3) and AIPW (Method~5) rest only on Assumption~\ref{asm:mcar} plus regularity.

For any variable $V_i$ and index set $\mathcal{A}$, we write $\hat{\tau}_V^{\mathcal{A}} = \bar{V}_1^{\mathcal{A}} - \bar{V}_0^{\mathcal{A}}$ for the difference-in-means estimator, where $\bar{V}_t^{\mathcal{A}}$ averages over units in $\mathcal{A}$ with $T_i = t$. When $\mathcal{A}$ is the full sample, we drop the superscript.

\subsection{Seven Estimators}
\label{sec:eightmethods}

We compare seven estimators, numbered 0 to 6 in Table~\ref{tab:methods-summary}. Methods~0 to~5 produce a point estimate $\hat{\tau}$, a variance estimate $\hat{V}$, and a 95\% Wald interval $\hat{\tau} \pm 1.96\sqrt{\hat{V}}$; the hybrid (Method~6) has a simulation-calibrated interval (Appendix~\ref{app:hybrid}). Tuning and variance variants are ablation configurations, not further methods (Appendix~\ref{app:ablations}).

\paragraph{Prediction protocol: all-units cross-fitting.} Every method that uses unit-level predictions draws them from one protocol, so those estimators differ only in how they combine the same $\hat{Y}$ vector. Partition the $n$ units at random into $K = 5$ folds, write $k(i)$ for unit $i$'s fold and $\mathcal{I}_k$ for the units in fold $k$, and let $g(s,x) = (1, s, s^2, x)'$ be the design vector of the prediction model, $g_i = g(S_i, X_i)$. For each fold, fit the model by OLS on the labeled units outside that fold,
\begin{equation}
  \hat\beta_{(-k)} = \Bigl(\textstyle\sum_{j \in \calL \setminus \mathcal{I}_k} g_j g_j'\Bigr)^{-1} \textstyle\sum_{j \in \calL \setminus \mathcal{I}_k} g_j Y_j,
  \label{eq:fold-beta}
\end{equation}
and score every unit of that fold, labeled or not, from that fit: $\hat{Y}_i = g_i' \hat\beta_{(-k(i))}$. The surrogate index, PPI++, and AIPW consume this vector. The exceptions are the DGP~4 surrogate index, fit once on external historical units, and the composite proxy (Method~4), which uses no $\hat{Y}$ and rescales the current surrogate effect by weights fit on past experiments (Table~\ref{tab:implementation}). When the surrogate is binary, $s^2 = s$ and the design is exactly singular, as for the funnel surrogates of Section~\ref{sec:empirical}. The fold fit then solves $(\sum g_j g_j' + 10^{-10} I)^{-1} \sum g_j Y_j$, which leaves the fitted values unchanged to $10^{-10}$, and the sandwich of Proposition~\ref{prop:joint-si-ppi} inverts $Q$ on its column space, discarding eigenvalues below a relative floor of $10^{-12}$. Fitted values, estimators, and influence functions do not depend on which spanning subset of the design is retained; Section~\ref{sec:results-gbt} reports the floor's effect where the first-stage term is largest.
A fold whose training set holds fewer than $q + 1$ labeled units, with $q$ the number of design columns, raises an error rather than falling back to a full-sample fit that would put a unit's own outcome into its own prediction; the smallest training set here, LaLonde at $\pi_L = 0.10$, has about 45 labeled units per fold against $q \approx 11$. Two properties follow. No unit's own outcome enters its own prediction, so conditional on a fold's training fit its predictions are those of a fixed predictor (Proposition~\ref{prop:ppi-validity}). And labeled and unlabeled units within a fold share one coefficient, so the surrogate index is exactly the foldwise sum of Equation~\eqref{eq:si-fold-average}, whose first-order representation is $\hat{d}'\hat\beta$ with $\hat\beta$ the full labeled fit, and the pair $(\hat\tau_{\mathrm{SI}}, \hat\tau_{\mathrm{PPI++}})$ has the first-order expansion of one stacked system of estimating equations (Proposition~\ref{prop:joint-si-ppi}). The GBT comparison of Section~\ref{sec:results-gbt} replaces the OLS step inside the same protocol. The partition is unstratified; a stratified variant is implemented but used by no reported result.

\begin{table}[!ht]
\centering
\caption{Summary of estimation methods.}
\label{tab:methods-summary}
\footnotesize
\setlength{\tabcolsep}{4pt}
\begin{tabular}{@{}l >{\raggedright\arraybackslash}p{2.2cm} >{\raggedright\arraybackslash}p{2.2cm} c >{\raggedright\arraybackslash}p{2.8cm}@{}}
\toprule
\textbf{Method} & \textbf{Assumption} & \textbf{Guarantee} & \textbf{Best RE} & \textbf{When to Use} \\
\midrule
0: Labeled-Only & MCAR & Unbiased, valid CI & $1.0\times$ & Always valid baseline \\
1: Naive Surrogate & Surrogacy + linear $Y$--$S$ link & Correct sign & High (biased) & Decision-making only \\
2: Surrogate Index & Surrogacy + correct $f$ & Consistent if valid & $2.9$--$5.5\times$ & Well-validated surrogate \\
3: PPI++ & MCAR & Unbiased at fixed $f$, $\lambda$; asymptotically valid CI & $1.3$--$1.5\times$ & Uncertain surrogate \\
4: Composite Proxy & Surrogacy + $K_{\text{hist}}$ past experiments & Calibrated weights & $3.2$--$6.2\times$ & Multiple surrogates, with a library of past experiments \\
5: AIPW & MCAR (or correct propensity) & Doubly robust & $1.4$--$1.5\times$ & MAR; robustness \\
6: Hybrid (Appendix~\ref{app:hybrid}) & MCAR & Adaptive heuristic & $2.0$--$2.2\times$ & Exploratory shrinkage; use only with diagnostic + domain caveats \\
\bottomrule
\end{tabular}
\smallskip

{\scriptsize \noindent Notes: RE (relative efficiency) is the RMSE of labeled-only over the RMSE of the method (Section~\ref{sec:eval-metrics}); RE $> 1$ improves on the baseline. ``Best RE'' covers the valid-surrogate DGPs~1 and~5 from $\pi_L = 0.20$ to $\pi_L = 0.05$; gains decline as $\pi_L$ grows, to $1.8\times$ (surrogate index) and $1.19\times$ (PPI++) at $\pi_L = 0.50$ (Table~\ref{tab:dgp1}). The guarantee column is asymptotic for every method that learns its prediction function and relies on the normal approximation for the labeled arm means, which fails at the rare-event rates of the smallest Criteo subsamples (Section~\ref{sec:criteo}). The composite proxy's past-experiment library is described under Method~4 and in Table~\ref{tab:implementation}.}
\end{table}

\paragraph{Method 0: Labeled-Only (Baseline).}
The standard Neyman estimator applied to the labeled subset:
\begin{equation}
  \hat{\tau}^{(0)} = \hat{\tau}_Y^{\calL} = \bar{Y}_1^{\calL} - \bar{Y}_0^{\calL}.
  \label{eq:labeled-only}
\end{equation}
This is unbiased for $\tau$ under Assumption~\ref{asm:mcar} with variance $O(1/n_L)$. It ignores the surrogate entirely and is the benchmark: a method is useful only if it improves on Method~0.

\paragraph{Method 1: Naive Surrogate.}
Estimate the treatment effect on the surrogate from all $n$ units, $\hat{\tau}_S = \bar{S}_1 - \bar{S}_0$, and put it on the outcome scale with the simple (not partial) slope $\hat{\beta}_{YS}$ of $Y$ on $S$ in the labeled data:
\begin{equation}
  \hat{\tau}^{(1)} = \hat{\beta}_{YS}\,\hat{\tau}_S .
  \label{eq:naive}
\end{equation}
This is the reported naive-surrogate estimate, and every NS bias, RMSE, coverage, and decision entry is measured against $\tau$. Its variance is $\hat{\beta}_{YS}^2\,\hat{V}(\hat{\tau}_S)$, with $\hat{V}(\hat{\tau}_S)$ the Neyman variance of the surrogate difference in means; the interval treats $\hat{\beta}_{YS}$ as fixed and inherits the omitted-variable bias of the simple slope, and its below-nominal coverage reflects both. Under surrogacy with a linear link and $\beta_{YS} > 0$, $\text{sign}(\tau_S) = \text{sign}(\tau)$, so we evaluate Method~1 mainly on the correct decision rate (CDR) and report its coverage to show what the common shortcut costs.

\paragraph{Method 2: Surrogate Index \citep{athey2019surrogate}.}
Estimate $m(s, x) = \E[Y \mid S = s, X = x]$, the function appearing in Assumption~\ref{asm:surrogacy}, by the cross-fitted OLS of the prediction protocol above; we write $\hat{f}$ for that estimate and take $f \equiv m$ throughout. Impute $\hat{Y}_i = \hat{f}(S_i, X_i)$ for all units, and estimate the ATE on imputed outcomes:
\begin{equation}
  \hat{\tau}^{(2)} = \bar{\hat{Y}}_1 - \bar{\hat{Y}}_0,
  \label{eq:surrogate-index}
\end{equation}
the difference of arm means of $\hat{Y}$ over all $n$ units. Because every unit of a fold is scored by that fold's coefficient, Equation~\eqref{eq:surrogate-index} is exactly a fold average of linear indices. Write $n_k = |\mathcal{I}_k|$, let $n_{tk}$ be the number of units of fold $k$ in arm $t$ and $\bar g_{tk}$ their mean design vector, and set
\begin{equation}
  \hat{d}_k = \frac{n}{n_k}\left(\frac{n_{1k}}{n_1}\,\bar g_{1k} - \frac{n_{0k}}{n_0}\,\bar g_{0k}\right),
  \label{eq:fold-d}
\end{equation}
the between-arm difference of mean design vectors in fold $k$ carrying the arm weights the estimator itself uses. Then
\begin{equation}
  \hat{\tau}^{(2)} \;=\; \hat{\tau}_{\mathrm{SI}} \;=\; \sum_{k=1}^{K} \frac{n_k}{n}\; \hat{d}_k'\, \hat\beta_{(-k)} ,
  \label{eq:si-fold-average}
\end{equation}
and each $\hat{d}_k$ equals $\bar g_{1k} - \bar g_{0k}$ up to $O_p(n^{-1/2})$ and converges to $d_0 = \E[g \mid T = 1] - \E[g \mid T = 0]$. The foldwise sum~\eqref{eq:si-fold-average} is the exact definition of the estimator; the single index $\hat{d}'\hat\beta$, with $\hat{d} = \bar g_1 - \bar g_0$ and $\hat\beta$ the full labeled fit, differs from it by $o_p(n^{-1/2})$ and is its first-order representation.

Under Assumption~\ref{asm:surrogacy} and correct specification of $f$, $\hat{\tau}_{\mathrm{SI}}$ is consistent for $\tau$. For a \emph{known} index $f$ the sampling variance is $\sum_t \Var\{f(S,X) \mid T = t\}/n_t$; under balanced labeling and arm-invariant variances its ratio to the labeled-only variance $\sum_t \Var(Y \mid T = t)/n_{L,t}$ is
\begin{equation}
  \pi_L\,\frac{\Var\{f(S,X) \mid T\}}{\Var(Y \mid T)},
  \label{eq:si-known-index-ratio}
\end{equation}
so the advantage over the labeled-only baseline shrinks in proportion to the labeled fraction. A learned index adds the first-stage term of Proposition~\ref{prop:joint-si-ppi}, and every surrogate-index interval on the cross-fitted OLS index comes from that proposition; Table~\ref{tab:implementation} lists the designs that use the known-index expression instead.

\paragraph{The surrogate index interval.} Because the index is learned from the same experiment, a complete interval must carry the estimation error of $\hat\beta_{(-k)}$ as well as the sampling variation of the arm means at a fixed coefficient. The intervals we report do: they use the sandwich estimator of Proposition~\ref{prop:joint-si-ppi}, whose influence function contains both terms. Two narrower constructions appear as ablations. ``SI (plug-in variance)'' is the Neyman variance of a difference in means of $\hat{Y}$ with $\hat\beta$ held at its estimate, which is what the surrogate-index literature usually reports. ``SI (delta-method first stage)'' adds to that the term $d' V_\beta d$,
\begin{equation}
  \Var_{\mathrm{delta}}(\hat{\tau}^{(2)}) = \Var_{\mathrm{plug\text{-}in}}(\hat{\tau}^{(2)}) + d'\,V_\beta\, d,
  \qquad d = \bar{g}_1 - \bar{g}_0,
  \label{eq:si-first-stage}
\end{equation}
where $\bar{g}_t$ is the mean design vector over all units in arm $t$ and $V_\beta$ the HC1 sandwich covariance of the OLS fit on the full labeled set. That term dominates the SI variance whenever the index leans on a predictor that treatment moves. Equation~\eqref{eq:si-first-stage} omits the cross term $2\beta'\Cov(\bar{g}_1 - \bar{g}_0,\, \hat{\beta})\, d$, of the same $O(1/n)$ order, which the sandwich includes. Section~\ref{sec:surrogacy-test} checks the sandwich's calibration, and Section~\ref{sec:results-gbt} with Table~\ref{tab:si-variance-ablation} reports what each narrower construction costs.

The surrogate index fully trusts the prediction model: under Regime~A, $\hat{f}$ is fit on the labeled subset of the same experiment (the early enrollees whose outcome has matured) and applied to the entire cohort, so the estimate is maximally efficient if the model is right and inherits its bias if not.

\paragraph{Method 3: PPI++ \citep{angelopoulos2023prediction, angelopoulos2023ppipp}.}
Use the cross-fitted predictions for all units, then correct for systematic prediction error using the labeled set:
\begin{equation}
  \hat{\tau}^{(3)} = \hat{\tau}_Y^{\calL} + \hat{\lambda}\left(\hat{\tau}_{\hat{Y}}^{\text{all}} - \hat{\tau}_{\hat{Y}}^{\calL}\right),
  \label{eq:ppi}
\end{equation}
where $\hat{\tau}_{\hat{Y}}^{\text{all}}$ and $\hat{\tau}_{\hat{Y}}^{\calL}$ are the prediction-based ATEs over all $n$ units and over the labeled set. The coefficient $\hat{\lambda} \in \R$ is common to the two arms and is not clipped. At $\hat{\lambda} = 0$ the estimator is labeled-only; at $\hat{\lambda} = 1$ it is the unit-weight PPI estimator $\hat{\tau}_{\hat{Y}}^{\text{all}} + (\hat{\tau}_Y^{\calL} - \hat{\tau}_{\hat{Y}}^{\calL})$, which still carries the labeled residual correction, so no value of $\lambda$ reduces it to the surrogate index. The coefficient minimizes the exact variance of Proposition~\ref{prop:corrected-variance}, a quadratic in $\lambda$ with leading coefficient $\sum_t a_t\,\sigma^2_{\hat{Y},t}$:
\begin{equation}
  \hat{\lambda} = \frac{\sum_t a_t\,\hat{\gamma}_t}{\sum_t a_t\,\hat{\sigma}_{\hat{Y},t}^2},
  \qquad a_t = \frac{1}{n_{L,t}} - \frac{1}{n_t},
  \label{eq:lambda-exact}
\end{equation}
where $\hat{\gamma}_t = \widehat{\Cov}(Y_i, \hat{Y}_i | T_i = t)$ is estimated from labeled data and $\hat{\sigma}_{\hat{Y},t}^2 = \widehat{\Var}(\hat{Y}_i | T_i = t)$ from all data. The quadratic is strictly convex when some arm has $\sigma^2_{\hat{Y},t} > 0$ and $n_{L,t} < n_t$. At $\pi_L = 1$ every $a_t$ is zero, Equation~\eqref{eq:lambda-exact} is $0/0$, and the rule is defined as $\hat{\lambda} = 0$, the labeled-only estimator. The implementation returns labeled-only whenever the unlabeled set holds at most the larger of ten units and 5\% of the sample; that threshold covers every $\pi_L \geq 0.95$, where it departs from standard PPI++, but no reported row is affected because the labeled fractions we run jump from $0.80$ to $1.00$. In the balanced, arm-homogeneous case Equation~\eqref{eq:lambda-exact} reduces to $\hat{\gamma}/\hat{\sigma}_{\hat{Y}}^2$, which does not depend on $\pi_L$. Intervals use the exact variance, so the tuning rule and the interval come from one objective. Two guarantees are kept apart. At a fixed predictor and a fixed $\lambda$ the correction term has expectation zero under MCAR, so the estimator is exactly unbiased whatever the prediction model's quality (Proposition~\ref{prop:ppi-validity}). With the learned index and the estimated $\hat\lambda$ the guarantee is first-order: the estimator is asymptotically normal and centered at $\tau$ (Proposition~\ref{prop:joint-si-ppi}).

The unclipped coefficient makes Equation~\eqref{eq:ppi} the linear-recalibration form of the recalibrated prediction-powered estimator of \citet{ji2025predictions}, whose general version, not implemented here, learns a recalibration function of the predictions. The fitted GREG estimator of survey sampling \citep{sarndal1992model} has the same form with a different coefficient, the labeled-sample slope of $Y$ on $\hat{Y}$. Neither is a separate method here.

\paragraph{Relation to published PPI++, and what is ours.} Equation~\eqref{eq:ppi} is the \emph{all-units} parameterization. \citet{angelopoulos2023ppipp} write PPI++ for mean estimation in the labeled-unlabeled form $\hat{\tau}_Y^{\calL} + \lambda'(\hat{\tau}_{\hat{Y}}^{\calU} - \hat{\tau}_{\hat{Y}}^{\calL})$. The identity $\bar{\hat{Y}}_t - \bar{\hat{Y}}_{L,t} = (n_{U,t}/n_t)(\bar{\hat{Y}}_{U,t} - \bar{\hat{Y}}_{L,t})$ holds arm by arm, so the change of variables is
\begin{equation}
  \lambda'_t = \lambda\,\frac{n_{U,t}}{n_t},
  \label{eq:lambda-reparam}
\end{equation}
one coefficient per arm, and it collapses to a single $\lambda' = \lambda\, n_U/n$ only when the two arms share the same unlabeled share. Tuning constant and variance must be transformed together. With a common, unclipped $\lambda$ chosen by Equation~\eqref{eq:lambda-exact}, Equation~\eqref{eq:ppi} is standard PPI++ applied to the scalar ATE contrast, and the exact variance of Proposition~\ref{prop:corrected-variance} is the variance the labeled-unlabeled form assigns after the reparameterization~\eqref{eq:lambda-reparam}. We claim no equivalence beyond that case. Three variants depart from it and are ours rather than theirs. Clipping $\lambda$ to $[0,1]$ is an additional restriction, since published PPI++ mean estimation leaves the coefficient unrestricted, and under the reparameterization the clip must be transformed too. Tuning one coefficient per arm rather than one for the contrast is a different estimator. Minimizing the plug-in rather than the exact variance objective is a third,
\begin{equation}
  \hat{\lambda}_{\mathrm{plug\text{-}in}} = \frac{\sum_{t} \hat{\gamma}_t / n_{L,t}}{\sum_{t} \hat{\sigma}_{\hat{Y},t}^2 (1/n_t + 1/n_{L,t})},
  \label{eq:lambda-opt}
\end{equation}
which in the balanced case is Equation~\eqref{eq:lambda-exact} shrunk by $1/(1 + \pi_L)$ and costs the fraction $(\pi_L/(1+\pi_L))^2$ of the variance gain PPI++ achieves at its own optimum, about 3\% of that gain at $\pi_L = 0.20$ and 20\% at $\pi_L = 0.80$; the shrinkage is toward the labeled-only endpoint, so it is conservative. All three, together with the plug-in variance estimator, are reported as ablations under the names ``PPI++ (clipped)'', ``PPI++ (per-arm $\lambda$)'', ``PPI++ (plug-in $\lambda$)'', and ``PPI++ (plug-in variance)''.

Unqualified ``PPI++'' means the common, unclipped, exactly tuned rule of Equation~\eqref{eq:lambda-exact} with the exact variance of Proposition~\ref{prop:corrected-variance}: the estimator in every headline table and figure and the one Section~\ref{sec:recommendations} recommends. The qualified names refer to the ablations only.

\paragraph{Method 4: Composite Proxy \citep{tripuraneni2024choosing}.}
Rescale the current surrogate effects by weights fit on $K_{\text{hist}}$ past experiments:
\begin{equation}
  \hat{\tau}^{(4)} =\hat{\mathbf{w}}' \hat{\boldsymbol{\tau}}_S, \quad \text{where} \quad \hat{\mathbf{w}} = \left(\sum_{k=1}^{K_{\text{hist}}} \hat{\boldsymbol{\tau}}_S^{(k)} \hat{\boldsymbol{\tau}}_S^{(k)\prime}\right)^{-1} \sum_{k=1}^{K_{\text{hist}}} \hat{\boldsymbol{\tau}}_S^{(k)} \hat{\tau}_Y^{(k)},
  \label{eq:composite}
\end{equation}
where $\hat{\boldsymbol{\tau}}_S$ is the current all-units difference in surrogate means and $\hat{\boldsymbol{\tau}}_S^{(k)}$ and $\hat{\tau}_Y^{(k)}$ are the surrogate and outcome differences in means estimated in past experiment $k$. The weights are a no-intercept least-squares fit of the estimated past outcome effects on the estimated past surrogate effects, and the interval uses $\hat{\mathbf{w}}'\widehat{\Var}(\hat{\boldsymbol{\tau}}_S)\hat{\mathbf{w}}$ with $\hat{\mathbf{w}}$ held fixed. The method uses neither $\hat{Y}$ nor the current labels, and it needs data the other methods do not: a library of past experiments that measured both the surrogate and the outcome. In the simulations one library of $K_{\text{hist}} = 30$ past experiments ($50$ in the best-case DGP~7 run of Table~\ref{tab:dgp7}) is drawn once from a fixed seed and held fixed across the replications and configurations of a run. Each past experiment has $n_k$ uniform on $\{2{,}000, \ldots, 10{,}000\}$, a true effect $\tau_k \sim \mathcal{N}(0, 0.1^2)$, and surrogate effect $\tau_k/\beta_{YS}$ at the historical slope $\beta_{YS} = 0.5$ (the pre-drift slope in DGP~4); in DGP~7 a common random factor scales the three surrogate effects. The weights use each past experiment's estimated effects, not the drawn $\tau_k$, so the composite proxy is not oracle-assisted; but every past experiment follows the valid-surrogate DGP~1 design, which makes its calibration favorable relative to practice. Its coverage is therefore conditional on one realized library per run and does not average over the estimation uncertainty of the library. With no past experiments the implementation falls back to the surrogate index, as in the DGP~6 portfolio and the Hillstrom analysis, where the rows are labeled ``CP (SI fallback)'' and carry no independent evidence about the composite proxy. The method is most valuable with multiple surrogates ($J > 1$), which DGP~7 tests.

\paragraph{Method 5: AIPW (Augmented Inverse Probability Weighted).}
The doubly-robust estimator combines the prediction model with an observation-probability model. For each unit, define the augmented outcome
\begin{equation}
  \phi_i = \frac{M_i}{\pi_i}(Y_i - \hat{Y}_i) + \hat{Y}_i,
  \label{eq:aipw}
\end{equation}
where $M_i = \mathbf{1}\{i \in \calL\}$ indicates whether $Y_i$ is observed and $\pi_i = P(M_i = 1)$ is the observation propensity. The AIPW estimator is $\hat{\tau}^{(5)} = \bar{\phi}_1 - \bar{\phi}_0$ over all units in each arm. It is consistent if either the outcome model or the propensity model is correctly specified \citep{robins1994estimation, tsiatis2006semiparametric}, with one qualification here: the outcome branch requires the arm-specific regression $\E[Y \mid S, X, T = t]$, whereas our pooled $\hat{f}(S, X)$ equals it only under Assumption~\ref{asm:surrogacy}. Under a surrogacy violation AIPW's consistency in our MCAR design therefore rests on the propensity branch, trivially correct because $\pi_i = \pi_L$ is known, which is why AIPW tracks PPI++ closely in the violated-surrogate DGPs: both rely on the same design-based argument. Under MCAR, AIPW is a close relative of PPI with $\lambda = 1$; its distinct value appears in DGP~5's MAR variant, where the propensity is estimated and AIPW's relative bias stays between $-1.6\%$ and $+0.7\%$ across $q$, against $-4.7\%$ to $-8.5\%$ for labeled-only. Cross-PPI \citep{zrnic2024cross} gives nearly identical results to PPI++ under our cross-fitting protocol and is compared in Appendix~\ref{app:cross-ppi}.

\subsection{Structural Properties}
\label{sec:properties}

The following four propositions formalize the key tradeoffs among the methods and the inference that goes with them.

\begin{proposition}[Surrogate index bias under partial mediation]
\label{prop:si-bias}
Suppose the data-generating process is
\begin{align}
  S_i &= \alpha_S + \beta_{SX} X_i + \gamma_S T_i + \varepsilon_{S,i}, \label{eq:dgp-s}\\
  Y_i &= \alpha_Y + \beta_{YS} S_i + \beta_{YX} X_i + \delta \cdot T_i + \varepsilon_{Y,i}, \label{eq:dgp-y}
\end{align}
with $T_i \sim \mathrm{Bernoulli}(p)$ independent of $(X_i, \varepsilon_{S,i}, \varepsilon_{Y,i})$, the errors mutually independent, and $\sigma_{\varepsilon_S}^2 = \Var(\varepsilon_{S,i})$, so that the true ATE is $\tau = \beta_{YS} \gamma_S + \delta$. Train $\hat{f}$ as the paper's procedure does, by pooled OLS of $Y$ on $(1, S, X)$ with treatment omitted. Then the probability limit of the fitted coefficient on $S$ is
\begin{equation}
  \beta_{YS} + \delta\,\frac{\gamma_S\, p(1-p)}{\sigma_{\varepsilon_S}^2 + \gamma_S^2\, p(1-p)},
  \label{eq:pooled-slope}
\end{equation}
and the surrogate index estimator $\hat{\tau}^{(2)}$ satisfies
\begin{equation}
  \plim_{n \to \infty}\, \hat{\tau}^{(2)} = \tau - \delta\,\kappa,
  \qquad
  \kappa \;=\; \frac{\sigma_{\varepsilon_S}^2}{\sigma_{\varepsilon_S}^2 + \gamma_S^2\, p(1-p)},
  \label{eq:si-bias}
\end{equation}
so that $\mathrm{Bias}(\hat{\tau}_{\mathrm{SI}}) = -\delta\,\kappa$. If instead $\hat{f}$ is replaced by the oracle structural prediction function $m(s,x) = \beta_{YS} s + \beta_{YX} x$, the bias is exactly $-\delta$.
\end{proposition}

\begin{remark}
\label{rem:si-bias-scope}
Three scope conditions. First, the attenuation factor $\kappa$ is a property of the training procedure, not of the structural model: because $S$ is affected by treatment, a pooled regression that omits $T$ loads part of $\delta T$ onto the surrogate coefficient, and $\kappa < 1$ measures how little of $\delta$ that reallocation absorbs. Under the DGP~2 parameters used throughout ($\sigma_{\varepsilon_S}^2 = 4$, $\gamma_S = 0.3$, $p = 0.5$) the reallocation is negligible: $\kappa = 4/4.0225 = 0.9944$, so the pooled-fit bias is $-0.9944\,\delta$ and the simulations cannot distinguish it from $-\delta$. At $\rho = 0.2$, where $\delta = 0.0375$, the predicted bias is $-0.0373$ against the simulated $-0.036$ (Table~\ref{tab:dgp2}). Every statement elsewhere in the paper that the surrogate index misses the direct effect should be read with this factor attached; we write it once here and do not repeat it. Second, the implemented prediction model adds $S^2$; under the linear DGP its population coefficient is negligible and the simulated bias confirms this, but the closed form is stated for the $(1, S, X)$ projection. Third, the closed form relies on the linear specification. Under nonlinear $Y = g(S) + \delta T + \varepsilon$, the qualitative conclusion (the surrogate index misses the direct effect) holds generally, but the quantitative formula is specific to the linear case.
\end{remark}

\begin{proposition}[Unbiasedness of PPI at a fixed predictor under MCAR]
\label{prop:ppi-validity}
Let Assumption~\ref{asm:mcar} (MCAR labeling) hold, fix $\lambda \in \R$, and let $\hat{Y}_i = f(S_i, X_i)$ for a prediction function $f: \R^{1+p} \to \R$ that is fixed, or trained on data independent of the current experiment's outcomes and labels. Conditionally on the treatment vector and on labeled counts $n_{L,t} \geq 1$ in each arm, the PPI estimator
\begin{equation}
  \hat{\tau}^{\mathrm{PPI}(\lambda)} = \hat{\tau}_Y^{\calL} + \lambda \left(\hat{\tau}_{\hat{Y}}^{\mathrm{all}} - \hat{\tau}_{\hat{Y}}^{\calL}\right)
\end{equation}
satisfies $\E[\hat{\tau}^{\mathrm{PPI}(\lambda)}] = \tau$ exactly in finite samples. The same argument applies within one fold of the all-units protocol: conditionally on the fold's own training fit $\hat\beta_{(-k)}$ and on its labeled count $n_{L,tk} \geq 1$ in each arm, the fold's labeled units in arm $t$ are a simple random subset of the fold's units in arm $t$, so the fold's correction term has conditional mean zero.
\end{proposition}

The fold-level statement does not extend to the pooled correction of the implemented estimator: the labeled and all-units averages weight fold $k$ by $n_{L,tk}/n_{L,t}$ and $n_{tk}/n_t$, and the other folds' fits use fold $k$'s labels, so no single conditioning event fixes every fold's predictions at once. For the implemented OLS procedure, with the learned index and the estimated $\hat\lambda$, we claim first-order validity (Proposition~\ref{prop:joint-si-ppi}) and no exact finite-sample unbiasedness. The fixed-predictor case is the original framing of \citet{angelopoulos2023prediction} and the setting of the external DGP~4 index (Table~\ref{tab:implementation}). Training the prediction model on the labeled outcomes used in the correction without cross-fitting makes $\hat{Y}_i$ for $i \in \calL$ a function of $Y_i$ and voids the argument; cross-fitting or Cross-PPI \citep{zrnic2024cross} avoids that.

Proposition~\ref{prop:ppi-validity} fixes the prediction function. The next result does not: it gives the joint limiting distribution of the two estimators when the index is the learned linear one, which is what the surrogate-index interval, the diagnostic of Section~\ref{sec:surrogacy-test}, and the hybrid of Appendix~\ref{app:hybrid} all require.

\begin{proposition}[Joint distribution of the surrogate index and PPI++ with a learned linear index]
\label{prop:joint-si-ppi}
Let $(X_i, S_i(0), S_i(1), Y_i(0), Y_i(1), T_i)$, $i = 1, \ldots, n$, be i.i.d. Treatment is randomized: $T_i \sim \mathrm{Bernoulli}(p)$ with $p \in (0,1)$, independent of $(X_i, S_i(0), S_i(1), Y_i(0), Y_i(1))$. The observed surrogate and outcome are the potential outcomes under the assigned arm, $S_i = S_i(T_i)$ and $Y_i = Y_i(T_i)$. Labeling is MCAR (Assumption~\ref{asm:mcar}): the labeled set, with indicators $L_i$, is independent of $\{(T_i, X_i, S_i(0), S_i(1), Y_i(0), Y_i(1))\}_{i=1}^n$ and is drawn either by i.i.d.\ $\mathrm{Bernoulli}(\pi_L)$ indicators or as a uniformly random subset of fixed size $\lfloor \pi_L n \rfloor$, with $\pi_L \in (0, 1]$. Let $K$ be fixed and the folds a random partition into $K$ equal blocks of $n/K$ units, and let predictions be formed by the all-units protocol of Equation~\eqref{eq:fold-beta}. Assume $\E\|g\|^4 < \infty$, $\E Y^4 < \infty$, and that $Q = \E[g g' \mid L = 1] = \E[g g']$ is finite and positive definite on the column space of the design, with $Q^{-1}$ denoting its inverse on that space. That restriction is the eigen-truncation of Section~\ref{sec:eightmethods}, which keeps the eigenspaces of $Q$ above the relative floor, and not the deletion of columns; for a binary surrogate, where $s^2 = s$ makes the design exactly rank-deficient, the two give the same fitted values, and the variance computation performs the first. Define the labeled-population pooled OLS coefficient and its residual,
\begin{equation}
  \beta_0 = Q^{-1}\,\E[g Y], \qquad \varepsilon_i = Y_i - g_i'\beta_0,
  \label{eq:beta0-eps}
\end{equation}
where $\varepsilon$ is orthogonal to $g$ but need not have mean zero within a treatment arm; that within-arm imbalance is the mechanism of Proposition~\ref{prop:si-bias}. Let $d_0 = \E[g \mid T = 1] - \E[g \mid T = 0]$, $\tau_{\mathrm{SI}}^{0} = d_0'\beta_0$, and let $\hat\tau_{\mathrm{PPI}}$ be Equation~\eqref{eq:ppi} at a fixed $\lambda$. Write $A_i = T_i/p - (1 - T_i)/(1-p)$. Then
\begin{equation}
  \sqrt{n}\left\{\begin{pmatrix}\hat\tau_{\mathrm{SI}}\\[2pt] \hat\tau_{\mathrm{PPI}}\end{pmatrix} - \begin{pmatrix}\tau_{\mathrm{SI}}^{0}\\[2pt] \tau\end{pmatrix}\right\} \;\xrightarrow{d}\; \mathcal{N}(0, \Omega), \qquad \Omega = \E[\phi_i \phi_i'],
  \label{eq:joint-clt}
\end{equation}
where $\phi_i = (\phi_i^{\mathrm{SI}}, \phi_i^{\mathrm{PPI}})'$ with
\begin{align}
  \phi_i^{\mathrm{SI}} &= A_i\bigl(g_i'\beta_0 - \E[g'\beta_0 \mid T_i]\bigr) \;+\; \frac{L_i}{\pi_L}\, d_0'\,Q^{-1} g_i\,\varepsilon_i,
  \label{eq:phi-si}\\
  \phi_i^{\mathrm{PPI}} &= \frac{L_i}{\pi_L}\,A_i\Bigl\{(Y_i - \lambda\, g_i'\beta_0) - \E[Y - \lambda\, g'\beta_0 \mid T_i]\Bigr\} \;+\; \lambda\,A_i\bigl(g_i'\beta_0 - \E[g'\beta_0 \mid T_i]\bigr).
  \label{eq:phi-ppi}
\end{align}
\end{proposition}

\noindent The proof is in Appendix~\ref{app:proofs}. The limit is the same under either labeling scheme: every $L$-weighted summand of the expansion is centered, and the finite-population correction of sampling a fixed-size subset without replacement does not enter at first order. Four points qualify the result.

\emph{The first stage is first order for SI and second order for PPI++.} The two terms of $\phi^{\mathrm{SI}}$ are the two terms of the expansion $\hat\tau_{\mathrm{SI}} - \tau_{\mathrm{SI}}^{0} = \beta_0'(\hat d - d_0) + d_0'(\hat\beta - \beta_0) + o_p(n^{-1/2})$, so the estimation error of the index enters both $\Var(\hat\tau_{\mathrm{SI}})$ and $\Cov(\hat\tau_{\mathrm{SI}}, \hat\tau_{\mathrm{PPI}})$. Cross-fitting does not remove it: the surrogate index is not orthogonal to estimation of its own prediction function. Because $\Omega$ is the second moment of the joint influence function, it also carries the cross moment between those two terms, which the delta-method expression of Equation~\eqref{eq:si-first-stage} drops. PPI++ depends on the first stage only through $\lambda\sum_k (w^{\mathrm{all}}_k - w^{\calL}_k)'(\hat\beta_{(-k)} - \beta_0)$, where $w^{\mathrm{all}}_k$ and $w^{\calL}_k$ are fold $k$'s shares of the all-units and labeled between-arm design differences; each term is a product of two $O_p(n^{-1/2})$ factors under MCAR, so the sum is $O_p(1/n)$; this is why $\phi^{\mathrm{PPI}}$ is evaluated at $\beta_0$.

\emph{The estimand of the surrogate index is $\tau_{\mathrm{SI}}^{0}$, not $\tau$.} Equation~\eqref{eq:joint-clt} centers the first coordinate at $d_0'\beta_0$, which equals $\tau$ under Assumption~\ref{asm:surrogacy} and correct specification and differs from it by $-\delta\kappa$ in the setting of Proposition~\ref{prop:si-bias} otherwise. The proposition is a distributional statement, not an identification statement.

\emph{Estimated tuning.} The result is stated at a fixed $\lambda$. Replacing $\lambda$ by an estimate with $\hat\lambda - \lambda = O_p(n^{-1/2})$, which the exact-variance rule satisfies under the moment conditions above, adds $(\hat\lambda - \lambda)(\hat\tau_{\hat Y}^{\mathrm{all}} - \hat\tau_{\hat Y}^{\calL})$, a product of two $O_p(n^{-1/2})$ factors (consistency of $\hat\lambda$ alone would make it $o_p(n^{-1/2})$, which also suffices), so it does not affect Equation~\eqref{eq:joint-clt}; this is the first-order validity of PPI++ with the estimated coefficient. We make no claim about its finite-sample contribution.

\emph{Estimating $\Omega$.} Let $\hat\beta$ be the OLS coefficient on the full labeled set (the fold average of $\hat\beta_{(-k)}$ is asymptotically equivalent), $\hat\varepsilon_i = Y_i - g_i'\hat\beta$ on labeled units, $\hat Q = n_L^{-1}\sum_{i \in \calL} g_i g_i'$, $\hat d$ the sample analogue of $d_0$, $\hat p = n_1/n$, $\hat\pi_L = n_L/n$, and replace each conditional expectation by the corresponding sample arm mean. The plug-in sandwich is $\hat\Omega = n^{-1}\sum_i \hat\phi_i \hat\phi_i'$. Every surrogate-index interval built on the cross-fitted OLS index uses $\hat\Omega_{11}/n$ (the exceptions, DGP~4 and the gradient-boosted-tree rows, are listed in Table~\ref{tab:implementation}); the diagnostic of Section~\ref{sec:surrogacy-test} uses
\begin{equation}
  \widehat{\Var}(\hat D) = \frac{1}{n}\,(1, -1)\,\hat\Omega\,(1, -1)',
  \label{eq:var-D-sandwich}
\end{equation}
which is invariant to the order of the difference; and the hybrid of Appendix~\ref{app:hybrid} uses $\hat\Sigma = \hat\Omega/n$.

\begin{proposition}[Efficiency ordering]
\label{prop:efficiency}
Fix the prediction function, condition on the arm sizes and labeled counts, and let $V(\lambda)$ be the exact variance of $\hat{\tau}^{\mathrm{PPI}(\lambda)}$ (Proposition~\ref{prop:corrected-variance}). Evaluate PPI++ at the minimizer $\lambda^*$ of $V$ over $\R$ and the clipped ablation at its minimizer $\lambda^*_{c}$ over $[0,1]$. Then
\begin{equation}
  \Var(\hat{\tau}^{\mathrm{PPI++}}) \leq \Var(\hat{\tau}^{\mathrm{PPI++\,(clipped)}}) \leq \Var(\hat{\tau}^{(0)}).
  \label{eq:efficiency-ordering}
\end{equation}
If some arm has $\sigma^2_{\hat{Y},t} > 0$ and $n_{L,t} < n_t$, $V$ is strictly convex, the first inequality is strict whenever $\lambda^* \notin [0, 1]$, and the second is strict whenever $\lambda^* > 0$. At $\pi_L = 1$ the three variances coincide. The ordering concerns these optima at a fixed coefficient; the plug-in rule of Equation~\eqref{eq:lambda-opt} minimizes a different objective and incurs the efficiency loss quantified in Section~\ref{sec:eightmethods}, and with an estimated $\hat\lambda$ the ordering holds at first order only.
\end{proposition}

\noindent Proofs for Propositions~\ref{prop:si-bias}--\ref{prop:efficiency} appear in Appendix~\ref{app:proofs}.

Together, Propositions~\ref{prop:si-bias}--\ref{prop:efficiency} delineate the tradeoff. The surrogate index can gain much efficiency but carries a bias of approximately the unmediated effect, exactly $-\delta\kappa$ in the setting of Proposition~\ref{prop:si-bias}. PPI++ rests on the labeling design rather than on surrogacy: exactly unbiased at a fixed predictor and coefficient, first-order valid with the learned index and estimated coefficient, and at the optimal coefficient no more variable than labeled-only; AIPW rests on the same design argument. With an estimated coefficient the gain can be negligible in finite samples and is in any case bounded by prediction quality.

\subsection{Corrected PPI++ Variance Formula}
\label{sec:corrected-variance}

\paragraph{Fixed prediction function.} This subsection states an exact identity for the variance of Equation~\eqref{eq:ppi} at a fixed prediction function and a fixed $\lambda$. The learned-index asymptotics are Proposition~\ref{prop:joint-si-ppi}; the two are kept apart because the first-stage term the proposition carries has no counterpart here.

A naive plug-in variance for the all-units parameterization treats the labeled and all-units averages as if they came from independent samples. Because that parameterization is the convenient one to implement, the mistake is easy to make; it is a hazard of the change of variables, not a defect in \citet{angelopoulos2023ppipp}, whose labeled-unlabeled form assigns the variance derived below. The plug-in drops two covariances that the $n_{L,t}$ shared units create, $\Cov(\bar{Y}_{L,t}, \bar{\hat{Y}}_t) = \gamma_t / n_t$ and $\Cov(\bar{\hat{Y}}_t, \bar{\hat{Y}}_{L,t}) = \sigma_{\hat{Y},t}^2 / n_t$. Restoring both gives
\begin{equation}
  C_t = \frac{2\lambda(\gamma_t - \lambda \sigma_{\hat{Y},t}^2)}{n_t},
  \label{eq:overlap-intuition}
\end{equation}
with total $C = \sum_t C_t$. Write $v_t = \sigma^2_{\hat{Y},t} = \Var(\hat{Y} \mid T = t)$. For $\gamma_t > 0$, $C_t$ is positive when $0 < \lambda < \gamma_t/v_t$, zero at $\lambda = 0$ and at the per-arm optimum $\lambda = \gamma_t/v_t$, and negative beyond that optimum or at a negative coefficient, which an unclipped rule can produce. The common exact optimum of Equation~\eqref{eq:lambda-exact} solves
\begin{equation}
  \sum_t a_t\,(\gamma_t - \lambda v_t) = 0, \qquad a_t = \frac{1}{n_{L,t}} - \frac{1}{n_t},
  \label{eq:lambda-exact-foc}
\end{equation}
whereas $C = 2\lambda\sum_t (\gamma_t - \lambda v_t)/n_t$. The two sums vanish together when $a_t$ is proportional to $1/n_t$, that is, when the armwise labeled fractions $n_{L,t}/n_t$ coincide, or when the per-arm optima $\gamma_t/v_t$ coincide; they do not in general, so $\hat C$ is formed and reported in every case. A coefficient set below the optimum leaves a positive correction. The plug-in rule of Equation~\eqref{eq:lambda-opt} is the clearest instance: in the balanced, arm-homogeneous case it sits below the optimum by the factor $1/(1 + \pi_L)$, so
\begin{equation}
  C_t = \frac{2\lambda_t^{*2}\,\sigma_{\hat{Y},t}^2\, \pi_L}{(1 + \pi_L)^2\, n_t}, \qquad \lambda_t^* = \gamma_t/v_t,
  \label{eq:overlap-piL}
\end{equation}
which increases in $\pi_L$ while the leading term $\sigma_{Y,t}^2 / n_{L,t}$ shrinks; hence the undercoverage worsens as the labeled set fills the sample. Clipping has the same effect whenever the clip binds below $\lambda_t^*$. We call $C_t$ the overlap covariance correction.

\begin{proposition}[Corrected PPI++ variance]
\label{prop:corrected-variance}
Let Assumption~\ref{asm:mcar} hold, fix the prediction function and $\lambda$, and condition on the arm sizes $n_t$ and labeled counts $n_{L,t} \geq 2$. With $\sigma^2_{Y,t} = \Var(Y \mid T = t)$ and $\gamma_t$, $v_t$, $a_t$ as above, the variance of $\hat{\tau}^{\mathrm{PPI}(\lambda)}$ is exactly
\begin{equation}
  V(\lambda) = \sum_t \Bigl[\frac{\sigma^2_{Y,t}}{n_{L,t}} + a_t\bigl(\lambda^2 v_t - 2\lambda\gamma_t\bigr)\Bigr] = V_{\mathrm{plug\text{-}in}}(\lambda) + C(\lambda),
  \label{eq:exact-variance}
\end{equation}
where $V_{\mathrm{plug\text{-}in}}(\lambda) = \sum_t [\sigma^2_{Y,t}/n_{L,t} + \lambda^2 v_t(1/n_t + 1/n_{L,t}) - 2\lambda\gamma_t/n_{L,t}]$ treats the labeled and all-units averages as independent and $C(\lambda) = \sum_t C_t$. Let $\hat{V}_{\mathrm{corrected}} = \hat{V}_{\mathrm{plug\text{-}in}} + \hat{C}$ be formed from the sample moments, with
\begin{equation}
  \hat{C} = \sum_t \frac{2\lambda(\hat{\gamma}_t - \lambda\hat{v}_t)}{n_t}.
\end{equation}
Suppose $n_t/n$ and $n_{L,t}/n$ converge to positive limits, the sample moments are consistent, and $n V(\lambda)$ converges to a positive limit. Then $n\hat{V}_{\mathrm{corrected}} - nV(\lambda) \xrightarrow{p} 0$, and hence $\hat{V}_{\mathrm{corrected}}/V(\lambda) \xrightarrow{p} 1$. The same holds with $\lambda$ replaced by an estimate $\hat\lambda \xrightarrow{p} \lambda$, where $V(\lambda)$ is then the first-order variance of the estimator that uses $\hat\lambda$.
\end{proposition}

The correction term is $O(1/n)$, the same order as the leading variance term $\sigma_{Y,t}^2/n_{L,t}$, so it does not become negligible as the sample grows. At fixed $\pi_L$ the plug-in variance is inconsistent by a constant ratio, and the undercoverage it produces does not vanish with $n$: the relative bias depends on $\pi_L$, through $n_{L,t}/n_t$, and on $\lambda$, not on $n$. This is why the coverage-correction evidence is reported on a $\pi_L$ grid at a single sample size; all cells in Table~\ref{tab:coverage-correction} are at $n = 10{,}000$.

\begin{remark}[A distinct second-order effect]
\label{rem:lambda-estimated}
Both variance formulas treat $\hat{\lambda}$ as fixed at its estimated value. Its sampling variability is a separate mechanism from the overlap covariance and does not enter the first-order variance; we do not characterize its finite-sample contribution.
\end{remark}

Section~\ref{sec:overlap-correction} runs both tuning rules on the same draws.

\subsection{A Diagnostic Test of Estimator Disagreement Between SI and PPI++}
\label{sec:surrogacy-test}

We frame this section as a test of \emph{estimator disagreement}, not a general test of surrogacy. The test asks whether the gap between $\hat{\tau}_{\mathrm{SI}}$ and $\hat{\tau}_{\mathrm{PPI++}}$ exceeds sampling variation. It is a specification test of the \citet{hausman1978specification} type: an estimator efficient under the maintained assumption (SI under surrogacy) is compared with one consistent without it (PPI++). The classical caveats to such pretests, in particular their limited power against local alternatives, apply in full and are quantified in Section~\ref{sec:detection-damage}. Rejection is evidence that the two estimands diverge in the current dataset; non-rejection means only that we cannot distinguish them at the achievable power. The test is strictly weaker than a test of Assumption~\ref{asm:surrogacy}: violations that do not move the gap (for example, when the prediction model is so coarse that PPI++ leans on labeled data alone) are invisible to it, and prediction-model misspecification can in principle open a gap when surrogacy holds.

Under the joint null $H_0$ of statistical surrogacy (Assumption~\ref{asm:surrogacy}) and correct prediction specification ($\hat{f} \xrightarrow{p} m$), both $\hat{\tau}_{\mathrm{SI}}$ and $\hat{\tau}_{\mathrm{PPI++}}$ are consistent for $\tau$, so $\hat{D} = \hat{\tau}_{\mathrm{PPI++}} - \hat{\tau}_{\mathrm{SI}}$ has mean zero asymptotically. Under departures from $H_0$, the gap reflects the estimand difference. In the specific linear partial-mediation DGP studied below (DGP~2), this gap equals the direct treatment effect attenuated by the factor of Proposition~\ref{prop:si-bias}: $\E[\hat{D}] = \delta\kappa$, with $\kappa = \sigma_{\varepsilon_S}^2/(\sigma_{\varepsilon_S}^2 + \gamma_S^2 p(1-p)) = 0.9944$ at the DGP~2 parameters. In general, the gap is a function of both surrogacy failure and model misspecification, and test rejection should not be read as proof of either failure in isolation.

\begin{proposition}[Asymptotic distribution of the diagnostic test]
\label{prop:surrogacy-test}
Assume the conditions of Proposition~\ref{prop:joint-si-ppi}, $H_0$, and $\sigma_D^2 = (1,-1)\,\Omega\,(1,-1)' = \Omega_{11} + \Omega_{22} - 2\Omega_{12} > 0$. Then $\sqrt{n}\,\hat{D} \xrightarrow{d} \mathcal{N}(0, \sigma_D^2)$, $\sqrt{n}\,\mathrm{SE}(\hat D) \xrightarrow{p} \sigma_D$ with $\mathrm{SE}(\hat D)$ the square root of Equation~\eqref{eq:var-D-sandwich}, and
\begin{equation}
  T_n = \frac{\hat{D}}{\mathrm{SE}(\hat{D})} \xrightarrow{d} \mathcal{N}(0, 1).
\end{equation}
\end{proposition}

\begin{remark}[Fixed prediction function]
\label{rem:fixed-f-cov}
For a fixed prediction function the covariance has the closed form $\Cov(\hat{\tau}_{\mathrm{SI}}, \hat{\tau}_{\mathrm{PPI++}}) = \sum_t \gamma_t/n_t$, derived in Appendix~\ref{app:hybrid-proofs}. That identity belongs to the fixed-predictor block alongside Propositions~\ref{prop:ppi-validity} and~\ref{prop:corrected-variance}; it understates $\Var(\hat{\tau}_{\mathrm{SI}})$ and misstates the covariance when the index is learned, which is why $\mathrm{SE}(\hat D)$ is built from $\hat\Omega$ instead.
\end{remark}

The two-sided test is the workflow default. On the power grid of Table~\ref{tab:power-curves}, at $n = 10{,}000$ and $\pi_L = 0.20$, its power at $\alpha = 0.05$ is 13\% at $\rho = 0.2$ and 63\% at $\rho = 0.4$. The one-sided test against $\hat{D} > 0$, the direction in which a reinforcing direct effect (as in DGP~2) pushes the gap, reaches 22\% at $\rho = 0.2$, 74\% at $\rho = 0.4$, and 95\% at $\rho = 0.5$; because the sign of a violation is rarely known ex ante, those figures are a best case and are not interchangeable with the two-sided ones. At $n = 100{,}000$ the test detects $\rho = 0.2$ with 79\% power two-sided and 87\% one-sided. All simulation tables report $\alpha = 0.05$; at the workflow's $\alpha = 0.10$ the paired run of Section~\ref{sec:detection-damage} puts two-sided power at 19.2\% at $\rho = 0.2$ and $n = 10{,}000$. Larger samples detect any fixed violation, but the coverage-destroying violation shrinks with $n$ as well (Section~\ref{sec:detection-damage}).

\paragraph{Null calibration of the standard error.} Because $\mathrm{SE}(\hat{D})$ carries the first-stage term, its calibration has to be checked directly rather than inferred from the fixed-predictor algebra. Two runs on the same five cells check it. Against the Monte Carlo variance, an $R = 2{,}000$ run without the bootstrap puts the ratio of sandwich to Monte Carlo (co)variance between $0.948$ and $1.040$ in all twenty entries, each within two Monte Carlo standard errors of one (Table~\ref{tab:joint-cov-R2000}). Against resampling, a joint paired bootstrap that resamples units stratified by arm, keeps every copy of a unit in its unit's fold, and refits both estimators is within $5.8\%$ of the sandwich for $\Var(\hat{\tau}_{\mathrm{SI}})$, $\Var(\hat{\tau}_{\mathrm{PPI++}})$, their covariance, and $\Var(\hat{D})$ on all four bootstrapped cells of the $R = 200$ run (Table~\ref{tab:joint-cov-check}), with the sandwich slightly below it for $\Var(\hat{\tau}_{\mathrm{PPI++}})$ and $\Var(\hat{D})$. The first-stage term is what makes the difference where the index leans on a treatment-responsive predictor: on the Criteo-calibrated rich index the ``SI (plug-in variance)'' ablation is $0.281$ of the Monte Carlo variance of $\hat\tau_{\mathrm{SI}}$ and its interval covers $0.665$, while the sandwich is $0.963$ of it and covers $0.960$; at $R = 2{,}000$ the plug-in interval covers $0.701$ and the sandwich $0.951$. The designs that Proposition~\ref{prop:joint-si-ppi} does not cover (Table~\ref{tab:implementation}) are checked in Appendix~\ref{app:ablations}.

Size is read from the $R = 2{,}000$ run, on the three Gaussian cells where the joint null holds exactly. With the sandwich $\mathrm{SE}(\hat{D})$ the two-sided rejection rate is $0.051$ at $\alpha = 0.05$ and $0.100$ at $\alpha = 0.10$ on DGP~1 at $\pi_L = 0.20$ and $n = 10{,}000$, $0.056$ and $0.107$ on DGP~2 at $\rho = 0$ and $n = 10{,}000$, and $0.054$ and $0.107$ on DGP~2 at $\rho = 0$ and $n = 100{,}000$, with Monte Carlo standard errors of $0.005$ and $0.007$: within Monte Carlo error of nominal at both sample sizes. The Criteo-calibrated rich index at $m = 1$ gives $0.058$ and $0.105$, but it is a near-null stress check rather than a size experiment (Section~\ref{sec:semisynthetic}). Smaller runs at $\alpha = 0.05$ range from $0.030$ (the paired run behind Figure~\ref{fig:detection-damage}, $R = 500$, below its Monte Carlo band) to $0.080$ (the $R = 200$ rows of Table~\ref{tab:joint-cov-check}, inside its band); the $R = 2{,}000$ run is the nominal-size evidence.

Proofs for the results of Sections~\ref{sec:corrected-variance} and~\ref{sec:surrogacy-test} are collected in Appendix~\ref{app:hybrid-proofs}.

\section{Simulation Design}
\label{sec:simulation}

\subsection{Data-Generating Processes}
\label{sec:dgps}

We design eleven DGPs to answer ten questions about surrogate method performance (Table~\ref{tab:dgps}). Each generates $(X_i, T_i, S_i, Y_i)$ for $n$ units with $T_i \iid \text{Bernoulli}(0.5)$; DGP~11 also assigns enrollment times and labels the earliest-enrolled cohort.

Two simplifications relate the DGPs to the motivating examples. The DGPs use continuous Gaussian surrogates and outcomes with linear or mildly nonlinear links, whereas the motivating outcomes are often binary or skewed; every estimator here operates on arm means, so the simulations bear on those outcomes through the normal approximation for the arm means, which the binary Hillstrom and Criteo analyses test and the sparse Criteo subsamples break (Section~\ref{sec:criteo}). And the main analysis uses one surrogate; DGP~7 and the Criteo-calibrated testbed of Section~\ref{sec:semisynthetic} study several.

\begin{table}[!ht]
\centering
\caption{Summary of data-generating processes.}
\label{tab:dgps}
\footnotesize
\setlength{\tabcolsep}{4pt}
\begin{tabular}{@{}cl >{\raggedright\arraybackslash}p{3.0cm} >{\raggedright\arraybackslash}p{3.0cm} c@{}}
\toprule
\textbf{DGP} & \textbf{Name} & \textbf{Key Feature} & \textbf{What It Tests} & \textbf{Valid?} \\
\midrule
1 & Valid surrogate & $Y$ depends on $T$ only via $S$ & Calibration benchmark & Yes \\
2 & Partial mediation & Direct effect $\delta \cdot T$ on $Y$ & Classic surrogacy violation & No \\
3 & Heterogeneous quality & $\beta_{YS}$ varies across groups & Simpson's paradox risk & Partial \\
4 & Stale historical index & Historical $\beta_{YS}$ drifts before the experiment & Transport of a historical index & Yes$^{*}$ \\
5 & Sparse or delayed outcome & $Y_i$ observed for fraction $q$ & Outcome missingness & Yes \\
6 & Portfolio & $K$ small experiments & Cumulative decision quality & Yes \\
7 & Multi-surrogate & $J = 3$ surrogates & Multi-surrogate scaling & Yes \\
8 & Nonlinear & Quadratic $Y$--$S$ & Robustness to nonlinearity & Yes \\
9 & Antagonistic surrogate & $\beta_{YS} < 0$, large $\delta$ & Unclipped vs.\ clipped PPI++ & No \\
10 & Nonlinear + partial med. & Quadratic + direct $\delta T$ & Prop.~\ref{prop:si-bias} under nonlinearity & No \\
11 & Enrollment-time labeling & Earliest cohort labeled, drift & The maturation mechanism itself & Varies \\
\bottomrule
\end{tabular}
\smallskip

{\scriptsize \noindent Notes: ``Valid?'' refers to statistical surrogacy (Assumption~\ref{asm:surrogacy}) within the current experiment. $^{*}$In DGP~4 surrogacy holds within the experiment; what fails is the transport of the historical index and of the composite-proxy weights, both fit before the slope drift.}
\end{table}

\paragraph{DGP 1 (Valid Surrogate).}
The baseline. Treatment affects $Y$ entirely through $S$ (Assumption~\ref{asm:surrogacy} holds):
\begin{equation}
  S_i = 5 + X_i + 0.3 \cdot T_i + \varepsilon_{S,i}, \qquad Y_i = 0.5 \cdot S_i + 0.2 \cdot X_i + \varepsilon_{Y,i},
\end{equation}
with $X_i \sim \mathcal{N}(0,1)$, $\varepsilon_{S,i} \sim \mathcal{N}(0, 4)$, $\varepsilon_{Y,i} \sim \mathcal{N}(0, 1)$. The true ATE is $\tau = 0.5 \times 0.3 = 0.15$.

\paragraph{DGP 2 (Partial Mediation).}
Adds a direct treatment effect on $Y$ not captured by $S$:
\begin{equation}
  Y_i = 0.5 \cdot S_i + 0.2 \cdot X_i + \delta \cdot T_i + \varepsilon_{Y,i}.
\end{equation}
The direct-effect share $\rho = \delta / (\beta_{YS}\gamma_S + \delta)$ is swept over $\{0, 0.2, 0.4, 0.6, 0.8\}$ on the core grid; the diagnostic power curves (Appendix~\ref{app:power-curves}) use the finer grid $\{0, 0.1, 0.2, 0.3, 0.4, 0.5, 0.6, 0.8\}$. Holding the mediated part $\beta_{YS}\gamma_S = 0.15$ fixed and dialing $\rho$ means dialing $\delta = \rho\tau$ with $\tau = 0.15/(1-\rho)$, so the true ATE grows with $\rho$ and $\rho$ is a nonlinear reparameterization of the direct effect: equal steps in $\rho$ are unequal steps in $\delta$. We report on the $\rho$ scale because it is the interpretable one, and convert to $\delta$ wherever a rate is at stake (Section~\ref{sec:detection-damage}).

\paragraph{DGP 3 (Heterogeneous Surrogate Quality).}
A latent group $G_i \in \{1, 2\}$ with $P(G_i = 1) = 0.3$ introduces subgroup variation: $\beta_{YS}^{(1)} = 0.8$ (power users), $\beta_{YS}^{(2)} = \beta_{YS}^{\text{low}}$ (casual users), with $\beta_{YS}^{\text{low}}$ swept over $\{0.0, 0.2, 0.4\}$. The treatment effect on the surrogate also varies by group, $\gamma_S^{(1)} = 0.5$ for power users and $\gamma_S^{(2)} = 0.2$ for casual users, so the true ATE is $\tau = 0.3(0.5)(0.8) + 0.7(0.2)\beta_{YS}^{\text{low}}$. A pooled prediction model averages the two regimes and is therefore biased for $\tau$.

\paragraph{DGP 4 (Stale Historical Index).}
The surrogate index is fit on $n_{\text{cal}} = 50{,}000$ historical observations generated with $\beta_{YS}^{(0)} = 0.5$, and the experiment runs after the slope drifts to $\beta_{YS}^{(1)} = \beta_{YS}^{(0)} + \Delta_\beta$, with $\Delta_\beta \in \{0.0, 0.1, 0.2, 0.4\}$. Within the current experiment treatment reaches $Y$ only through $S$, so statistical surrogacy holds; the drift breaks the transport of the historical index, not surrogacy. DGP~4 therefore compares a stale historical index, and composite-proxy weights fit at the pre-drift slope (Method~4), with methods that use current labels (labeled-only, the naive surrogate, PPI++, and AIPW, the last two through the cross-fitted predictions). Because training data change along with the estimator, it does not isolate the value of the PPI correction under drift, and we make no claim about it. It is also narrower than Regime~B: only the $Y$--$S$ slope moves, so covariate drift, treatment-policy shifts, and changes in $\gamma_S$ are not probed. The surrogate-index interval treats the historical fit as fixed (the fixed-predictor plug-in variance of Proposition~\ref{prop:ppi-validity}'s setting) and does not propagate its training uncertainty; the calibration sample is redrawn in every replication, so the reported coverage is unconditional, and at $\Delta_\beta = 0$ and $\pi_L = 0.20$ it is $0.950$ (Table~\ref{tab:appendix-dgp4}), so the omitted training term is negligible at $n_{\text{cal}} = 50{,}000$.

\paragraph{DGP 5 (Sparse or Delayed Outcome).}
Identical to DGP~1 but with outcome observation rate $q \in \{0.05, 0.10, 0.20, 0.30\}$, under either MCAR or MAR. Under MCAR the labeling indicator $M_i$ is drawn independently of $(T_i, S_i, Y_i)$; under MAR, $P(M_i = 1 |S_i) = \mathrm{logit}^{-1}(\eta_0 + 0.3 S_i)$ with $\eta_0$ calibrated so the marginal observation rate equals $q$. Labeled-only and PPI++ have no design-based guarantee under MAR unless the labeling is modeled, and we report them unadjusted, so the MAR rows are a stress test. AIPW uses a logistic labeling propensity in $(S, S^2)$ estimated from the data, correctly specified for this model. The surrogate-index interval and $\mathrm{SE}(\hat D)$ use the sandwich of Proposition~\ref{prop:joint-si-ppi} with the labeled-design $\hat Q$, the appropriate Gram matrix for the surrogate index here: labeling depends only on $(S, X)$ and the outcome model is correct, so $\E[\varepsilon \mid g, L = 1] = 0$, the labeled OLS limit is still $\beta_0$, and the first-stage term is $(L_i/\pi_L)\,d_0'Q_L^{-1}g_i\varepsilon_i$ with $Q_L = \E[gg' \mid L = 1]$. Nonrandom labeling breaks PPI++'s design-based centering instead, and with it the diagnostic's null (Appendix~\ref{app:ablations}).

\paragraph{DGP 6 (Portfolio).}
A platform runs $K = 200$ independent experiments, each with sample size $n_k \sim \text{Uniform}\{500, \ldots, 5{,}000\}$. Half are null ($\tau_k = 0$); the other half have small effects drawn from $\mathcal{N}(0, 0.01)$. Within each experiment the outcome model is DGP~1's. The rule is to launch if $\hat{\tau}_k > 0$ and $p < 0.05$, and we evaluate \textit{cumulative regret} $V^* - V$, with $V^* = \sum_k \max(\tau_k, 0)$ the oracle value and $V$ the realized payoff. A variant in which 30\% of experiments carry sign-preserving partial mediation ($\rho = 0.3$) gives the same ranking (relative regret 77\% for labeled-only, 35\% for the naive surrogate, 40\% for the surrogate index, and 66\% for PPI++ at $R = 200$); in antagonistic settings (DGP~9) surrogate-based launch rules can flip sign.

\paragraph{DGP 7 (Multi-Surrogate).}
Three surrogates capture different aspects of the treatment: engagement ($S_1$, treatment effect $\gamma_1 = 0.3$), content completions ($S_2$, $\gamma_2 = 0.2$), and social interactions ($S_3$, $\gamma_3 = 0.1$). The primary outcome depends on all three: $Y_i = 0.3 S_{1,i} + 0.4 S_{2,i} + 0.2 S_{3,i} + 0.2 X_i + \varepsilon_{Y,i}$. The true ATE is $\tau = \sum_j \beta_j \gamma_j = 0.3(0.3) + 0.4(0.2) + 0.2(0.1) = 0.19$. The surrogate index uses a multi-dimensional prediction model $\hat{f}(S_1, S_2, S_3, X)$, OLS on $(1, S_1, S_2, S_3, S_1^2, S_2^2, S_3^2, X)$. The composite proxy calibrates its weight vector from $K_{\text{hist}} = 50$ historical experiments.

\paragraph{DGP 8 (Nonlinear).}
The surrogate--outcome relationship is concave: $Y_i = 0.5 S_i - 0.03(S_i - 5)^2 + 0.2 X_i + \varepsilon_{Y,i}$, producing diminishing returns to higher surrogate values. The surrogate equation is identical to DGP~1. The true ATE is $\tau = 0.1473$, slightly below the linear case because the concavity attenuates the treatment effect. The quadratic prediction model $(S, S^2, X)$ can capture this relationship, but a deliberately misspecified linear model (omitting $S^2$) is also tested to evaluate robustness.

\paragraph{DGP 9 (Antagonistic Surrogate).}
The surrogate responds positively to treatment ($\gamma_S = 0.3$) but relates negatively to the outcome ($\beta_{YS} = -0.3$), and a large direct effect ($\delta = 0.5$) bypasses it, a designed instance of the surrogate paradox \citep{vanderweele2013surrogate}. This DGP sets $\beta_{YX} = 0.8$, against $0.2$ elsewhere, so the covariate carries most of the outcome variation and PPI++ remains a useful rectifier though the surrogate misleads. The true ATE is $\tau = (-0.3)(0.3) + 0.5 = 0.41$, while the surrogate index estimates $\beta_{YS}\gamma_S = -0.09$, the wrong sign. Clipping the coefficient to $[0,1]$ binds here, so this is where the primary unclipped PPI++ is set against the ``PPI++ (clipped)'' ablation (Table~\ref{tab:appendix-dgp9}).

\paragraph{DGP 10 (Nonlinear + Partial Mediation).}
Combines DGP~8's quadratic surrogate--outcome relationship with DGP~2's partial mediation, sweeping $\rho \in \{0, 0.2, 0.4\}$, to test whether the bias of Proposition~\ref{prop:si-bias}, derived under linearity, extends qualitatively to nonlinear settings (Section~\ref{sec:results-gbt}).

\paragraph{DGP 11 (Enrollment-Time Labeling Under Drift).}
Every other DGP draws the labeled set as a uniform random subset (DGP~5's MAR variant aside). DGP~11 labels the units the maturation process would label: each unit receives an enrollment time $e_i \sim \text{Uniform}(0,1)$ and the labeled set is the earliest-enrolled fraction $\pi_L$. Two drift channels are dialed separately. Under \textit{mix drift} the covariate distribution shifts over the window, $X_i = d(e_i - 0.5) + \mathcal{N}(0,1)$, and the treatment effect on the surrogate is heterogeneous in $X$ ($\gamma_i = 0.5 + 0.3 X_i$), so the early cohort's average effect differs from the full-sample ATE while $Y \mid S, X$ stays stable and surrogacy remains valid. Under \textit{effect drift} a direct effect $g_e \, e_i \, T_i$ on $Y$ grows over the window through a channel that bypasses the surrogate, which no model of $(S, X)$ can recover. The outcome model is $S_i = 5 + X_i + \gamma_i T_i + \mathcal{N}(0,4)$, $Y_i = 0.3 S_i + 0.2 X_i + g_e e_i T_i + \mathcal{N}(0,1)$. The same effect-drift configuration under a random split isolates the labeling mechanism. As in the DGP~5 MAR rows, the surrogate-index interval applies the MCAR sandwich with the labeled-design $\hat Q$ as a stress test.

\subsection{Evaluation Design}
\label{sec:eval-metrics}

We evaluate each method over $R$ replications on five metrics: the bias $\widehat{\text{Bias}} = R^{-1} \sum_r (\hat{\tau}^{(r)} - \tau)$, also reported as relative bias $\widehat{\text{Bias}}/\tau$ in percent; the RMSE; the coverage of the true $\tau$ by the 95\% interval; the correct decision rate (CDR), the fraction of replications that reject the null when $\tau \neq 0$ or fail to reject it when $\tau = 0$; and the relative efficiency $\text{RE}_m = \text{RMSE}_{\text{LO}} / \text{RMSE}_m$, the RMSE ratio against labeled-only. Its square, $\text{MSE}_{\text{LO}}/\text{MSE}_m$, is the effective-sample-size (ESS) multiplier, reported for the CUPED, multi-surrogate, and Cross-PPI comparisons: a $5\times$ RE is a $25\times$ ESS multiplier. Columns labeled ``RE'' use the RMSE ratio and columns labeled ``ESS mult.'' the MSE ratio. Main-text tables suppress RE for methods with absolute relative bias above 10\%, because tight clustering around a biased estimate is not estimation quality.

The core grid (DGPs 1--5), the hybrid evaluation, and the hard-switch comparison use $R = 2{,}000$ replications per configuration, at which the Monte Carlo standard error of coverage is $\sqrt{0.95 \times 0.05 / 2000} = 0.0049$ and the 95\% simulation band around the nominal level is $[0.940, 0.960]$. Auxiliary components use $R = 1{,}000$ (DGP~11 and the hybrid sensitivity grid; band $[0.936, 0.964]$), $R = 500$ (DGPs~7--10 and the power curves; band $[0.931, 0.969]$), and $R = 200$ (the GBT comparison; band $[0.920, 0.980]$), as each table's note states. The default sample size is $n = 10{,}000$, with $\pi_L \in \{0.05, 0.20, 0.50, 1.00\}$ for DGP~1 and $\pi_L \in \{0.10, 0.20, 0.50, 1.00\}$ for DGPs~2--4; $\pi_L = 0.80$ appears only in Tables~\ref{tab:coverage-correction} and~\ref{tab:bootstrap-sensitivity}. The primary prediction model is OLS on $(S, S^2, X)$ under the protocol of Section~\ref{sec:eightmethods}. The design yields 514 method-by-scenario cells in 93 scenarios, a cell being one primary estimator (Methods 0 to 5) on one (DGP, configuration, $\pi_L$) scenario: 36 cells over 7 scenarios for DGP~1, 132 over 23 for DGP~2, 54 over 9 for DGP~3, 72 over 12 for DGP~4, 52 over 9 for DGP~5, 40 over 7 for DGP~6, 12 over 3 for each of DGPs~7 and~8, 18 over 3 for DGP~9, 54 over 9 for DGP~10, and 32 over 8 for DGP~11. The ablation configurations run on the same scenarios and are not counted.

\paragraph{Implementation scope.} Proposition~\ref{prop:joint-si-ppi} covers the cross-fitted OLS index under MCAR labeling. Table~\ref{tab:implementation} records, for each design family, the training data, learner, labeling, surrogate-index variance, diagnostic covariance, and justification, including the three designs outside the proposition: the external DGP~4 index, whose interval conditions on its fit; the GBT rows, whose surrogate-index interval is heuristic; and the DGP~5 MAR rows and DGP~11, which apply the MCAR sandwich outside MCAR as a stress test.

\begin{table}[!ht]
\centering
\caption{Implementation of the surrogate-index interval, the diagnostic, and the composite proxy by design family.}
\label{tab:implementation}
\scriptsize
\setlength{\tabcolsep}{3pt}
\begin{tabular}{@{}>{\raggedright\arraybackslash}p{2.3cm} >{\raggedright\arraybackslash}p{2.3cm} >{\raggedright\arraybackslash}p{1.9cm} >{\raggedright\arraybackslash}p{1.8cm} >{\raggedright\arraybackslash}p{2.0cm} >{\raggedright\arraybackslash}p{1.8cm} >{\raggedright\arraybackslash}p{2.8cm}@{}}
\toprule
\textbf{Design family} & \textbf{Training data} & \textbf{Learner} & \textbf{Labeling} & \textbf{SI variance} & \textbf{Diagnostic covariance} & \textbf{Justification} \\
\midrule
Core Gaussian DGPs 1--3, 5 (MCAR), 7--10 & Current labels, all-units five-fold cross-fitting & OLS on $(1, S, S^2, X)$, elementwise for $J = 3$ & MCAR random subset & Joint sandwich & Joint sandwich & Proposition~\ref{prop:joint-si-ppi}; validated at $R = 2{,}000$ (Table~\ref{tab:joint-cov-R2000}) \\
DGP 4 & SI: one fit on $n_{\text{cal}} = 50{,}000$ historical units at the pre-drift slope; CP: past experiments at that slope; LO, NS, PPI++, AIPW: current labels & OLS on $(1, S, S^2, X)$ & MCAR random subset & Plug-in, external fit held fixed & Not reported & Surrogacy holds within the experiment. Fixed-predictor setting of Proposition~\ref{prop:ppi-validity}; training uncertainty not propagated; coverage unconditional (calibration sample redrawn) \\
DGP 5 MAR; DGP 11 & Current labels, cross-fitted & OLS on $(1, S, S^2, X)$ & MAR on $S$; earliest-enrolled cohort & MCAR sandwich with labeled-design $\hat Q$ & Same & SI: labeled-design $\hat Q$ correct if labeling acts through $(S, X)$ and the model is correct (variance ratios $1.08$, $0.94$). PPI++ has no design-based centering under nonrandom labels; its exact-variance ratio is $0.967$ (MAR) and $0.873$ (mix drift) \\
GBT rows (Table~\ref{tab:gbt}) & Current labels, cross-fitted & Gradient-boosted trees & MCAR random subset & Plug-in (fixed predictor) & Fixed-predictor covariance (two-cell check only) & None; heuristic interval, checked on two cells (Appendix~\ref{app:ablations}) \\
Semi-synthetic testbed; multi-surrogate & Current labels, cross-fitted & OLS on $(1, S, S^2, X)$; rich index $(1, S_1, S_2,$ $S_1^2, S_2^2, X)$, $S_1^2 = S_1$ & MCAR random subset & Joint sandwich (column-space inverse) & Joint sandwich & Proposition~\ref{prop:joint-si-ppi}; rich index at $m = 1$ is a near-null \\
Hillstrom, Criteo, LaLonde & Masked labels, cross-fitted & OLS on $(1, S, S^2, X)$ & Random masking (MCAR) & Joint sandwich & Joint sandwich & Proposition~\ref{prop:joint-si-ppi}; target is the masking target \\
Composite proxy (every design that runs it) & Estimated $\hat\tau_S^{(k)}$, $\hat\tau_Y^{(k)}$ from $K_{\text{hist}} = 30$ simulated past experiments ($50$ in Table~\ref{tab:dgp7}); no $\hat{Y}$ & No-intercept least squares across past experiments & Current labels not used & CP: $\hat{\mathbf{w}}'\widehat{\Var}(\hat{\boldsymbol{\tau}}_S)\hat{\mathbf{w}}$, weights held fixed & Not used & Past experiments satisfy surrogacy, so calibration is favorable; falls back to SI without past experiments (DGP~6, Hillstrom) \\
\bottomrule
\end{tabular}
\smallskip

{\scriptsize \noindent Notes: ``Plug-in'' is the Neyman variance of the difference in arm means of $\hat{Y}$ with the predictor held fixed. ``Fixed-predictor covariance'' is Remark~\ref{rem:fixed-f-cov}. PPI++ uses the exact variance of Proposition~\ref{prop:corrected-variance} in every row. Every design uses the all-units protocol except two historical calibrations, of the DGP~4 surrogate index and of the composite proxy. DGP~6 uses the core protocol and reports decisions rather than intervals. The two-cell checks are at DGP~5 MAR with $q = 0.10$ and DGP~11 mix drift $d = 2$, and at GBT DGP~1 with $\pi_L = 0.20$ and DGP~2 with $\rho = 0.2$, $\pi_L = 0.20$ (Appendix~\ref{app:ablations}).}
\end{table}

All estimators are sub-millisecond at $n = 10{,}000$ with OLS predictions; the analytic variance correction is approximately $100\times$ cheaper than bootstrap (see Table~\ref{tab:timing} in the Appendix for full timing benchmarks).

\section{Results}
\label{sec:results}

At the featured review, $\pi_L = 0.20$, conversion windows have closed for a fifth of the customers enrolled by week 2.5 while every enrolled customer has funnel signals. Sections~\ref{sec:results-valid}--\ref{sec:results-violations} establish the core tradeoff on DGPs~1--5, Section~\ref{sec:detection-damage} locates the detection-damage gap, and Sections~\ref{sec:results-enrollment}--\ref{sec:results-hybrid} take up enrollment-time labeling, portfolio decisions, the overlap correction, robustness checks, and the hybrid. Per-DGP results beyond the tables printed here are listed in Appendix~\ref{app:full-results}.

\textbf{Bold} coverage entries in the main-text tables fall outside that table's Monte Carlo band around the nominal 95\% (Section~\ref{sec:eval-metrics}), restated in each note.

\subsection{Efficiency Under Valid Surrogates}
\label{sec:results-valid}

Table~\ref{tab:dgp1} reports the DGP~1 results.

\begin{table}[!ht]
\centering
\caption{DGP 1 (Valid Surrogate): method performance across labeled fractions.}
\label{tab:dgp1}
\footnotesize
\setlength{\tabcolsep}{4pt}
\begin{tabular}{@{}llrrrrc@{}}
\toprule
\textbf{Method} & $\boldsymbol{\pi_L}$ & \textbf{Bias} & \textbf{RMSE} & \textbf{Coverage} & \textbf{RE} & \textbf{CDR} \\
\midrule
Labeled-Only     & 0.05 & $-0.004$ & $0.142$ & $0.948$ & $1.00$ & $0.18$ \\
Naive Surrogate  & 0.05 & $+0.012$ & $0.028$ & $\mathbf{0.918}$ & --- & $1.00$ \\
Surrogate Index  & 0.05 & $-0.0001$ & $0.026$ & $0.949$ & $5.54$ & $1.00$ \\
PPI++            & 0.05 & $-0.004$ & $0.093$ & $0.953$ & $1.52$ & $0.35$ \\
Composite Proxy  & 0.05 & $+0.002$ & $0.023$ & $0.951$ & $6.15$ & $1.00$ \\
AIPW             & 0.05 & $-0.004$ & $0.092$ & $0.953$ & $1.53$ & $0.35$ \\
\midrule
Labeled-Only     & 0.20 & $-0.003$ & $0.071$ & $0.949$ & $1.00$ & $0.55$ \\
Naive Surrogate  & 0.20 & $+0.012$ & $0.027$ & $\mathbf{0.928}$ & --- & $1.00$ \\
Surrogate Index  & 0.20 & $-0.0005$ & $0.024$ & $0.956$ & $2.91$ & $1.00$ \\
PPI++            & 0.20 & $-0.002$ & $0.051$ & $0.950$ & $1.40$ & $0.83$ \\
Composite Proxy  & 0.20 & $+0.002$ & $0.022$ & $0.956$ & $3.15$ & $1.00$ \\
AIPW             & 0.20 & $-0.002$ & $0.051$ & $0.950$ & $1.40$ & $0.84$ \\
\midrule
Labeled-Only     & 0.50 & $-0.002$ & $0.045$ & $0.953$ & $1.00$ & $0.92$ \\
Naive Surrogate  & 0.50 & $+0.012$ & $0.027$ & $\mathbf{0.919}$ & --- & $1.00$ \\
Surrogate Index  & 0.50 & $-0.0001$ & $0.025$ & $0.949$ & $1.82$ & $1.00$ \\
PPI++            & 0.50 & $-0.001$ & $0.038$ & $0.953$ & $1.19$ & $0.98$ \\
Composite Proxy  & 0.50 & $+0.002$ & $0.023$ & $0.947$ & $1.96$ & $1.00$ \\
AIPW             & 0.50 & $-0.001$ & $0.038$ & $0.954$ & $1.19$ & $0.98$ \\
\bottomrule
\end{tabular}
\smallskip

{\scriptsize \noindent Notes: $n = 10{,}000$, $\tau = 0.15$, $R = 2{,}000$, all-units five-fold cross-fitting. Every row is a primary configuration: PPI++ with the exact tuning rule of Equation~\eqref{eq:lambda-exact} and the exact variance, the surrogate index with the joint sandwich variance of Proposition~\ref{prop:joint-si-ppi}; the ablations, on the same replications, are in Tables~\ref{tab:ppi-ablation} and~\ref{tab:si-variance-ablation}. \textbf{Bold} coverage entries fall outside the band $[0.940, 0.960]$. RE is $\mathrm{RMSE}_{\mathrm{LO}}/\mathrm{RMSE}_{\mathrm{method}}$. The naive surrogate's RE is suppressed (---) by convention: its $+7.7\%$ to $+7.9\%$ relative bias is the omitted-variable bias of the simple slope, and it is evaluated on decisions.}
\end{table}

When the surrogacy assumption holds, SI and the composite proxy (CP) deliver large efficiency gains: RE of $5.5$--$6.2\times$ at $\pi_L = 0.05$, narrowing to $1.8$--$2.0\times$ at $\pi_L = 0.50$. PPI++ provides more modest improvements ($1.2$--$1.5\times$), and AIPW stays within one hundredth of PPI++ in RE at every labeled fraction.
The plug-in tuning rule costs $0.3\%$, $1.4\%$, and $2.3\%$ of PPI++'s RE at $\pi_L = 0.05$, $0.20$, and $0.50$, with coverage inside the band under both rules (Table~\ref{tab:ppi-ablation}). The gap between the two families is structural: SI fully trusts the prediction model, while PPI++ hedges. Against a CUPED-adjusted baseline (Table~\ref{tab:cuped-full}), SI delivers RE of $1.6$--$4.9\times$ (ESS multiplier $2.5$--$24\times$) and PPI++ $1.08$--$1.35\times$ ($1.2$--$1.8\times$) in the valid-surrogate cells, so surrogate methods add value beyond CUPED.

These relative efficiencies do not translate directly into calendar time, because advancing a review date enrolls more units as well as maturing more outcomes, whereas the simulations sweep $\pi_L$ at a fixed $n$. Table~\ref{tab:calendar} therefore runs the review date directly: uniform enrollment over four weeks to an eventual $n = 10{,}000$, a two-week maturation window, DGP~1's outcome model with no drift, and a prespecified review at $t \in \{2.5, 3.0, 3.5, 4.0, 4.5, 5.0, 6.0\}$, at which $n(t)$ and $\pi_L(t)$ of Equation~\eqref{eq:piL-calendar} move together. Coverage is of the eventual-sample ATE and lies between 0.938 and 0.960 for all three estimators at every review date, so the comparison is one of precision and power at equal validity.

\begin{table}[!ht]
\centering
\caption{Calendar-review experiment: precision, coverage, and power at prespecified review dates.}
\label{tab:calendar}
\footnotesize
\setlength{\tabcolsep}{4pt}
\begin{tabular}{@{}rrrr rrrr rrr@{}}
\toprule
& & & & \multicolumn{4}{c}{\textbf{Mean 95\% CI width}} & \multicolumn{3}{c}{\textbf{Power vs} $\boldsymbol{\tau = 0}$} \\
\cmidrule(lr){5-8}\cmidrule(lr){9-11}
$\boldsymbol{t}$ & $\boldsymbol{n(t)}$ & $\boldsymbol{n_L(t)}$ & $\boldsymbol{\pi_L(t)}$ & LO & SI & PPI++ & plug-in $\lambda$ & LO & SI & PPI++ \\
\midrule
2.5 & 6{,}251  & 1{,}250  & 0.200 & $0.350$ & $0.122$ & $0.254$ & $0.257$ & $0.373$ & $0.997$ & $0.612$ \\
3.0 & 7{,}501  & 2{,}501  & 0.333 & $0.248$ & $0.111$ & $0.192$ & $0.196$ & $0.652$ & $1.000$ & $0.860$ \\
3.5 & 8{,}751  & 3{,}750  & 0.429 & $0.202$ & $0.103$ & $0.164$ & $0.168$ & $0.806$ & $1.000$ & $0.947$ \\
4.0 & 10{,}000 & 5{,}000  & 0.500 & $0.175$ & $0.096$ & $0.147$ & $0.150$ & $0.914$ & $1.000$ & $0.987$ \\
4.5 & 10{,}000 & 6{,}251  & 0.625 & $0.157$ & $0.096$ & $0.138$ & $0.141$ & $0.952$ & $1.000$ & $0.994$ \\
5.0 & 10{,}000 & 7{,}501  & 0.750 & $0.143$ & $0.096$ & $0.132$ & $0.134$ & $0.987$ & $1.000$ & $0.997$ \\
6.0 & 10{,}000 & 10{,}000 & 1.000 & $0.124$ & $0.096$ & $0.124$ & $0.124$ & $0.999$ & $1.000$ & $0.999$ \\
\bottomrule
\end{tabular}
\smallskip

{\scriptsize \noindent Notes: DGP~1 (valid surrogacy, no drift), uniform enrollment over weeks 0--4 to an eventual $n = 10{,}000$, two-week outcome window, $\tau = 0.15$, $R = 1{,}000$, all-units five-fold cross-fitting. Enrollment ends at week 4, which is why $n(t)$ stops growing there while $n_L(t)$ continues. ``PPI++'' is the primary configuration, the exact-variance tuning rule of Equation~\eqref{eq:lambda-exact} with a common unclipped coefficient and the exact variance; ``plug-in $\lambda$'' is the ablation of Equation~\eqref{eq:lambda-opt} on the same replications, and also uses the exact variance. SI uses the joint sandwich variance of Proposition~\ref{prop:joint-si-ppi}. Coverage of the eventual-sample ATE lies in $[0.938, 0.960]$ for every method at every review date, inside the $R = 1{,}000$ band $[0.936, 0.964]$, and is omitted from the table for space. At $t = 6$ the labeled set is the full sample, so labeled-only and both PPI++ rules coincide exactly.}
\end{table}

In this calendar design, PPI++ achieves comparable interval width approximately half a week earlier. At the week-2.5 review its interval has mean width $0.254$, which labeled-only reaches at about week $2.97$; at the week-3 review it is at $0.192$, which labeled-only reaches at about week $3.69$. A fixed-$n$ efficiency calculation overstates the saving because waiting enlarges the sample as well as maturing it. The surrogate index is the opposite case: its width at the week-2.5 review, $0.122$, is already narrower than labeled-only's at full maturation in week six, $0.124$, so a valid surrogate index saves the full three and a half weeks. Neither PPI++ rule reaches labeled-only's week-six width before week six. At the week-2.5 review PPI++ rejects $\tau = 0$ on 61.2\% of replications, labeled-only on 37.3\%, and SI on 99.7\%. All of this assumes the stationarity the design imposes; Section~\ref{sec:results-enrollment} relaxes it.

Figure~\ref{fig:hero} summarizes the pattern on DGPs~1, 2, and~4: RMSE and RE against the labeled fraction under valid surrogacy (Panels~A and~D), and coverage against the direct-effect share $\rho$ and the drift $\Delta_\beta$ (Panels~B and~C).

\begin{figure}[!htbp]
  \centering
  \includegraphics[width=0.95\textwidth]{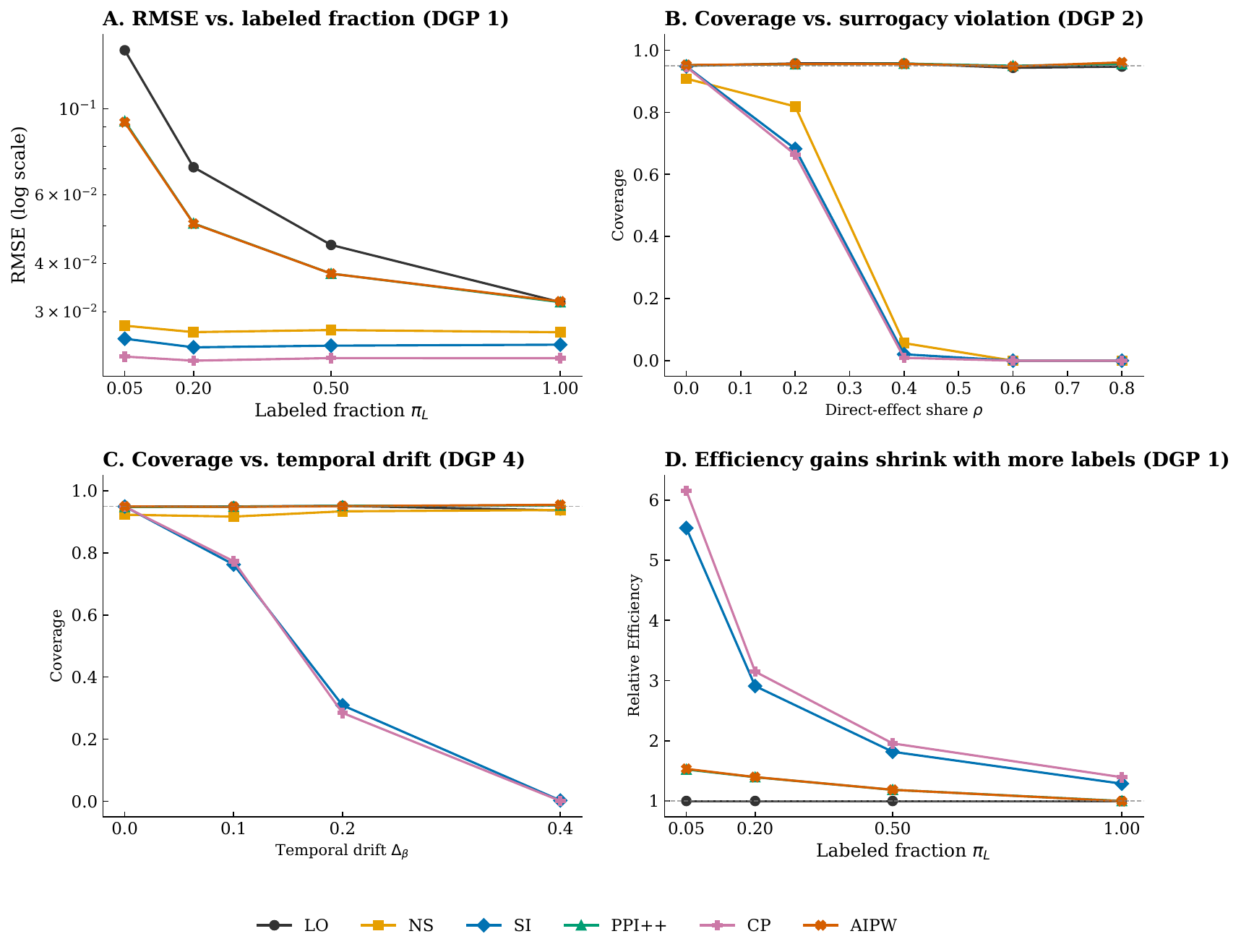}
  \caption{Representative simulation results (DGPs~1, 2, and~4). The hybrid is not plotted.}
  \label{fig:hero}
  \smallskip
  {\scriptsize \noindent Notes: surrogate-based methods (SI, CP) deliver low RMSE (Panel~A, log scale) and large RE (Panel~D) under valid surrogacy but lose coverage under partial mediation (Panel~B) and, fit on a historical cohort, under drift in the $Y$--$S$ slope (Panel~C; at $\Delta_\beta = 0.4$ SI covers $0.003$ and CP $0.001$, labeled-only $0.937$ and PPI++ $0.953$, Table~\ref{tab:appendix-dgp4}). PPI-family methods gain less and stay near nominal throughout. The $\pi_L$ axis has ticks at $0.05$, $0.20$, $0.50$, and $1.00$.}
\end{figure}

\subsection{Coverage Collapse Under Violations}
\label{sec:results-violations}

The robustness picture changes under surrogate violations (Table~\ref{tab:dgp2}).

\begin{table}[!ht]
\centering
\caption{DGP 2 (Partial Mediation): method performance as the direct-effect share $\rho$ increases.}
\label{tab:dgp2}
\footnotesize
\setlength{\tabcolsep}{4pt}
\begin{tabular}{@{}llrrrrc@{}}
\toprule
\textbf{Method} & $\boldsymbol{\rho}$ & $\boldsymbol{\tau}$ & \textbf{Bias} & \textbf{RMSE} & \textbf{Coverage} & \textbf{RE} \\
\midrule
Labeled-Only     & 0.0 & $0.150$ & $+0.001$ & $0.071$ & $0.950$ & $1.00$ \\
Naive Surrogate  & 0.0 & $0.150$ & $+0.012$ & $0.028$ & $\mathbf{0.908}$ & --- \\
Surrogate Index  & 0.0 & $0.150$ & $-0.0003$ & $0.025$ & $0.948$ & $2.77$ \\
PPI++            & 0.0 & $0.150$ & $-0.0001$ & $0.051$ & $0.952$ & $1.38$ \\
Composite Proxy  & 0.0 & $0.150$ & $+0.002$ & $0.023$ & $0.947$ & $3.01$ \\
AIPW             & 0.0 & $0.150$ & $-0.0001$ & $0.051$ & $0.953$ & $1.38$ \\
\midrule
Labeled-Only     & 0.2 & $0.188$ & $+0.003$ & $0.070$ & $0.958$ & $1.00$ \\
Naive Surrogate  & 0.2 & $0.188$ & $-0.024$ & $0.034$ & $\mathbf{0.819}$ & --- \\
Surrogate Index  & 0.2 & $0.188$ & $-0.036$ & $0.044$ & $\mathbf{0.682}$ & --- \\
PPI++            & 0.2 & $0.188$ & $+0.002$ & $0.051$ & $0.955$ & $1.38$ \\
Composite Proxy  & 0.2 & $0.188$ & $-0.034$ & $0.041$ & $\mathbf{0.664}$ & --- \\
AIPW             & 0.2 & $0.188$ & $+0.002$ & $0.051$ & $0.956$ & $1.38$ \\
\midrule
Labeled-Only     & 0.4 & $0.250$ & $-0.001$ & $0.069$ & $0.958$ & $1.00$ \\
Naive Surrogate  & 0.4 & $0.250$ & $-0.087$ & $0.091$ & $\mathbf{0.057}$ & --- \\
Surrogate Index  & 0.4 & $0.250$ & $-0.099$ & $0.102$ & $\mathbf{0.021}$ & --- \\
PPI++            & 0.4 & $0.250$ & $+0.0000$ & $0.050$ & $0.957$ & $1.38$ \\
Composite Proxy  & 0.4 & $0.250$ & $-0.098$ & $0.100$ & $\mathbf{0.009}$ & --- \\
AIPW             & 0.4 & $0.250$ & $-0.0001$ & $0.050$ & $0.956$ & $1.38$ \\
\midrule
Labeled-Only     & 0.6 & $0.375$ & $+0.001$ & $0.072$ & $0.944$ & $1.00$ \\
Naive Surrogate  & 0.6 & $0.375$ & $-0.213$ & $0.214$ & $\mathbf{0.000}$ & --- \\
Surrogate Index  & 0.6 & $0.375$ & $-0.225$ & $0.226$ & $\mathbf{0.000}$ & --- \\
PPI++            & 0.6 & $0.375$ & $+0.0001$ & $0.051$ & $0.950$ & $1.41$ \\
Composite Proxy  & 0.6 & $0.375$ & $-0.223$ & $0.224$ & $\mathbf{0.000}$ & --- \\
AIPW             & 0.6 & $0.375$ & $-0.0000$ & $0.051$ & $0.948$ & $1.41$ \\
\midrule
Labeled-Only     & 0.8 & $0.750$ & $-0.001$ & $0.071$ & $0.948$ & $1.00$ \\
Naive Surrogate  & 0.8 & $0.750$ & $-0.586$ & $0.586$ & $\mathbf{0.000}$ & --- \\
Surrogate Index  & 0.8 & $0.750$ & $-0.597$ & $0.598$ & $\mathbf{0.000}$ & --- \\
PPI++            & 0.8 & $0.750$ & $-0.0001$ & $0.051$ & $0.956$ & $1.39$ \\
Composite Proxy  & 0.8 & $0.750$ & $-0.598$ & $0.598$ & $\mathbf{0.000}$ & --- \\
AIPW             & 0.8 & $0.750$ & $-0.001$ & $0.051$ & $\mathbf{0.961}$ & $1.39$ \\
\bottomrule
\end{tabular}
\smallskip

{\scriptsize \noindent Notes: $\pi_L = 0.20$, $n = 10{,}000$, $R = 2{,}000$. Conventions as in Table~\ref{tab:dgp1}; RE is also suppressed (---) for methods with absolute relative bias above 10\%.}
\end{table}

SI coverage collapses under partial mediation: 68.2\% at $\rho = 0.2$ and exactly 0\% at $\rho \geq 0.6$. PPI++ and labeled-only stay inside the Monte Carlo band at every $\rho$ (PPI++ at 95.0\%--95.7\%); among the safe methods only AIPW at $\rho = 0.8$, at 96.1\%, falls outside it, one thousandth above the upper edge. The composite proxy and the naive surrogate, whose calibrated estimate $\hat\beta_{YS}\hat\tau_S$ also omits the direct effect, track the surrogate index into collapse; the naive surrogate falls from $0.908$ at $\rho = 0$ to $0.057$ at $\rho = 0.4$. In DGP~4 the historical surrogate index covers 30.9\% at $\Delta_\beta = 0.2$ while labeled-only, PPI++, and AIPW hold their bands (Table~\ref{tab:appendix-dgp4}). What fails there is the transport of the historical fit, not surrogacy, and because the two sides differ in training data as well as estimator, the comparison says nothing about the PPI correction alone under drift. At $\rho = 0.2$ the plug-in-$\lambda$ ablation has RE $1.360$ against the primary $1.378$, with coverage inside the band (Table~\ref{tab:ppi-ablation}), so neither result turns on the tuning rule.

Figure~\ref{fig:coverage} shows the pattern across representative configurations, including the two enrollment-drift columns of DGP~11. SI and CP fall deep below nominal under DGPs~2--4 and~9. In the selected DGP~1--10 configurations displayed in the heatmap, all labeled-only, PPI++, and AIPW cells lie within their Monte Carlo bands. The DGP~11 effect-drift column shows the limit of that robustness: labeled-only, PPI++, and AIPW fall to 0.57--0.68, because no within-experiment estimator can recover an effect that trends with enrollment time (Section~\ref{sec:results-enrollment}).

\begin{figure}[!ht]
  \centering
  \includegraphics[width=0.86\textwidth]{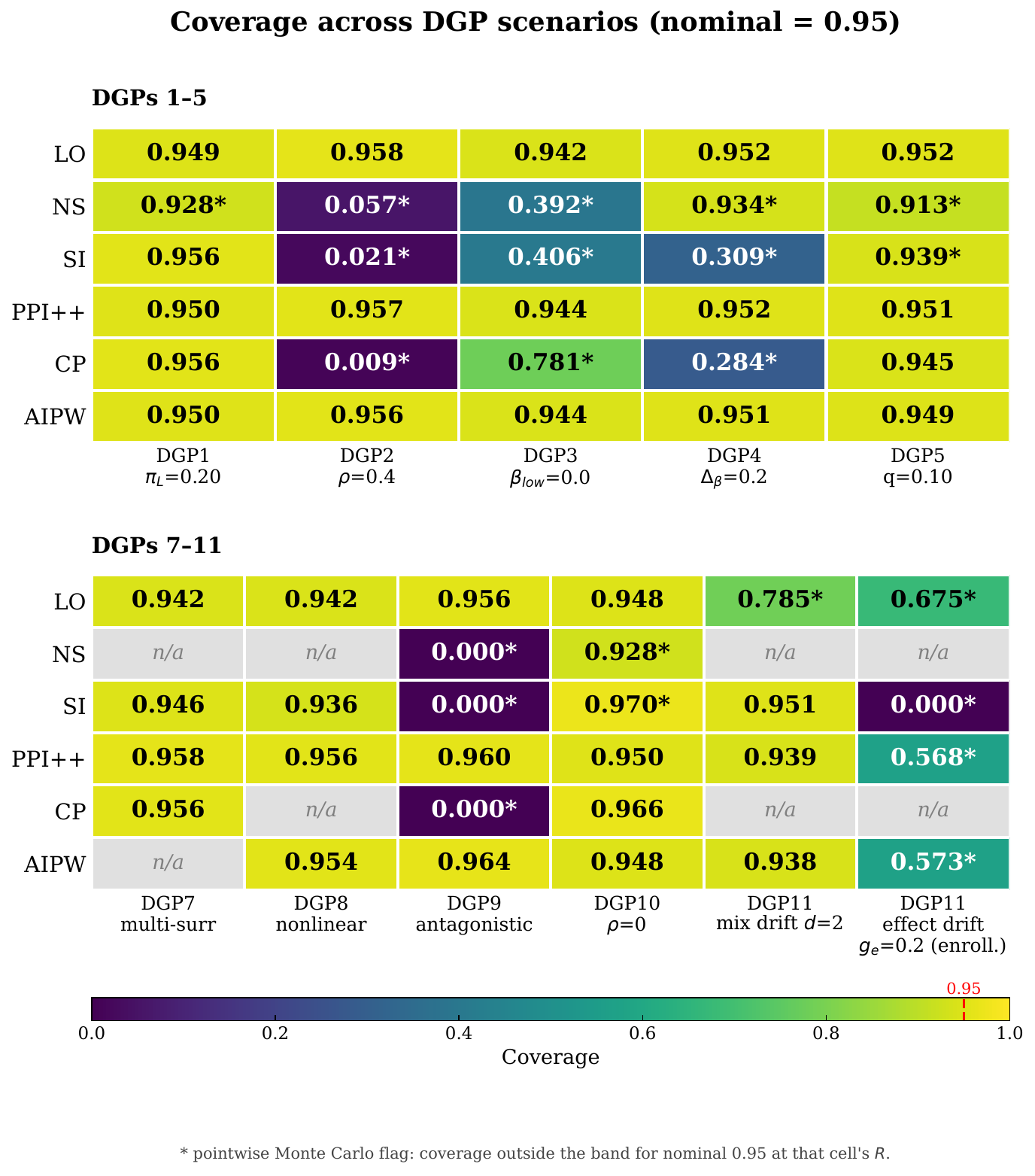}
  \caption{Coverage heatmap across selected DGP configurations (DGPs 1--5, 7--10, and the two DGP~11 drift columns).}
  \label{fig:coverage}
  \smallskip
  {\scriptsize \noindent Notes: eleven configurations in two panels (DGPs~1--5 above, DGPs~7--11 below) with the six primary methods as rows; ``n/a'' marks a method not run in that configuration. An asterisk marks a cell outside its Monte Carlo band: $[0.940, 0.960]$ at $R = 2{,}000$ (DGPs~1--5), $[0.931, 0.969]$ at $R = 500$ (DGPs~7--10), and $[0.936, 0.964]$ at $R = 1{,}000$ (DGP~11). Of the twenty-two flagged cells, eighteen are in the aggressive rows: the naive surrogate in seven columns (down to 0.057 at DGP~2 with $\rho = 0.4$ and 0.000 at DGP~9), the surrogate index in seven (0.021 at DGP~2, 0.406 at DGP~3, 0.309 at DGP~4, 0.939 at DGP~5, 0.000 at DGP~9, 0.970 at DGP~10, and 0.000 under DGP~11 effect drift), and the composite proxy in four (0.009, 0.781, 0.284, and 0.000 at DGPs~2, 3, 4, and~9). The other four are the safe methods under enrollment drift: labeled-only at 0.785 and 0.675, PPI++ at 0.568, and AIPW at 0.573. In the displayed DGP~1--10 configurations, all labeled-only, PPI++, and AIPW cells lie within their bands, the lowest being $0.942$ for labeled-only and $0.944$ for PPI++ and AIPW at DGP~3; elsewhere in DGPs~1--10, ten such cells fall outside their bands, all within $0.007$ of an edge, for example AIPW at $0.961$ (DGP~2, $\rho = 0.8$; Table~\ref{tab:dgp2}) and labeled-only at $0.937$ (DGP~4, $\Delta_\beta = 0.4$; Table~\ref{tab:appendix-dgp4}). The hybrid is not a row; its coverage is in Tables~\ref{tab:hybrid-sensitivity} and~\ref{tab:appendix-dgp9}.}
\end{figure}

We term this the \textit{robustness--efficiency tradeoff}. Figure~\ref{fig:frontier}, a benchmark summary, plots each method's mean RE over the valid-surrogate configurations of DGPs 1 and~5 against its mean coverage over the configurations of DGPs 2, 3, 4, and~9 in which the surrogate-based methods' assumptions fail (surrogacy in DGPs 2, 3, and~9; transport of the historical index in DGP~4). The two axes average different collections of scenarios with equal weights, so the display summarizes where each method landed on this benchmark rather than tracing an attainable frontier. SI and CP occupy the high-efficiency, low-coverage lower right and PPI++ and AIPW the high-coverage, lower-efficiency upper left; no method in the benchmark reaches the upper-right corner.

\begin{figure}[!ht]
  \centering
  \includegraphics[width=0.75\textwidth]{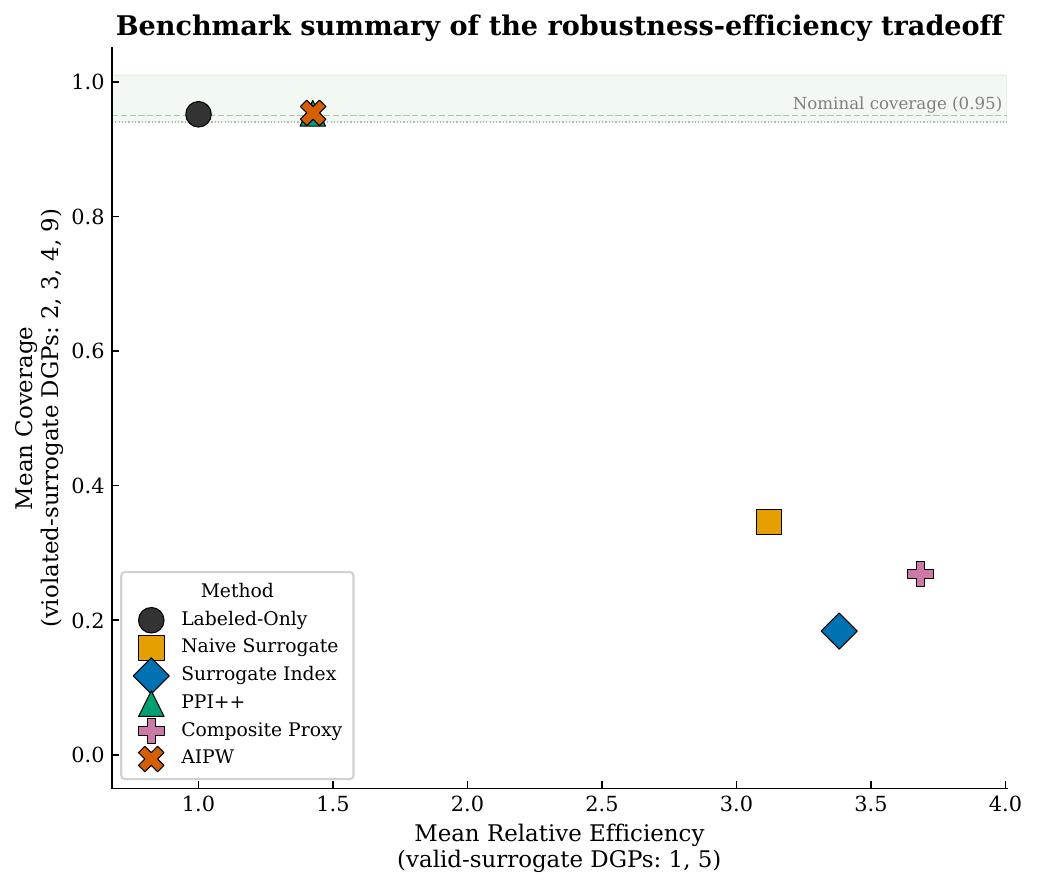}
  \caption{Benchmark summary of the robustness--efficiency tradeoff.}
  \label{fig:frontier}
  \smallskip

  {\scriptsize \noindent Notes: each point is a method-level aggregate, with the axes defined in the text. RE inputs: DGP~1 ($\pi_L = 0.20$) and DGP~5 ($q = 0.10$, MCAR); coverage inputs: DGP~2 ($\rho = 0.4$) and DGP~4 ($\Delta_\beta = 0.2$), both at $\pi_L = 0.20$, DGP~3 ($\beta_{YS}^{\text{low}} = 0$, $\pi_L = 0.10$), and DGP~9 ($\pi_L = 0.20$). Labeled-only sits at $(1.000, 0.952)$, PPI++ at $(1.425, 0.953)$, AIPW at $(1.426, 0.954)$, the naive surrogate at $(3.121, 0.346)$, the surrogate index at $(3.382, 0.184)$, and the composite proxy at $(3.683, 0.269)$. The hybrid is not plotted.}
\end{figure}

\subsection{The Detection-Damage Gap}
\label{sec:detection-damage}

Overlaying the coverage results on the diagnostic's power gives the paper's central practical warning. Figure~\ref{fig:detection-damage} plots, on a common $\rho$ axis, SI coverage (falling), two-sided diagnostic power (rising), and PPI++ coverage (flat at nominal). Between the $\rho$ at which SI coverage first falls below 90\% and the $\rho$ at which power first exceeds 50\% lies a band of violations severe enough to invalidate SI inference yet too small for the diagnostic to flag on most datasets. At $n = 10{,}000$ and $\pi_L = 0.20$ this detection-damage gap spans $\rho \in [0.09, 0.37]$: in the paired $R = 500$ run that locates it, a violation at $\rho = 0.2$ leaves SI coverage at 68.0\% (68.2\% at $R = 2{,}000$, Table~\ref{tab:dgp2}) while the diagnostic flags it on 13.6\% of replications at $\alpha = 0.05$ and 19.2\% at $\alpha = 0.10$.

Increasing the sample size relocates the gap rather than closing it. At $n = 100{,}000$ the band is $\rho \in [0.03, 0.16]$, its width narrowing from $0.277$ to $0.130$: detection improves at any fixed $\rho$, but the coverage-destroying threshold falls in tandem, because damage and detection are both driven by the ratio of bias to a shrinking standard error. The rate is best read in the direct effect $\delta = 0.15\,\rho/(1-\rho)$, since $\rho$ is a nonlinear reparameterization of $\delta$ (Section~\ref{sec:dgps}). With 95\% intervals from a bootstrap over replications, conditional on linear interpolation between grid points (Table~\ref{tab:edge-ci}), the detection edge is $0.089$ $[0.085, 0.092]$ at $n = 10{,}000$ and $0.029$ $[0.027, 0.030]$ at $n = 100{,}000$, a shrink factor of $3.10$ $[2.92, 3.36]$ for the tenfold increase in $n$ ($3.06$ $[2.77, 3.36]$ at $\alpha = 0.10$). On the refined local grid the damage edge is $0.0156$ $[0.0125, 0.0182]$ and $0.0046$ $[0.0043, 0.0057]$, a shrink factor of $3.37$ $[2.26, 4.02]$. Both intervals contain $\sqrt{10} = 3.16$, so both ratios are consistent with root-$n$ scaling over the studied range, the damage-edge ratio more weakly given the location caveat below. On the $\rho$ scale the detection factor is only $2.32$ $[2.22, 2.48]$, as the nonlinear map from $\delta$ to $\rho$ implies; $\rho$-scale factors are not evidence about the rate.
The damage edge is read from a refined local grid around each crossing (Appendix~\ref{app:threshold-sensitivity}). At $n = 100{,}000$ it is $0.030$ under both linear and pchip interpolation, with conditional interval $[0.028, 0.036]$, and across grids, interpolants, and an independent coarser run the estimates span the grid-resolution bracket $[0.018, 0.036]$. At $n = 10{,}000$ it is $0.094$, but SI coverage stays within about 1.5 Monte Carlo standard errors of $0.90$ from $\rho = 0.080$ to $0.105$, so that edge is located only to roughly $[0.08, 0.11]$.

The band is nonempty at both sample sizes under DGP~2, the one design in which we locate both edges. We do not claim it is nonempty in every design: both edges are $O(n^{-1/2})$ in $\delta$ units, but whether the interval between them is empty at a given $n$ depends on constants set by the estimators' variances and by how sharply coverage falls with the standardized bias. The band is the surrogate-metrics instance of a general property of pretest inference: biases of the order of the sampling noise can neither be detected with high power nor absorbed by the intervals \citep{leeb2005model, guggenberger2010hausman}, and uniform validity requires widening the intervals rather than testing \citep{berk2013valid}. Our contribution is its location and width in an interpretable violation unit at realistic sample sizes. Hillstrom and Criteo, whose effect sizes, outcome rates, and variances differ from DGP~2's, show damage with rare detection and damage with detection (Section~\ref{sec:empirical}). We have not mapped the band for other violation types.

The three base curves come from one paired run, on the same Monte Carlo draws and fitted prediction model, with $\mathrm{SE}(\hat{D})$ from the sandwich of Proposition~\ref{prop:joint-si-ppi}. That run's size at $\alpha = 0.05$ is $0.030$ (Monte Carlo standard error $0.008$) at $n = 10{,}000$, below its band, and $0.047$ ($0.012$) at $n = 100{,}000$; an undersized test moves the detection edge to the right, so the upper edge at $n = 10{,}000$ is, if anything, overstated, and the nominal-size evidence is the $R = 2{,}000$ run of Section~\ref{sec:surrogacy-test}. Raising the level to $0.10$ moves the detection edge in to $\rho = 0.32$ at $n = 10{,}000$ and $0.13$ at $n = 100{,}000$, narrowing the band without closing it. Across damage thresholds of $0.85$ and $0.90$, detection thresholds of $0.50$ and $0.80$, both sample sizes, and both levels, the band is nonempty with width $0.09$ to $0.36$ in $\rho$ (Table~\ref{tab:detection-damage-sensitivity}), and replacing the joint sandwich by the delta-method or fixed-predictor variance moves the edges by at most $0.003$ in $\rho$.

\begin{figure}[!ht]
  \centering
  \includegraphics[width=0.98\textwidth]{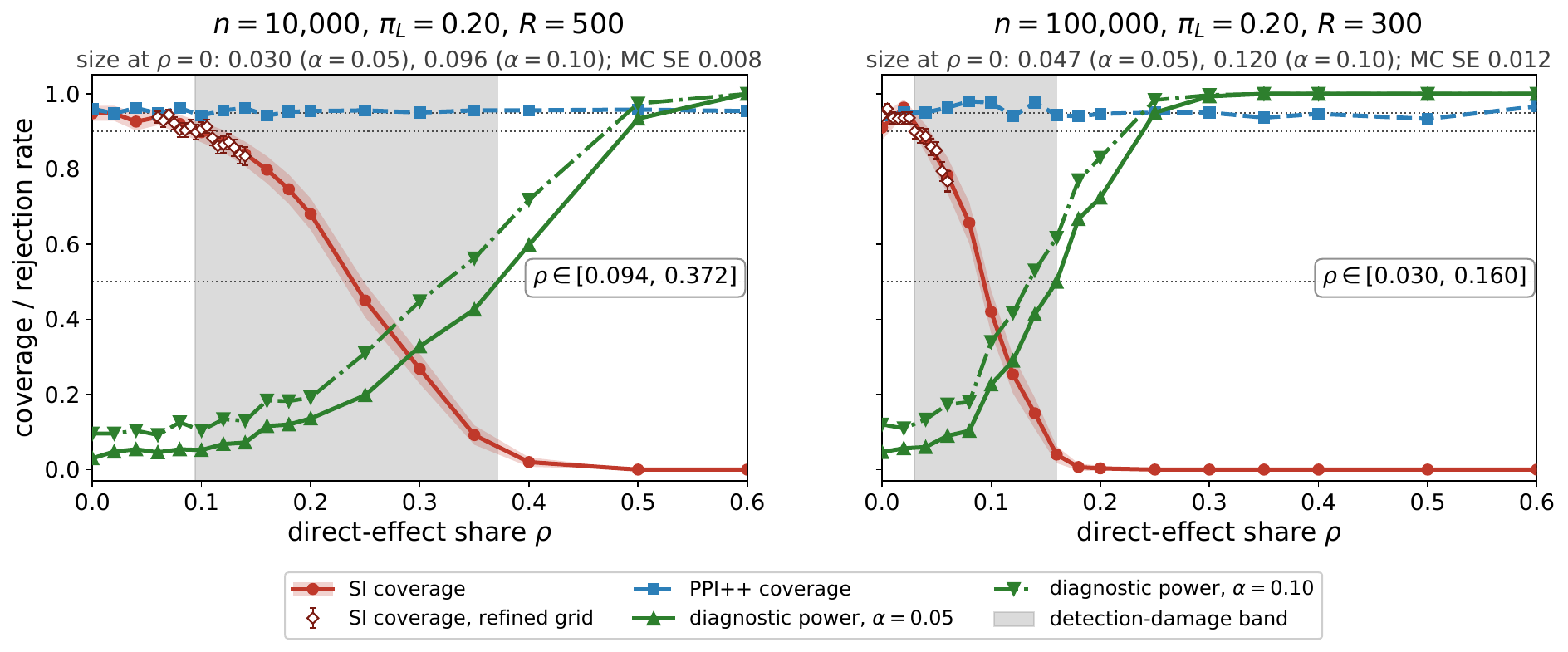}
  \caption{The detection-damage gap under DGP~2 at $\pi_L = 0.20$. Left: $n = 10{,}000$. Right: $n = 100{,}000$.}
  \label{fig:detection-damage}
  \smallskip

  {\scriptsize \noindent Notes: the base curves come from one paired run ($R = 500$ at $n = 10{,}000$, $R = 300$ at $n = 100{,}000$): SI coverage under the joint sandwich variance (red circles, with a ribbon of $\pm 2$ Monte Carlo standard errors), PPI++ coverage with the exact variance (blue squares), and two-sided diagnostic power (green triangles, solid at $\alpha = 0.05$, dot-dashed at $\alpha = 0.10$); each panel's subtitle gives the size at $\rho = 0$. The grey block is the band at $\alpha = 0.05$, from the $\rho$ at which SI coverage falls below 90\% to the $\rho$ at which power exceeds 50\%, with printed edges $[0.094, 0.372]$ and $[0.030, 0.160]$. Open diamonds ($\pm 2$ Monte Carlo standard errors) are SI coverage from a separate refined-grid run ($\rho$ in steps of $0.005$; $R = 1{,}000$, and $R = 2{,}000$ at $\rho = 0$) that extends the base run's draws at shared $\rho$; the printed damage edges come from that run and the detection edges from the base run. The Monte Carlo intervals for the edges, and their bootstrap construction, are in Appendix~\ref{app:threshold-sensitivity}; Table~\ref{tab:detection-damage-sensitivity} varies the thresholds and the level.}
\end{figure}

\subsection{Enrollment-Time Labeling}
\label{sec:results-enrollment}

Every result so far draws the labeled set as a uniform random subset, but the maturation process of Figure~\ref{fig:timeline} labels the earliest-enrolled cohort. DGP~11 tests that mechanism directly, and Table~\ref{tab:enrollment} reports the three regimes.

With no drift, enrollment-time labeling behaves like MCAR: every point estimate behaves as in DGP~1 and every coverage entry sits inside the $R = 1{,}000$ band $[0.936, 0.964]$. Under mix drift the labeling is no longer MCAR, and labeled-only, which estimates the early cohort's effect, is badly biased for the full-sample ATE ($-23\%$ at $d = 1$, $-47\%$ at $d = 2$, where its coverage falls to 79\%). The model-based estimators survive because the conditional outcome model is stable across cohorts: SI stays unbiased ($+0.3\%$ at $d = 2$, coverage $0.951$) and AIPW remains approximately unbiased, while PPI++ exhibits a small bias ($-0.4\%$ at $d = 1$, $-5.0\%$ at $d = 2$) from estimating its rectifier on the shifted cohort. Its exact variance estimator also underestimates the sampling variance here, at $0.873$ of the Monte Carlo variance at $d = 2$ (Appendix~\ref{app:ablations}), so its coverage loss is not bias alone. This is the one design in which the tuning rule matters: the plug-in rule of Equation~\eqref{eq:lambda-opt} shrinks the coefficient toward the biased labeled cohort and carries $-11.9\%$ relative bias at $d = 2$, with coverage $0.923$ (Table~\ref{tab:ppi-ablation}). Labeling by maturity breaks the simplest estimator first and the prediction-based estimators last; this is MAR given $X$, repairable in principle by reweighting the labeled cohort on observables.

Effect drift is the failure the framing cannot repair. When the treatment effect changes over the enrollment window through a channel that bypasses the surrogate, the early cohort understates the full-sample effect and nothing estimated from it recovers the gap: at $g_e = 0.2$, labeled-only, PPI++, and AIPW all carry about $-34\%$ bias with coverage of 68\%, 57\%, and 57\%, and SI is worse ($-40\%$, coverage 0\%). Under a random split the same DGP leaves labeled-only, PPI++, and AIPW at nominal coverage (only SI breaks, since effect drift is also a surrogacy violation), so the damage comes from labeling by enrollment time. Choosing the review date chooses $\pi_L$ only when the treatment effect is stable over the enrollment window or drifts only through observables; otherwise the review date is confounded with the estimand, and a mid-experiment review is an early read on the early cohort, not an estimate of the eventual full-sample effect.

\begin{table}[!ht]
\centering
\caption{DGP 11: enrollment-time labeling under drift ($n = 10{,}000$, $\pi_L = 0.20$, $R = 1{,}000$).}
\label{tab:enrollment}
\small
\setlength{\tabcolsep}{4pt}
\begin{tabular}{@{}llrrrr@{}}
\toprule
\textbf{Config} & \textbf{Method} & \textbf{Bias} & \textbf{Rel.\ bias} & \textbf{Coverage} & \textbf{RE} \\
\midrule
Enrollment, no drift & LO & $+0.0041$ & $+2.8\%$ & $0.953$ & $1.00$ \\
 & SI & $+0.0003$ & $+0.2\%$ & $0.944$ & $3.29$ \\
 & PPI++ & $+0.0039$ & $+2.6\%$ & $0.963$ & $1.26$ \\
 & AIPW & $+0.0039$ & $+2.6\%$ & $0.963$ & $1.26$ \\
\midrule
Mix drift, $d = 2$ & LO & $-0.0701$ & $-46.7\%$ & $\mathbf{0.785}$ & $1.00$ \\
 & SI & $+0.0005$ & $+0.3\%$ & $0.951$ & $4.65$ \\
 & PPI++ & $-0.0075$ & $-5.0\%$ & $0.939$ & $1.81$ \\
 & AIPW & $+0.0013$ & $+0.9\%$ & $0.938$ & $1.68$ \\
\midrule
Effect drift, $g_e = 0.2$, MCAR split & LO & $-0.0056$ & $-2.3\%$ & $0.952$ & $1.00$ \\
 & SI & $-0.0987$ & $-39.5\%$ & $\mathbf{0.000}$ & $0.59$ \\
 & PPI++ & $-0.0048$ & $-1.9\%$ & $0.952$ & $1.24$ \\
 & AIPW & $-0.0048$ & $-1.9\%$ & $0.952$ & $1.24$ \\
\midrule
Effect drift, $g_e = 0.2$, enrollment & LO & $-0.0855$ & $-34.2\%$ & $\mathbf{0.675}$ & $1.00$ \\
 & SI & $-0.1000$ & $-40.0\%$ & $\mathbf{0.000}$ & $1.03$ \\
 & PPI++ & $-0.0841$ & $-33.6\%$ & $\mathbf{0.568}$ & $1.08$ \\
 & AIPW & $-0.0841$ & $-33.7\%$ & $\mathbf{0.573}$ & $1.08$ \\
\bottomrule
\end{tabular}
\smallskip

{\scriptsize \noindent Notes: all-units five-fold cross-fitting. The labeled set is the earliest-enrolled 20\% except in the MCAR-split row. True $\tau = 0.15$ without drift and under mix drift, and $0.25$ under effect drift; bias is printed to four decimals because the smaller entries fall below the third. \textbf{Bold} coverage entries fall outside the band $[0.936, 0.964]$. PPI++ is the primary configuration, SI uses the joint sandwich variance, and AIPW a logistic propensity in $S$; the naive surrogate and the composite proxy are not run. RE is reported for every row, so beside a double-digit relative bias it is a ratio of RMSEs against a biased baseline, not an efficiency gain. The $d = 1$, MCAR-baseline, $g_e = 0.4$, and combined-drift configurations are in the replication package, and the PPI++ tuning-rule ablations on these draws in Table~\ref{tab:ppi-ablation}.}
\end{table}

\subsection{Portfolio Decision-Making}
\label{sec:results-decisions}

DGP~6 shows surrogate methods cutting cumulative regret sharply: against the labeled-only mean regret of $3.321$, the reduction is 57\% for the naive surrogate, 52\% for the surrogate index, and 12\% for PPI++ (Table~\ref{tab:portfolio}). The composite proxy has no past experiments here and falls back to the surrogate index, so its row, labeled ``CP (SI fallback)'', repeats the SI row and is not independent evidence. The $K = 50$ and $K = 100$ portfolios and the $\pi_L = 1.00$ comparison are in the replication package (Appendix~\ref{app:full-results}).

\begin{table}[!ht]
\centering
\caption{DGP 6 (Portfolio): Cumulative regret across $K = 200$ experiments.}
\label{tab:portfolio}
\small
\begin{tabular}{@{}lrrr@{}}
\toprule
\textbf{Method} & \textbf{Mean Regret} & \textbf{MC SE} & \textbf{Relative Regret (\%)} \\
\midrule
Labeled-Only & $3.321$ & $0.025$ & $83.5$ \\
Naive Surrogate & $1.432$ & $0.013$ & $36.0$ \\
Surrogate Index & $1.610$ & $0.014$ & $40.5$ \\
PPI++ & $2.926$ & $0.023$ & $73.6$ \\
CP (SI fallback) & $1.610$ & $0.014$ & $40.5$ \\
AIPW & $2.926$ & $0.022$ & $73.6$ \\
\bottomrule
\end{tabular}
\smallskip

{\scriptsize \noindent Notes: $\pi_L = 0.20$, $R = 500$ portfolios, all-units five-fold cross-fitting, mean oracle gain $V^* = 3.978$ (Monte Carlo standard error $0.029$). PPI++ is the primary configuration. CP (SI fallback): without past experiments the composite proxy returns the surrogate index (Method~4). Relative regret is the mean regret over the mean oracle gain, in percent, with Monte Carlo standard errors of $0.36$ to $0.39$ points. No coverage is reported, so the bolding convention does not apply.}
\end{table}

Regret here is effect-size-weighted: a decision is charged the size of the effect it forgoes or destroys, $V^* = \sum_k \max(\tau_k, 0)$ against the realized payoff, and launching a change whose true effect is zero is free. That favors estimators that launch readily, as the naive surrogate and the surrogate index do. Neither of two departures changes the ranking of the four methods (Table~\ref{tab:portfolio-sensitivity}, an $R = 200$ run whose cost-free column agrees with Table~\ref{tab:portfolio} within Monte Carlo error). A launch cost $c \in \{0, 0.01, 0.02\}$ leaves the ordering intact and, if anything, helps the surrogate methods slightly, because their excess launches are mostly small positives rather than nulls. Seeding the portfolio with 10\% sign-reversing experiments also leaves the ordering intact, but it raises the surrogate index's relative regret from 40.1\% to 54.9\% and the naive surrogate's from 36.0\% to 51.1\%, while PPI++ is unmoved at 72.8\% and 72.4\%; the surrogate methods' harmful launches per portfolio rise roughly fourteenfold, from 0.23 to 3.16 and 3.25, against PPI++'s 0.41 falling to 0.35. Sign-reversing experiments increase the frequency of launches with negative treatment effects.

The naive surrogate attains the lowest regret despite its biased point estimate because a launch decision depends on sign and significance, and its small bias is outweighed by the variance reduction from using all $n$ units. A platform whose objective resembles this one may use a surrogate estimator for go-or-no-go triage while relying on PPI++ for formal inference and effect-size reporting. The advantage requires sign preservation: in DGP~2 the direct effect reinforces the mediated effect, whereas in DGP~9 SI returns the wrong sign in 100\% of replications. The portfolio result is therefore a valid-surrogacy, sign-preserving result, not an endorsement of surrogate-only launch rules. Figure~\ref{fig:decisions} contrasts two decision-oriented criteria.

\begin{figure}[!ht]
  \centering
  \includegraphics[width=0.95\textwidth]{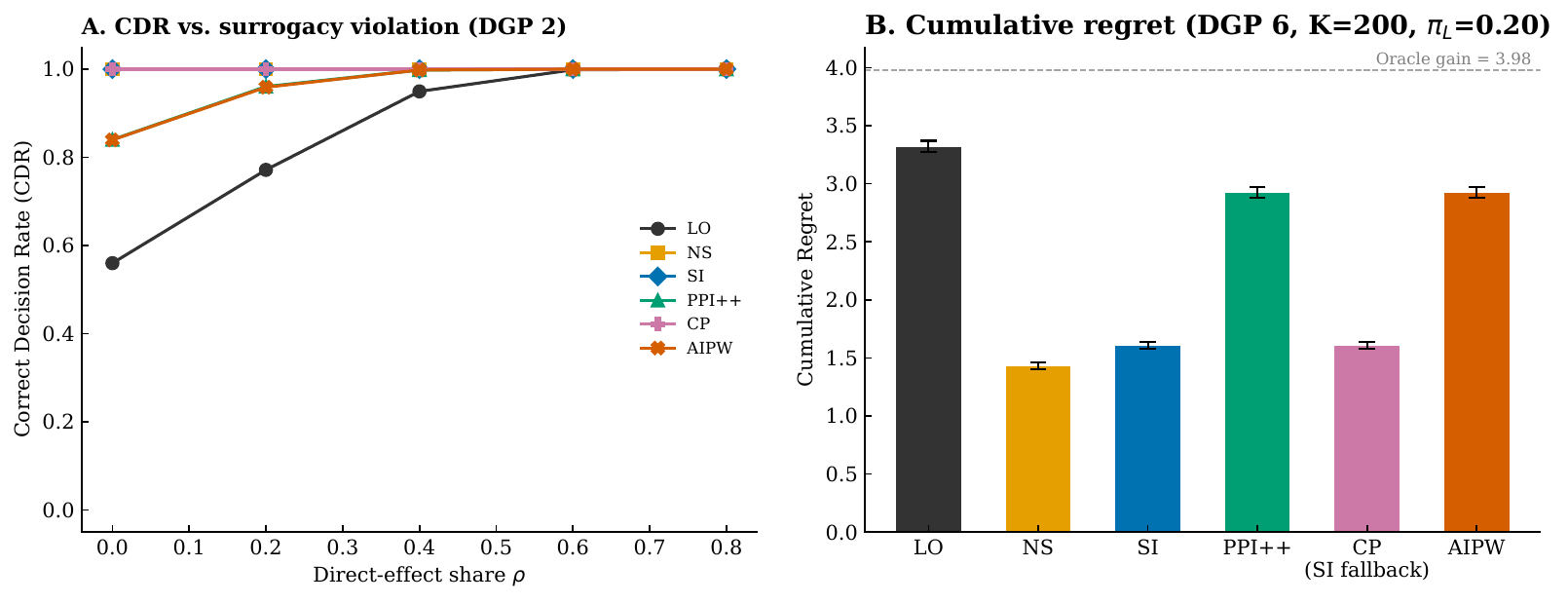}
  \caption{Decision-oriented performance. Panel~A: correct decision rate (CDR) versus direct-effect share $\rho$ in DGP~2 ($\pi_L = 0.20$, $n = 10{,}000$). Panel~B: cumulative regret in DGP~6 (portfolio of $K = 200$ small experiments, $\pi_L = 0.20$, mean oracle gain $3.98$).}
  \label{fig:decisions}
  \smallskip

  {\scriptsize \noindent Notes: error bars in panel~B are 95\% Monte Carlo intervals for the mean cumulative regret, $\pm 1.96$ Monte Carlo standard errors over $R = 500$ portfolios. CP (SI fallback) is the composite proxy, which has no historical experiments in this design and reverts to the surrogate index.}
\end{figure}

\subsection{The All-Units Plug-in Variance and the Overlap Covariance Correction}
\label{sec:overlap-correction}

A plug-in variance that treats the labeled and all-units averages as independent samples is exact only where the overlap correction $C$ vanishes, and whether that matters depends on the tuning rule. Table~\ref{tab:coverage-correction} runs both rules on the same draws. Block~A is the plug-in rule of Equation~\eqref{eq:lambda-opt} with the coefficient clipped to $[0,1]$. It sits below the per-arm optimum by the factor $1/(1 + \pi_L)$, so by Equation~\eqref{eq:overlap-piL} $\hat{C}$ is $0.7\%$, $2.1\%$, $6.6\%$, $22.1\%$, and $35.4\%$ of the plug-in variance at $\pi_L = 0.05$, $0.10$, $0.20$, $0.50$, and $0.80$. The undercoverage follows, because the omitted term grows while the leading term $\sigma_{Y,t}^2/n_{L,t}$ shrinks: the plug-in variance covers $94.3\%$ at $\pi_L = 0.20$, $92.6\%$ at $0.50$, and $89.8\%$ at $0.80$, and the correction restores $94.9\%$, $95.4\%$, and $94.8\%$, inside the band.

Block~B is the primary PPI++. Its optimum condition~\eqref{eq:lambda-exact-foc} and the correction vanish together only when the armwise labeled fractions or the per-arm optima coincide, or when $\lambda^* = 0$. In DGP~1 the arms share $(\gamma_t, \sigma^2_{\hat{Y},t})$ and have nearly equal labeled fractions, so the correction cancels numerically, to within $0.03\%$ of the plug-in variance, and coverage is $95.0\%$, $95.3\%$, and $94.5\%$ at the same three labeled fractions; $\hat{C}$ is nonetheless formed and reported in every case. At $\pi_L = 1.0$ the estimator returns labeled-only (Section~\ref{sec:eightmethods}) and $\hat{C}$ is never formed. The bootstrap is the third route and the most expensive: at $B = 200$ its percentile interval undercovers at four of the six labeled fractions, falling to $0.931$ at $\pi_L = 0.80$, and it needs $B \geq 500$ to match the analytic correction (Table~\ref{tab:bootstrap-sensitivity}) at roughly a hundred times the cost (Table~\ref{tab:timing}).

\begin{table}[!ht]
\centering
\caption{The overlap covariance correction under the two tuning rules (Proposition~\ref{prop:corrected-variance}).}
\label{tab:coverage-correction}
\small
\setlength{\tabcolsep}{5pt}
\begin{tabular}{@{}lcccc@{}}
\toprule
$\boldsymbol{\pi_L}$ & \textbf{Plug-in variance} & \textbf{Exact variance} & \textbf{Bootstrap ($B = 200$)} & $\boldsymbol{\hat{C} / \hat{V}_{\mathrm{plug\text{-}in}}}$ \\
\midrule
\multicolumn{5}{l}{\textit{Block A: the plug-in-rule estimator (plug-in tuning rule, clipped coefficient)}} \\
0.05 & $0.949$ & $0.949$ & $0.941$ & $0.007$ \\
0.10 & $0.940$ & $0.945$ & $\mathbf{0.933}$ & $0.021$ \\
0.20 & $0.943$ & $0.949$ & $\mathbf{0.939}$ & $0.066$ \\
0.50 & $\mathbf{0.926}$ & $0.954$ & $\mathbf{0.939}$ & $0.221$ \\
0.80 & $\mathbf{0.898}$ & $0.948$ & $\mathbf{0.931}$ & $0.354$ \\
1.00 & $0.948$ & $0.948$ & $0.940$ & --- \\
\midrule
\multicolumn{5}{l}{\textit{Block B: the primary PPI++ (exact tuning rule, common unclipped coefficient)}} \\
0.05 & $0.953$ & $0.953$ & --- & $0.000$ \\
0.10 & $0.944$ & $0.944$ & --- & $0.000$ \\
0.20 & $0.950$ & $0.950$ & --- & $0.000$ \\
0.50 & $0.953$ & $0.953$ & --- & $0.000$ \\
0.80 & $0.945$ & $0.945$ & --- & $0.000$ \\
1.00 & $0.948$ & $0.948$ & --- & --- \\
\bottomrule
\end{tabular}
\smallskip

{\scriptsize \noindent Notes: DGP 1, $n = 10{,}000$, $R = 2{,}000$, all-units five-fold cross-fitting. The blocks share draws and fitted prediction model, so they are paired cell by cell. Block~A is the configuration ``PPI++ (plug-in $\lambda$, plug-in variance)'' against ``PPI++ (plug-in $\lambda$, clipped)''; block~B is ``PPI++ (plug-in variance)'' against the primary ``PPI++''. The last column is the mean overlap correction over the mean plug-in variance. \textbf{Bold} coverage entries fall outside the band $[0.940, 0.960]$. Block~B has no bootstrap column because the bootstrap loop re-estimates $\lambda$ under the plug-in rule and clips it. At $\pi_L = 1.00$, $\hat{C}$ is never formed.}
\end{table}

\subsection{Robustness Checks}
\label{sec:results-gbt}

The core findings could depend on the single-surrogate linear DGPs; the checks below vary the number of surrogates, the functional form, the sign of the surrogate effect, and the learner.

With three surrogates (DGP~7) the surrogate index attains RE $5.36\times$ at $\pi_L = 0.05$ with coverage $0.952$, inside the $R = 500$ band $[0.931, 0.969]$, and PPI++ keeps its guarantee at RE $1.2$--$1.4\times$ (Table~\ref{tab:dgp7}). Under best-case historical calibration the composite proxy reaches $6.30\times$ in the same cell, against the surrogate index's $5.36\times$ (ESS multipliers $39.7$ and $28.7$). An independent $R = 500$ run with a separately drawn library of $30$ rather than $50$ past experiments puts the composite proxy at $7.01\times$ and the surrogate index at $5.46\times$ (Table~\ref{tab:appendix-dgp7}). Each run conditions on one realized library and does not average over the estimation uncertainty of that library, so the composite-proxy figures are best cases that depend on the library, not point estimates.

Under the concave link of DGP~8 the surrogate index attains RE $5.5\times$ at $\pi_L = 0.05$ with bias $0.002$ and coverage $0.934$, inside the $R = 500$ band, and PPI++ and AIPW keep valid coverage (Table~\ref{tab:appendix-dgp8}). Dropping $S^2$ from the prediction model under DGP~8, where surrogacy holds by construction, leaves the two-sided diagnostic rejecting at $\alpha = 0.05$ on 5.2\%, 5.4\%, and 5.7\% of replications at $\pi_L = 0.05$, $0.20$, and $0.50$ ($R = 1{,}000$, Monte Carlo standard error about 0.007; replication package): the diagnostic cannot see this misspecification. Its null remains joint, surrogacy together with correct specification, so a rejection is not attributable to surrogacy failure alone, and a non-rejection certifies neither component.

Under the antagonistic surrogate of DGP~9 the surrogate index estimates $-0.09$ against the true $0.41$, the wrong sign, with 0\% coverage; the diagnostic rejects in every replication and the hybrid reverts to PPI++ (Table~\ref{tab:appendix-dgp9}). Clipping the coefficient to $[0,1]$ binds here, but at the simulated effect sizes it moves RMSE by $0.03\%$ at $\pi_L = 0.05$ and leaves RE at $1.20$ (Table~\ref{tab:ppi-ablation}). In DGP~10 the relative surrogate-index bias tracks approximately $-\rho$, at $-0.8\%$, $-19.4\%$, and $-39.7\%$ for $\rho = 0$, $0.2$, and $0.4$ at $\pi_L = 0.20$ (Table~\ref{tab:appendix-dgp10}), so nonlinearity and partial mediation act approximately additively; the relation is approximate even under linearity, where the bias is $-\delta\kappa$ (Proposition~\ref{prop:si-bias}). PPI++ remains unbiased throughout; the hybrid was not run on this DGP.

\paragraph{Bias versus omitted interval uncertainty.} A surrogate index can miss the truth because a surrogacy violation biases its estimand, because the prediction model is misspecified, or because its interval omits the first-stage uncertainty of Equation~\eqref{eq:si-first-stage}. Recomputing every SI interval with the narrower variances, on the same point estimates, isolates the third (Table~\ref{tab:si-variance-ablation}). In the linear Gaussian designs the omitted term is 0.4 to 7.2\% of the joint sandwich variance and moves coverage by at most 0.008; it does not rescue SI under violation (0.678 against 0.682 at $\rho = 0.2$, 0.019 against 0.021 at $\rho = 0.4$), so the collapse of Section~\ref{sec:results-violations} is bias, not a missing variance term. The Criteo-calibrated rich index of Table~\ref{tab:multisurrogate}, which adds a treatment-responsive engagement proxy to the visit indicator, is the exception: the first-stage term is 71\% of the sandwich variance at $m = 1$, where the rich index's relative bias is $-0.4\%$, and the plug-in interval covers $0.665$ against the sandwich's $0.960$, an interval defect; at $m = 0.5$ and $m = 0$ the sandwich still covers only $0.800$ and $0.414$, which is bias. That index is exactly rank-deficient ($S_1^2 = S_1$), and relative eigenvalue floors from $10^{-8}$ to $10^{-16}$ give the same first-stage share ($70.9\%$), variance ratio, coverage, and size to every reported digit, because $d_0$ and every design row are orthogonal to the discarded direction (replication package). Whenever the index leans on a predictor that treatment moves, the first-stage term should be carried; it does not restore coverage where the estimand itself is wrong.

Re-running representative DGPs with gradient-boosted trees leaves the method ordering unchanged (Table~\ref{tab:gbt}): the surrogate index keeps its efficiency under valid surrogacy (RE $5.59$ with either model at DGP~1, $\pi_L = 0.05$) and its collapse under violation (coverage $0.030$ with OLS and $0.010$ with GBT at DGP~2, $\rho = 0.4$, $\pi_L = 0.20$), while PPI++ stays within the band except under OLS at DGP~8 with $\pi_L = 0.50$ ($0.915$; GBT $0.920$, at the band edge). GBT loses a little efficiency rather than gaining any: across the surrogate-index, PPI++, and AIPW cells of the full comparison (replication package) its RE never exceeds the OLS RE by more than $0.02$ and falls as much as 14\% below it, so the efficiency gap between the two families is structural, not model-driven. Trees have no coefficient expansion, so the GBT surrogate-index interval uses the fixed-predictor plug-in variance and is heuristic; on two cells it is $0.95$ and $1.13$ of the Monte Carlo variance, and a tree-refitting bootstrap that keeps every copy of a resampled unit in its unit's fold agrees on the DGP~1 cell (Appendix~\ref{app:ablations}). The comparison's evidential weight is therefore on bias, RMSE, and RE.

\subsection{The Adaptive Hybrid Estimator}
\label{sec:results-hybrid}

A hard-switch rule would choose between SI and PPI++; the exploratory hybrid instead combines them continuously. It reports $w\,\hat{\tau}_{\mathrm{SI}} + (1-w)\,\hat{\tau}_{\mathrm{PPI++}}$ with the Cauchy-kernel weight $w = c/(c + T_n^2)$ at $c = 1.5$, so the weight slides from the surrogate index toward PPI++ as the estimator-disagreement statistic grows, and it pairs that point estimate with an interval calibrated by simulation from the sandwich $\hat{\Sigma}$ of Proposition~\ref{prop:joint-si-ppi} (Appendix~\ref{app:hybrid}).

In the canonical $R = 2{,}000$ hybrid evaluation, under valid surrogacy (DGP~1) the hybrid attains RE $2.0$--$2.2\times$ at $\pi_L = 0.05$--$0.20$ with 94--96\% coverage, against the surrogate index's $2.9$--$5.5\times$ in the same cells (Table~\ref{tab:dgp1}), the gap being the efficiency cost of hedging. At the featured violation cell (DGP~2, $\rho = 0.2$, $\pi_L = 0.20$) it carries $-8.4\%$ relative bias at $0.934$ coverage, and under DGP~9 its weight is near zero (mean $0.012$ at $\pi_L = 0.20$), so it reverts to PPI++. On Hillstrom the diagnostic rejects on 11.4\% of splits, the mean weight stays at $0.64$, and the hybrid inherits 10.7\% relative bias at 93.0\% coverage (Table~\ref{tab:hillstrom-hybrid}). Appendix~\ref{app:hybrid} gives its limit-experiment risk (Proposition~\ref{prop:limit-risk}), the diverging-$c_n$ variant, and the status of the interval, which inherits PPI++ validity under fixed detectable alternatives while we make no consistency claim at fixed $c$ (Remark~\ref{rem:ci-validity}).

\section{Public-Data Illustrations}
\label{sec:empirical}

Three public experiments, spanning 445 to 13{,}979{,}592 units, illustrate the simulation findings without proprietary data. On the Hillstrom email-marketing experiment ($n = 64{,}000$) the surrogate index fails while PPI++ remains valid, and, as in the detection-damage gap of Section~\ref{sec:detection-damage}, the violation damages the surrogate index while the diagnostic rarely flags it. The Criteo uplift experiment ($n = 13.98$ million) shows the same funnel violation at industrial scale, where the diagnostic detects it. Effect sizes, outcome rarity, and variances in both datasets differ from DGP~2's. A testbed calibrated to the Criteo funnel places the observed violation on a dial, and the LaLonde experiment ($n = 445$) is a negative control in which the ``surrogate'' is a pre-treatment covariate.

The experiments are fixed datasets, so we mask outcomes at random to create labeled-unlabeled splits and measure how well each estimator recovers the full-sample difference in means, which we call the masking target. Coverage against it is a recovery diagnostic, not repeated-sampling coverage of the causal ATE: the target is itself an estimate that the labeled subsample overlaps, which makes labeled-only and PPI++ coverage conservative. Every bias and coverage figure in Sections~\ref{sec:hillstrom}, \ref{sec:criteo}, and~\ref{sec:lalonde} is measured against the masking target; the semi-synthetic testbed of Section~\ref{sec:semisynthetic}, whose true $\tau$ is known, supplies known-truth validation. Random masking makes MCAR hold by construction, so these analyses illustrate estimator behavior on real outcome distributions rather than validating the labeling mechanism.

\subsection{Hillstrom email marketing experiment}
\label{sec:hillstrom}

This instantiates the e-commerce example as a conversion funnel. The Hillstrom data \citep{hillstrom2008minethatdata}, accessed via scikit-uplift, cover 64{,}000 customers randomized in equal thirds to a men's merchandise campaign, a women's merchandise campaign, or no email. We pool the two campaigns (treatment $n_1 = 42{,}694$; control $n_0 = 21{,}306$), so the estimand is the effect of receiving either campaign. The surrogate is a binary website visit within two weeks and the outcome is conversion (a purchase); both are recorded over the same two-week window, so the data exercise the funnel structure rather than calendar maturation. Visits are roughly sixteen to nineteen times as common as conversions, and every conversion follows a visit, a fact about how events are recorded rather than identification of a mediation mechanism. Treatment also raises the conversion rate among visitors, from 5.4\% to 6.4\% (Table~\ref{tab:hillstrom-descriptive}), and the within-visitor shift appears in each campaign ($+1.46$ points for the men's and $+0.44$ for the women's). Adjusting for the covariates (recency, history, mens, womens, newbie, and the channel and zip-code dummies) in a logistic regression among the 9{,}394 visitors, with HC1 standard errors, gives a treatment coefficient of $+0.186$ (SE $0.105$, $p = 0.076$, odds ratio $1.20$) and an average marginal effect of $+1.03$ points (SE $0.55$ points, $p = 0.064$). The violation survives conditioning on $X$ with its size intact but is not significant at conventional levels, so Hillstrom suggests conditional-mean noninvariance rather than establishing it; Criteo establishes it (Section~\ref{sec:criteo}). The estimator behavior is not in doubt: a prediction model trained on pooled labeled data estimates a single $\E[Y|S=1]$ and underestimates conversion among treated visitors, an estimator problem that PPI++ repairs, unlike the estimand problem of converter-only metrics (Section~\ref{sec:setup-notation}).

\begin{table}[!ht]
\centering
\caption{Hillstrom descriptive statistics.}
\label{tab:hillstrom-descriptive}
\small
\begin{tabular}{@{}lrr@{}}
\toprule
\textbf{Quantity} & \textbf{Value} & \textbf{Note} \\
\midrule
$P(S=1|T=0)$ (visit rate, control) & 0.1062 & \\
$P(S=1|T=1)$ (visit rate, treated) & 0.1671 & \\
ATE on surrogate ($\tau_S$) & 0.0609 & treatment $\to$ visit \\
$P(Y=1|T=0)$ (conversion, control) & 0.00573 & \\
$P(Y=1|T=1)$ (conversion, treated) & 0.01068 & \\
ATE on primary outcome ($\tau$) & 0.00495 & treatment $\to$ conversion \\
$P(Y=1|S=1, T=0)$ (conversion given visit, control) & 0.0539 & \\
$P(Y=1|S=1, T=1)$ (conversion given visit, treated) & 0.0639 & \\
Within-visitor cond. mean diff. & $+0.0100$ & violates conditional-mean invariance \\
$P(Y=1|S=0, T=t)$ (conversion given no visit) & 0.0000 & no conversions among non-visitors \\
\bottomrule
\end{tabular}
\smallskip

{\scriptsize \noindent Notes: ATEs computed on the full sample as differences in means.}
\end{table}

For each of $R = 500$ random splits we label a fraction $\pi_L$ of the units; Table~\ref{tab:hillstrom} reports $\pi_L = 0.20$.

\begin{table}[!ht]
\centering
\caption{Hillstrom email marketing experiment results.}
\label{tab:hillstrom}
\small
\begin{tabular}{@{}lrrrr@{}}
\toprule
\textbf{Method} & \textbf{Bias (MC SE)} & \textbf{Abs.\ Rel.\ Bias (\%)} & \textbf{Coverage (MC SE)} & \textbf{RE} \\
\midrule
Labeled-Only & $-0.0001$ (0.0001) & 1.4 & \textbf{0.972} (0.007) & 1.00 \\
Naive Surrogate & $-0.0012$ (0.00001) & 24.9 & \textbf{0.002} (0.002) & --- \\
Surrogate Index & $-0.0012$ (0.00001) & 24.9 & \textbf{0.072} (0.012) & --- \\
PPI++ & $-0.0001$ (0.0001) & 1.6 & \textbf{0.972} (0.007) & 1.03 \\
CP (SI fallback) & $-0.0012$ (0.00001) & 24.9 & \textbf{0.072} (0.012) & --- \\
Hybrid ($c = 1.5$) & $-0.00053$ (0.00005) & 10.7 & \textbf{0.930} (0.011) & --- \\
\bottomrule
\end{tabular}
\smallskip

{\scriptsize \noindent Notes: $R = 500$ random splits, all-units five-fold cross-fitting. Masking target $0.004955$. PPI++ is the primary configuration and the surrogate index uses the joint sandwich variance. \textbf{Bold} coverage entries fall outside the band $[0.931, 0.969]$, in either direction. RE is suppressed (---) for methods with absolute relative bias above 10\%. The hybrid uses $c = 1.5$ and the simulation-calibrated interval of Appendix~\ref{app:hybrid}. CP (SI fallback): no past experiments exist for this dataset, so the composite proxy returns the surrogate index (Method~4). The diagnostic (two-sided, $\alpha = 0.05$) rejects on 11.4\% of splits. Labeled-only and PPI++ coverage rises to 1.000 for both at $\pi_L = 0.50$. AIPW is not run in this analysis.}
\end{table}

The surrogate index has 24.9\% relative bias and 7.2\% coverage, so a violation of about one percentage point on the conditional conversion rate is enough to break it; carrying the first-stage term is what separates its 7.2\% from the naive surrogate's 0.2\%, and it repairs neither. PPI++ is effectively unbiased, with conservative 97.2\% coverage and modest gains of 1.01--1.04$\times$ across $\pi_L$ (1.03$\times$ at $\pi_L = 0.20$): it gives up a small amount of efficiency to remain valid under arbitrary surrogate-outcome misspecification, and it is the appropriate default on this dataset.

The diagnostic rejects on only 11.4\% of splits at $\pi_L = 0.20$, because the absolute SI--PPI++ gap (about 0.0012) is small relative to its sampling variability. The hybrid's weight therefore stays in the SI-dominant regime (mean 0.64, 10th to 90th percentiles 0.26 to 0.99), and it carries 10.7\% relative bias at 93.0\% coverage; across labeled fractions its bias runs from 8.4\% at $\pi_L = 0.05$ to 10.7\% at $\pi_L = 0.20$, and replacing the sandwich by the plug-in SI variance, a fifth as large, moves it very little (Table~\ref{tab:hillstrom-hybrid}). This is the local-alternative regime of Proposition~\ref{prop:limit-risk}: when $T_n$ is small relative to $c$, the weight stays near one, the hybrid keeps part of the SI bias, and its interval carries no guarantee (Remark~\ref{rem:ci-validity}). On Hillstrom, PPI++ is the analysis to report, and the hybrid is at most a display of estimator disagreement.

\subsection{Criteo: the same funnel violation at industrial scale}
\label{sec:criteo}

The Criteo uplift data \citep{diemert2018large}, released by Criteo AI Lab, come from a randomized advertising experiment; the v2.1 release has 13{,}979{,}592 users (treatment $n_1 = 11{,}882{,}655$; control $n_0 = 2{,}096{,}937$), twelve anonymized covariates, a visit indicator (the surrogate), and a conversion indicator (the outcome). The funnel replicates Hillstrom's at 218 times the scale: no user converts without visiting, visits are roughly sixteen times as common as conversions, and treatment raises the visit rate from 3.82\% to 4.85\%, the conversion rate from 0.194\% to 0.309\% (masking target $0.00115$), and the conversion rate among visitors from 5.07\% to 6.36\%, an unadjusted within-visitor conditional-mean difference of $+1.29$ points. A logistic regression of conversion on treatment and the covariates $f_0$ through $f_{11}$ among the 656{,}929 visitors, with HC1 standard errors, gives a treatment coefficient of $+0.310$ (SE $0.020$, $p < 10^{-16}$, odds ratio $1.36$) and a covariate-adjusted within-visitor difference of $+1.35$ points (SE $0.08$ points, $p < 10^{-16}$). That is a decisive rejection within one fitted specification, not a model-free verification that conditional-mean invariance fails: another specification of the visitor-level model could in principle absorb the coefficient. Two independent advertising-response funnels thus show the same within-funnel pattern at similar magnitude; other funnel types (signup, activation, trial-to-paid) remain untested.

The two adjusted violations are nearly the same size, $+1.35$ points at Criteo and $+1.03$ at Hillstrom, but only one is statistically visible: $p < 10^{-16}$ at $n = 14$ million against $p = 0.064$ at $n = 64{,}000$. Damage before detection applies to the evidence for the violation as well as to the diagnostic meant to catch it.

\begin{table}[!ht]
\centering
\caption{Criteo: detection versus scale ($\pi_L = 0.20$ within each subsample).}
\label{tab:criteo-scale}
\small
\begin{tabular}{@{}rrrr@{}}
\toprule
$\boldsymbol{n}$ & \textbf{SI coverage} & \textbf{PPI++ coverage} & \textbf{Diagnostic rejection} \\
\midrule
10{,}000 & \textbf{0.847} & \textbf{0.897} & 0.270 \\
30{,}000 & \textbf{0.660} & \textbf{0.850} & 0.197 \\
100{,}000 & \textbf{0.183} & 0.937 & 0.203 \\
300{,}000 & \textbf{0.000} & 0.950 & 0.215 \\
1{,}000{,}000 & \textbf{0.000} & 0.920 & 0.560 \\
3{,}000{,}000 & \textbf{0.000} & 0.967 & 1.000 \\
13{,}979{,}592 & \textbf{0.000} & 1.000 & 1.000 \\
\bottomrule
\end{tabular}
\smallskip

{\scriptsize \noindent Notes: subsamples drawn without replacement, all-units five-fold cross-fitting; $R = 300$ per row for $n \leq 10^5$, then $200$, $100$, $60$, and $40$. PPI++ is the primary configuration and the surrogate index uses the joint sandwich variance. \textbf{Bold} coverage entries fall outside the band at each row's $R$: $[0.925, 0.975]$ at $R = 300$, $[0.920, 0.980]$ at $R = 200$, $[0.907, 0.993]$ at $R = 100$, and wider at $R = 60$ and $40$. Coverage is measured against the masking target.}
\end{table}

\begin{figure}[!ht]
  \centering
  \includegraphics[width=0.9\textwidth]{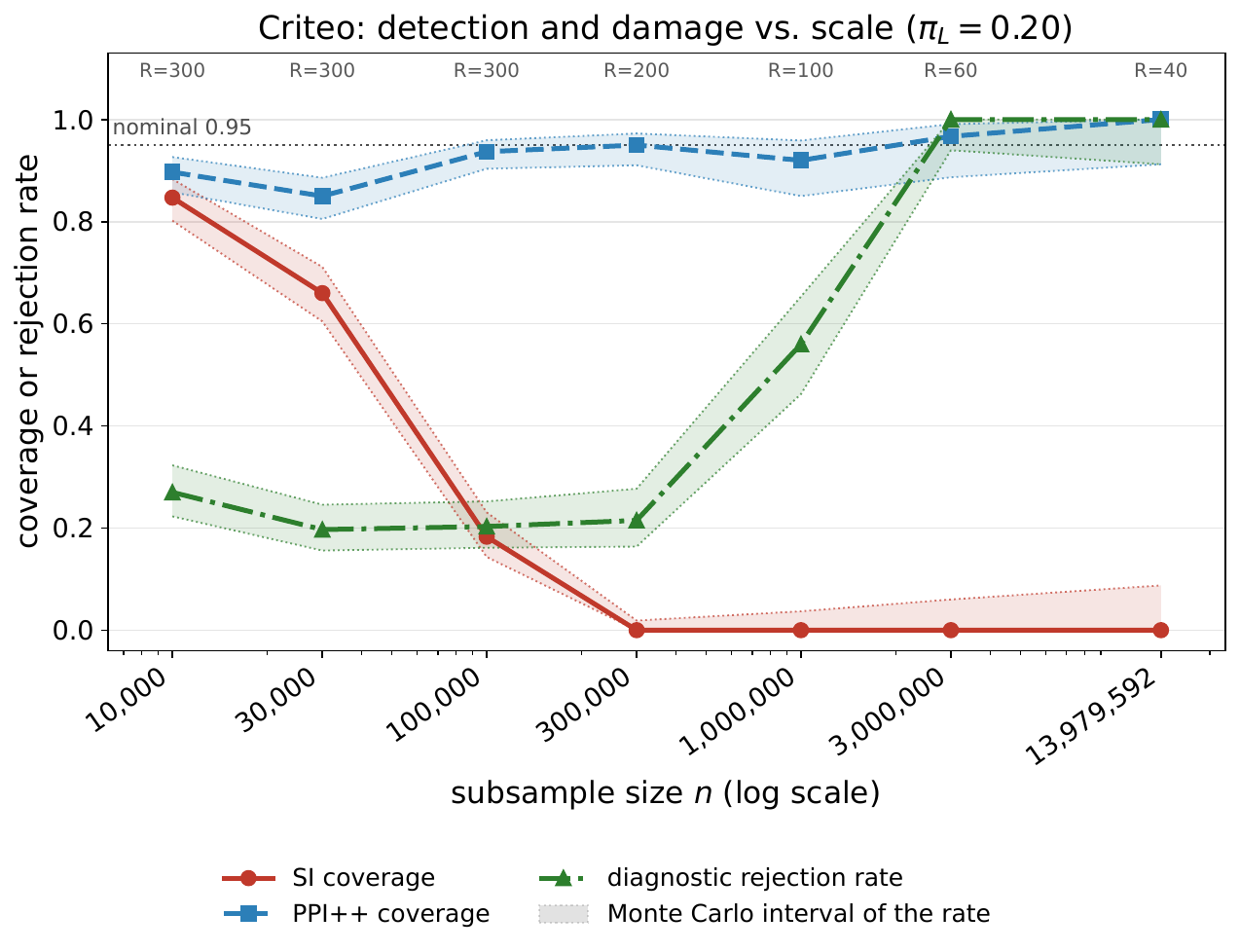}
  \caption{Criteo subsampling: SI coverage, PPI++ coverage, and two-sided diagnostic rejection against $\log n$.}
  \label{fig:criteo-scale}
  \smallskip

  {\scriptsize \noindent Notes: seven subsample sizes from $10{,}000$ to $13{,}979{,}592$ at $\pi_L = 0.20$, with $R = 300$, $300$, $300$, $200$, $100$, $60$, and $40$ in increasing order of $n$. Ribbons are 95\% Wilson score Monte Carlo intervals for each rate, not method confidence intervals. Red circles: SI coverage under the joint sandwich variance; blue squares: PPI++ coverage with the exact variance; green triangles: two-sided diagnostic rejection at $\alpha = 0.05$; dotted line: nominal $0.95$.}
\end{figure}

At full scale, across $R = 100$ random splits at $\pi_L = 0.20$, the surrogate index carries $-48.8\%$ relative bias with zero coverage, PPI++ is unbiased ($0.6\%$ relative bias, RE $1.06$), and the diagnostic rejects on every split (mean $T_n = 8.1$; at $\pi_L = 0.05$ it rejects on 99\% of splits with mean $T_n = 4.0$; Table~\ref{tab:criteo-full}). At this scale the violation is both fatal to the surrogate index and detected.

Subsampling traces the transition (Table~\ref{tab:criteo-scale}). The SI bias is nearly scale-invariant, $-5.1$ to $-5.7 \times 10^{-4}$ ($-44\%$ to $-50\%$ relative), as expected of a structural violation, while SI coverage falls from $0.847$ at $n = 10{,}000$ to $0.660$ at $30{,}000$, $0.183$ at $100{,}000$, and zero from $300{,}000$ on. The diagnostic's rejection rate is flat between $0.197$ and $0.270$ through $n = 300{,}000$, a spread of about two Monte Carlo standard errors, crosses majority detection only around $n \approx 10^6$, and saturates by $n = 3$ million. Between roughly $10^4$ and $10^6$ lies the real-data detection-damage gap: two orders of magnitude of sample size across which SI inference is invalid and the diagnostic is usually silent (Figure~\ref{fig:criteo-scale}). Hillstrom shows the same pattern at $n = 64{,}000$, with 7.2\% SI coverage and an 11\% rejection rate; its effect size, conversion rate, and variances differ from Criteo's.

The same table carries the clearest counterexample to PPI++'s own validity: its exact-variance coverage is $0.850$ at $n = 30{,}000$ and $0.897$ at $n = 10{,}000$, both outside their bands (the ordering of the two is within Monte Carlo error). The mechanism is outcome rarity, not label scarcity. With conversions near two in a thousand and an 85/15 treatment--control split, a subsample's labeled control arm holds a handful of conversions, the arm means are far from normal, and the Wald interval PPI++ inherits from the labeled-only estimator undercovers although the exact variance is in use. DGP~5 varies how many Gaussian outcomes are observed, so the simulations do not cover this regime.

\subsection{A Criteo-calibrated semi-synthetic testbed}
\label{sec:semisynthetic}

The DGPs of Section~\ref{sec:simulation} are stylized, and the Hillstrom and Criteo analyses are single realizations with violations of fixed size. A semi-synthetic testbed bridges the two. We fit logistic models for the two stages of the Criteo funnel (visit given covariates and treatment; conversion given a visit, covariates, and treatment), resample covariate rows from the real data, and keep conversion impossible without a visit. The fitted within-funnel treatment coefficient, $\hat{\kappa} = 0.31$, is placed on a dial in units of the observed violation: dial 0 imposes conditional-mean invariance exactly, dial 1 reproduces Criteo's violation, and dial 2 doubles it, with covariates, funnel rates, and outcome sparsity held at their calibrated values. We run the standard protocol at the Hillstrom scale ($n = 64{,}000$, $\pi_L = 0.20$, OLS prediction model, $R = 500$ per dial), measuring bias and coverage against the known simulated $\tau$ of each dial.

Table~\ref{tab:semisynthetic} reports the sweep. At dial 0 the calibrated funnel behaves like DGP~1: SI is nearly unbiased with large gains (RE $5.7\times$) and mild over-coverage at $0.976$, PPI++ delivers RE $1.06$ at $0.942$ coverage, and the diagnostic rejects at its nominal size (6.2\%). The residual SI relative bias at dial 0 ($+3.5\%$) is within 1.4 Monte Carlo standard errors of zero; to the extent it is real, the OLS model's inability to represent the logistic conversion surface is the natural source, since surrogacy holds by construction. At this tiny $\tau$ the relative biases of the noisier estimators carry Monte Carlo standard errors of roughly $\pm 14$ percentage points, so the single-digit LO and PPI++ entries of Table~\ref{tab:semisynthetic-full} are indistinguishable from zero. At dial 1, the violation the real data exhibit, SI coverage collapses to 8\% at $-68\%$ relative bias while the diagnostic detects the violation on only 10\% of replications, close to the Hillstrom rejection rate (11\%) at the same sample size; the testbed is calibrated to Criteo, not Hillstrom, so the agreement reflects similar violation-to-noise ratios rather than a fitted correspondence. At twice the real violation, detection reaches only 35\% while SI coverage is exactly zero, and PPI++ is unbiased with near-nominal coverage at every dial.

\begin{table}[!ht]
\centering
\caption{Criteo-calibrated semi-synthetic sweep across the violation dial.}
\label{tab:semisynthetic}
\small
\begin{tabular}{@{}rrrrrr@{}}
\toprule
\textbf{Dial} & $\boldsymbol{\tau \times 10^4}$ & \textbf{SI rel.\ bias (\%)} & \textbf{SI coverage} & \textbf{PPI++ coverage} & \textbf{Diagnostic rejection} \\
\midrule
0.0 & 2.65 & $+3.5$ & $\mathbf{0.976}$ & 0.942 & 0.062 \\
0.5 & 5.70 & $-51.6$ & $\mathbf{0.550}$ & 0.944 & 0.052 \\
1.0 & 9.10 & $-67.5$ & $\mathbf{0.084}$ & 0.948 & 0.100 \\
1.5 & 12.90 & $-74.6$ & $\mathbf{0.004}$ & 0.942 & 0.236 \\
2.0 & 17.13 & $-79.2$ & $\mathbf{0.000}$ & 0.940 & 0.354 \\
\bottomrule
\end{tabular}
\smallskip

{\scriptsize \noindent Notes: $n = 64{,}000$, $\pi_L = 0.20$, $R = 500$ per dial, all-units five-fold cross-fitting. PPI++ is the primary configuration and the surrogate index uses the joint sandwich variance. The dial multiplies the fitted within-funnel violation $\hat{\kappa} = 0.31$; dial 1 is the within-visitor difference observed in Criteo. \textbf{Bold} coverage entries fall outside the band $[0.931, 0.969]$, in either direction. SI relative bias grows with the dial because the surrogate index tracks the invariant part of the funnel while the true $\tau$ grows with $\kappa$. Per-method results are in Table~\ref{tab:semisynthetic-full}.}
\end{table}

\paragraph{Does a richer surrogate index absorb the violation?}
The single binary visit indicator makes the violation maximally hidden, and platforms in practice feed surrogate indices many signals. We add a post-treatment engagement proxy and split the fitted conversion-stage violation between it and a residual direct effect. With covariate rows $X$ resampled from the Criteo pool and $T \sim \mathrm{Bernoulli}(0.5)$, the rich funnel is
\begin{align}
  S_1 \mid T, X &\sim \mathrm{Bernoulli}\bigl(\sigma(c_v + X'a_v + b_v T)\bigr), \qquad S_2 = \zeta T + \varepsilon_2, \quad \varepsilon_2 \sim \mathcal{N}(0,1), \label{eq:rich-visit}\\
  Y &= S_1 B, \qquad B \mid T, X, S_2 \sim \mathrm{Bernoulli}\bigl(\sigma(c_c + X'a_c + \theta S_2 + \kappa_{\mathrm{resid}} T)\bigr), \label{eq:rich-conversion}
\end{align}
with $\sigma$ the logistic function, $(c_v, a_v, b_v)$ and $(c_c, a_c)$ the fitted Criteo stage coefficients ($b_v = 0.23$), $\zeta = 0.5$, $\theta = m\hat\kappa/\zeta$, $\kappa_{\mathrm{resid}} = (1-m)\hat\kappa$, and $\hat\kappa = 0.31$, so the total treatment coefficient on the conversion index, $\theta\zeta + \kappa_{\mathrm{resid}}$, equals $\hat\kappa$ at every mediated share $m$. The rich index is OLS on $(1, S_1, S_2, S_1^2, S_2^2, X)$, with $S_1^2 = S_1$; the visit-only index is OLS on $(1, S_1, S_1^2, X)$ and satisfies surrogacy at no $m$. At $m = 1$ surrogacy holds for the rich index, since $\E[Y \mid T, S_1, S_2, X] = S_1\,\sigma(c_c + X'a_c + \theta S_2)$ does not involve $T$, but the OLS basis is misspecified for that logistic mean, so the restriction $d_0'\beta_0 = \tau$ that the diagnostic's null requires is not guaranteed. Computed exactly over the covariate pool, with $S_2$ integrated by quadrature, the population offset $d_0'\beta_0 - \tau$ at $m = 1$ is $+1.07 \times 10^{-5}$, $+1.06\%$ of $\tau$ and $0.013$ Monte Carlo standard deviations of $\hat{D}$ at $n = 64{,}000$. The $m = 1$ cell is therefore a near-null stress check, and the Gaussian cells of Table~\ref{tab:joint-cov-R2000} remain the null validation.

Table~\ref{tab:multisurrogate} shows the outcome. The visit-only SI carries $-65$ to $-66\%$ relative bias at every $m$. The rich index's bias falls with the mediated share, $-63.5\%$ at $m = 0$, $-31.5\%$ at $m = 0.5$, and $-0.4\%$ at $m = 1$, against population offsets of $-62.1\%$, $-30.4\%$, and $+1.1\%$ of $\tau$, each within about one Monte Carlo standard error of the observed bias. The near-proportionality to the unmediated share $1 - m$ is an empirical finding on the probability scale, not an identity, since splitting a coefficient inside a logistic model does not split the effect on conversion probability linearly. A residual direct effect of half the observed size still leaves the rich index with $-31.5\%$ bias and 80\% coverage while the diagnostic rejects on under 6\% of replications, and a signal off the path of the direct effect (the $m = 0$ row) changes little. The question to ask of a proposed index is whether any plausible direct pathway bypasses all of its signals.

\begin{table}[!ht]
\centering
\caption{Multi-dimensional surrogate on the Criteo-calibrated funnel ($n = 64{,}000$, $\pi_L = 0.20$, $R = 500$; total violation fixed at dial 1, mediated share $m$).}
\label{tab:multisurrogate}
\small
\begin{tabular}{@{}crrrrr@{}}
\toprule
& \multicolumn{2}{c}{\textbf{SI rel.\ bias (\%)}} & \multicolumn{2}{c}{\textbf{SI coverage}} & \textbf{Diag.\ rejection} \\
$\boldsymbol{m}$ & visit only & visit $+ S_2$ & visit only & visit $+ S_2$ & visit $+ S_2$ \\
\midrule
0.0 & $-65.9$ & $-63.5$ & $\mathbf{0.112}$ & $\mathbf{0.414}$ & 0.138 \\
0.5 & $-66.0$ & $-31.5$ & $\mathbf{0.080}$ & $\mathbf{0.800}$ & 0.056 \\
1.0 & $-65.3$ & $-0.4$ & $\mathbf{0.080}$ & $0.954$ & 0.050 \\
\bottomrule
\end{tabular}
\smallskip

{\scriptsize \noindent Notes: all-units five-fold cross-fitting. PPI++ is the primary configuration and is unbiased with 0.938--0.962 coverage in every cell. SI coverage uses the joint sandwich variance, which carries the first-stage uncertainty of the learned index (71\% of the sandwich variance on the rich index at $m = 1$). \textbf{Bold} coverage entries fall outside the band $[0.931, 0.969]$.}
\end{table}

\subsection{LaLonde: a negative control with a pre-treatment ``surrogate''}
\label{sec:lalonde}

The LaLonde National Supported Work (NSW) experiment \citep{lalonde1986evaluating} randomized 445 individuals into a job-training program (treatment, $n_1 = 185$) or control ($n_0 = 260$). We treat 1975 earnings (RE75) as the ``surrogate'' and 1978 earnings (RE78) as the outcome; the masking target is $\$1{,}794$, and the results over $R = 500$ random splits are in the replication package. RE75 was realized before the 1976--1977 program began, so treatment cannot have caused it, and it fails the Prentice requirement that the surrogate lie on the causal path. As the simulations predict, the surrogate index is massively biased: at $\pi_L = 0.20$ its bias is $-\$1{,}645$ ($-92\%$) with 5.2\% coverage, and at $\pi_L = 0.50$ its coverage is 0\%. PPI++ keeps near-zero bias and 92.6--98.8\% coverage at $\pi_L \in \{0.10, 0.20, 0.50\}$, reverting to the labeled-only baseline (RE $0.99$--$1.00$); the 92.6\% at $\pi_L = 0.10$ reflects the small labeled sample ($n_L \approx 45$). The diagnostic rejects on 11.8\%, 15.2\%, and 36.8\% of splits at the three labeled fractions, the limited power available at $n = 445$: in small samples, domain knowledge must drive surrogate selection.

\section{Practical Recommendations}
\label{sec:recommendations}

The workflow below separates a prespecified primary inferential analysis from sensitivity analyses and decision exploration. Its thresholds are calibrated to our DGPs and should be adapted to a platform's context.

\paragraph{The principle.} Prespecify the primary inferential analysis before reading the experiment. Under MCAR labeling and adequate labeled outcome counts in each arm, PPI++ with the exact tuning constant of Equation~\eqref{eq:lambda-exact} and the exact variance of Proposition~\ref{prop:corrected-variance} is the default primary analysis when surrogacy is uncertain. The surrogate index and the hybrid are sensitivity analyses or tools for decision exploration. A team that wants the surrogate index as its primary analysis must justify surrogacy independently of the diagnostic, from prior validation studies or a mechanism by which treatment cannot bypass the surrogate, and should record that justification before the data are read. A data-adaptive rule that selects the interval after the diagnostic would need validation of the whole selection procedure, which we do not provide.

\paragraph{Step 0: Establish that the workflow applies.}
Three conditions need separate justification before the steps below. \emph{Labeling:} the PPI++ guarantee is an MCAR guarantee. Under nonrandom labeling unadjusted PPI++ has no design-based centering (under mix drift in DGP~11 it carries $-5.0\%$ relative bias and its exact variance estimator underestimates the sampling variance), and an estimator that models the labeling, such as AIPW with a correct observation model, is needed. \emph{Enrollment-effect drift:} under the calendar interpretation the labeled units are the earliest enrollees. Compare the $(X_i, S_i)$ distributions and the treatment effect on the surrogate across enrollment waves; covariate-mix drift is repairable by reweighting on observables, and the model-based estimators tolerate it (Section~\ref{sec:results-enrollment}). Covariate checks cannot establish that the treatment effect on the outcome is stable, however, and drift in that effect is not repairable: every estimator we test carries roughly a third of the effect as bias at $g_e = 0.2$. If effect drift is plausible, report the review as an early read on the early cohort, not as an estimate of the eventual effect. \emph{Sparse outcomes:} count labeled outcome events in each arm, not the total sample. At the $n = 30{,}000$ Criteo subsample PPI++ covers $0.850$ because the labeled control arm holds only a handful of conversions (Section~\ref{sec:criteo}); sparse outcomes need separately justified inference or more labels. A prediction model trained on a separate prior cohort (Regime~B) adds a transport assumption that DGP~4 shows can fail, and its estimates are exploratory until transport is verified.

\paragraph{Step 1: The primary analysis.}
Fit $\hat{f}$ on the labeled units under all-units cross-fitting and compute PPI++ with the common, unclipped coefficient of Equation~\eqref{eq:lambda-exact}; in the e-commerce example $\hat{f}$ maps funnel signals to conversion on the early enrollees whose conversion windows have closed. Form the interval from $\hat{V}_{\mathrm{corrected}} = \hat{V}_{\mathrm{plug\text{-}in}} + \sum_t 2\hat{\lambda}(\hat{\gamma}_t - \hat{\lambda}\hat{\sigma}^2_{\hat{Y},t})/n_t$ at that coefficient, so that the interval and the tuning rule come from one objective. The overlap correction is largest when the coefficient sits below the per-arm optimum, as under a plug-in rule or a binding clip, and matters most at $\pi_L > 0.30$; at the exact coefficient it is small when the arms have similar labeled fractions, and it is computed in every case. The population-optimal coefficient gives a variance no larger than labeled-only's (Proposition~\ref{prop:efficiency}), but an estimated coefficient carries no finite-sample guarantee; in our designs RE is $1.4$--$1.5\times$ at $\pi_L \leq 0.20$ and falls toward $1.0\times$ as $\pi_L$ approaches one.

\paragraph{Step 2: Read the diagnostic as a warning.}
Compute $\hat{\tau}_{\mathrm{SI}}$ and the two-sided estimator-disagreement statistic of Section~\ref{sec:surrogacy-test}. We suggest $\alpha = 0.10$, a heuristic that favors detection because an undetected violation usually costs more than a false alarm. A rejection is evidence that the two estimands diverge, through a surrogacy failure or prediction-model misspecification, and the surrogate index and the hybrid should then not inform the decision. A non-rejection neither establishes surrogacy nor licenses the SI interval: reporting the nominal SI interval after a non-rejection is the hard-switch rule of Table~\ref{tab:hard-switch}, whose coverage in DGP~2 is $0.699$ at $\rho = 0.2$ and $0.706$ at $\rho = 0.4$. Power is also low where it matters: under DGP~2 at $n = 10{,}000$ and $\pi_L = 0.20$ the test at $\alpha = 0.10$ detects $\rho = 0.2$ on 19.2\% of replications (Section~\ref{sec:detection-damage}); at $n = 100{,}000$ power at $\alpha = 0.05$ reaches 79\%, and at $n = 1{,}000$ it is negligible (Table~\ref{tab:power-curves}).

\paragraph{Step 3: Sensitivity analysis and decision exploration.}
Report the surrogate index beside the primary analysis as a sensitivity analysis. Under valid surrogacy it is far more precise, with RE $5.5\times$ at $\pi_L = 0.05$ and $2.9\times$ at $\pi_L = 0.20$ (Table~\ref{tab:dgp1}), the gain available when the surrogacy assumption is independently justified. The hybrid with $c = 1.5$ can be reported as an exploratory point estimate (RE $1.5$--$2.2\times$ under valid surrogacy); its interval has no coverage guarantee at fixed $c$ (Remark~\ref{rem:ci-validity}), and when its mean weight on SI exceeds about $0.5$ while a direct effect is plausible, it carries much of any undetected SI bias, as on Hillstrom (10.7\% relative bias at mean weight $0.64$).

\paragraph{Step 4: Separate formal inference from portfolio launch triage.}
Formal inference rests on the primary analysis of Step~1. For go-or-no-go launch triage, biased surrogate estimators can still achieve high correct-decision rates when the surrogate preserves the sign of the effect, and in the valid-surrogacy portfolio of DGP~6 the naive surrogate and the surrogate index attain the lowest regret (Section~\ref{sec:results-decisions}). That ranking depends on the objective: DGP~6 charges effect-size-weighted loss with no launch cost, so a platform whose launches carry real costs, or whose candidates include harmful changes, should re-derive the triage rule under its own objective.

\section{Discussion and Conclusion}
\label{sec:discussion}

The central finding is a sharp robustness--efficiency tradeoff among surrogate estimators. Surrogate-based methods deliver large efficiency gains when surrogacy holds and lose coverage under moderate violations, while PPI-family methods stay asymptotically valid under MCAR labeling with adequate labeled outcome counts in each arm, at modest gains (Section~\ref{sec:results}). On the calendar, the week-2.5 PPI++ interval matches labeled-only at about week 2.97, so PPI++ shortens the wait by roughly half a week without imposing surrogacy, while a valid surrogate index at the same review is already narrower than labeled-only at full maturation in week six, a saving of three and a half weeks; which of those claims a team will stake a launch on is the robustness--efficiency choice. Diagnostics help but do not validate surrogacy: under partial mediation some coverage-destroying violations are statistically invisible at both sample sizes we examine (Section~\ref{sec:detection-damage}), and the exploratory hybrid mitigates that risk without removing it. Enrollment-time labeling is benign without drift and survivable under covariate-mix drift, but drift in the treatment effect breaks every estimator (Section~\ref{sec:results-enrollment}), so reading the review date as a choice of $\pi_L$ requires effect stability over the enrollment window.

These conclusions have several limitations. We test two prediction model classes (OLS and GBT); neural networks or heavily tuned ensembles could change the efficiency profiles, though the SI bias is structural rather than model-driven. Apart from DGP~3's group heterogeneity and the covariate-heterogeneous surrogate effect behind DGP~11's mix drift, the DGPs have constant treatment effects, and a surrogate valid for the ATE may be invalid for conditional effects, which leaves surrogate-based heterogeneous-effect estimation open. The hybrid's $c = 1.5$ rests on the finite-sample grids and the limit experiment, and the covariance-adaptive rule of Table~\ref{tab:adaptive-c} matches it on the cells tested, but covariance structures far from ours are untested. The diagnostic has limited power at $n < 1{,}000$ and at moderate violations ($\rho \leq 0.2$ at $n = 10{,}000$). The composite proxy's weights come from one simulated library of past experiments per run, drawn from the same DGP family as the current experiment, so its results are optimistic relative to practice and do not average over library estimation uncertainty. Rare outcomes bound the main recommendation: at the sparse Criteo subsamples PPI++ covers $0.850$ and $0.897$ (Table~\ref{tab:criteo-scale}) because the normal approximation for the labeled arm means fails, a regime our simulations do not cover and for which we evaluate no sparse-event method. Finally, the real-data analyses use binary funnel surrogates; the calibrated multi-surrogate experiment (Table~\ref{tab:multisurrogate}) shows that enriching the index removes roughly the share of the violation the added signals mediate, but validation on real multi-metric indices, continuous engagement surrogates, revenue outcomes, and proprietary platform data remains for future work.

These findings connect to recent theory. \citet{ji2025predictions} unify the surrogate index and PPI as endpoints of a bias--variance spectrum, and our simulations characterize that spectrum in finite samples under realistic failure modes. \citet{mozer2026ppi} shows that PPI++ is the classical difference estimator of \citet{cassel1976results}, and survey-sampling experience favors such model-assisted estimators when auxiliary information is unreliable, the conclusion our simulations reach. \citet{kallus2025surrogates} derive the semiparametric efficiency bound for treatment effects with limited outcome data in a missing-data model that does \emph{not} impose statistical surrogacy, and construct estimators that attain it; that bound concerns their model, not the surrogacy-imposing one in which the surrogate index gains its efficiency, and we do not benchmark against it. Our corrected variance (Proposition~\ref{prop:corrected-variance}) writes the exact fixed-predictor variance of PPI++ in the all-units parameterization, the form in which practitioners most often implement it and in which the overlap covariance is easiest to lose; it agrees with the variance the labeled-unlabeled form of \citet{angelopoulos2023ppipp} assigns after the change of tuning variable.

Several extensions are open: a sequential hybrid that updates $c$ as data accumulate, paired with anytime-valid inference \citep{kilian2025anytime} to monitor the weight $w$ during an experiment; a Bayesian formulation with a prior on the violation $\rho$ that encodes domain knowledge about surrogate reliability; and a formal account of combining CUPED with PPI++, which exploit pre- and post-treatment information respectively and are routinely applied together.

For a team deciding at a mid-experiment review whether to call the result, the message is to prespecify PPI++ with the exact variance as the primary analysis when the conditions of Section~\ref{sec:recommendations} hold, to read the diagnostic as a warning, and to report the surrogate index as a sensitivity analysis. The workflow requires no additional data collection and adds negligible computational cost.

\paragraph{Code and data.} All code and the result tables behind every table and figure are in the replication repository at \url{https://github.com/zc-seattle/surrogate-metrics-replication} (release v1.0), with the three public datasets obtained as described there.

\FloatBarrier  
\bibliographystyle{plainnat}
\bibliography{references}

\clearpage
\appendix
\begin{center}
{\LARGE\bfseries Appendix}\\[6pt]
{\large Supplementary Proofs, Extended Results, and Ablations}
\end{center}
\vspace{1em}
\renewcommand{\thetable}{A\arabic{table}}
\setcounter{table}{0}
\renewcommand{\thefigure}{A\arabic{figure}}
\setcounter{figure}{0}

\section{The Adaptive Hybrid Estimator}
\label{app:hybrid}

This appendix gives the hybrid's construction, the risk analysis that fixes its tuning constant, the simulation grids behind Section~\ref{sec:results-hybrid}, and the asymptotics of a diverging-$c_n$ variant.

\subsection{Construction and Interval}
\label{app:hybrid-construction}

A hard-switch rule would choose between SI and PPI++; the exploratory hybrid instead combines them continuously, interpolating between the two according to the evidence against surrogacy. It belongs to the frequentist model-averaging family of \citet{hjort2003frequentist} (see also \citealp{claeskens2008model}), a compromise estimator whose data-dependent weight is a smooth function of a specification statistic, and it inherits that family's risk behavior, including the raw-risk hump at intermediate violations characteristic of pretest and shrinkage estimators \citep{judge1978pretest}, which Proposition~\ref{prop:limit-risk} characterizes exactly for our weight.

\paragraph{Definition.} The adaptive hybrid estimator is
\begin{equation}
  \hat{\tau}_{\mathrm{hybrid}} = w \cdot \hat{\tau}_{\mathrm{SI}} + (1 - w) \cdot \hat{\tau}_{\mathrm{PPI++}},
  \label{eq:hybrid}
\end{equation}
where the Cauchy-kernel weight
\begin{equation}
  w = \frac{c}{c + T_n^2}
  \label{eq:cauchy-weight}
\end{equation}
smoothly transitions from $w \approx 1$ (use SI) when $T_n \approx 0$ to $w \approx 0$ (use PPI++) when $|T_n|$ is large. We recommend $c = 1.5$ based on the minimax analysis of Section~\ref{app:hybrid-tuning}.

Two versions must be kept apart. The \emph{fixed-$c$} hybrid, run everywhere in this paper, holds $c$ constant and is characterized exactly by the limit experiment below. A \emph{diverging-$c_n$} variant has cleaner asymptotic endpoints: with $c_n \to \infty$ it attains the surrogate index's asymptotic variance under $H_0$, and with $c_n = o(\sqrt{n})$ it is root-$n$ equivalent to PPI++ under a fixed violation (Propositions~\ref{prop:hybrid-efficiency} and~\ref{prop:hybrid-safety}, Section~\ref{app:diverging-c}). Under local alternatives $\delta_n = h/\sqrt{n}$, however, the diagnostic statistic stays $O_p(1)$, so $c_n \to \infty$ forces $w \to 1$ and the diverging-$c_n$ hybrid inherits SI's local bias in full; the nondegenerate weight limit that makes the hybrid interesting under local alternatives belongs to the fixed-$c$ version. We therefore claim no uniform oracle bound for either, and justify $c = 1.5$ by the limit experiment and the minimax grids of Section~\ref{app:hybrid-tuning}.

The data-dependent weight $w$ invalidates the standard Wald CI, so we construct a simulation-calibrated CI as follows. The $2 \times 2$ covariance matrix of $(\hat{\tau}_{\mathrm{SI}}, \hat{\tau}_{\mathrm{PPI++}})$ is $\hat{\Sigma} = \hat\Omega/n$, the sandwich of Proposition~\ref{prop:joint-si-ppi}; the same $\hat{\Sigma}$ supplies $\mathrm{SE}(\hat{D})$, so the weight and the interval are built from one covariance estimate. We then draw $M = 10{,}000$ samples from $\mathcal{N}\bigl((\hat{\tau}_{\mathrm{SI}}, \hat{\tau}_{\mathrm{PPI++}})^T, \hat{\Sigma}\bigr)$; on each draw we compute $T_n^{(m)}$, $w^{(m)}$, and $\hat{\tau}_{\mathrm{hybrid}}^{(m)}$, and report the 2.5th and 97.5th percentiles as the 95\% CI. The construction parallels the smoothed post-model-averaging intervals of \citet{efron2014estimation}.

\begin{remark}[The simulation-calibrated interval is a heuristic, not a consistent procedure]
\label{rem:ci-validity}
We make no consistency claim for this interval at fixed $c$. For fixed $c$ the hybrid is a non-regular shrinkage estimator: its weight is a fixed smooth function of $T_n$, and the standardized gap $h/\sigma_D$ that determines the weight's limiting distribution is estimated with error that does not vanish at the $\sqrt{n}$ scale on which the weight operates. The plug-in percentile distribution therefore fails to converge to the sampling distribution of the hybrid. This is the shrinkage-estimator instance of the general failure of resampling and plug-in distributional approximations near a pretest null \citep{leeb2005model}. Under fixed detectable violations, $w \xrightarrow{p} 0$ and the interval inherits PPI++ validity: it collapses onto the exact-variance PPI++ interval, which is valid under Assumption~\ref{asm:mcar} and the regularity conditions of Proposition~\ref{prop:joint-si-ppi}.
We report the interval's coverage empirically rather than claiming it: 94--96\% under valid surrogacy across the tested configurations, 93.1--93.4\% at $\rho \in \{0.2, 0.4\}$, and 93.0\% on Hillstrom at $\pi_L = 0.20$ (Sections~\ref{app:hybrid-sims} and~\ref{sec:empirical}). The hybrid is an exploratory estimator, and its interval should be read as a display of estimator disagreement rather than as formal inference.
\end{remark}

\subsection{Limit-Experiment Risk at Fixed \texorpdfstring{$c$}{c}}
\label{app:hybrid-limit}

The fixed-$c$ implementation itself admits an exact local-asymptotic characterization.

\begin{proposition}[Limit-experiment risk of the fixed-$c$ hybrid]
\label{prop:limit-risk}
Assume the conditions of Proposition~\ref{prop:surrogacy-test}, a local sequence of violations along which the estimator gap satisfies $\sqrt{n}\,\E[\hat{D}] \to h$, and the fixed-$c$ weight $w(T_n) = c/(c + T_n^2)$ with $c$ constant. The local parameter $h$ is the limiting gap in the diagnostic's own units; under DGP~2 with $\delta_n = d/\sqrt{n}$ it equals $\kappa d$ with $\kappa$ the attenuation factor of Proposition~\ref{prop:si-bias}, so $h$ and the direct effect differ by that factor. Let $\sigma_S^2$, $\sigma_P^2$, and $\sigma_{SP}$ denote the probability limits of $n\Var(\hat{\tau}_{\mathrm{SI}})$, $n\Var(\hat{\tau}_{\mathrm{PPI++}})$, and $n\Cov(\hat{\tau}_{\mathrm{SI}}, \hat{\tau}_{\mathrm{PPI++}})$, and set $\sigma_D^2 = \sigma_P^2 + \sigma_S^2 - 2\sigma_{SP}$, $\varrho = (\sigma_P^2 - \sigma_{SP})/\sigma_D^2$, and $\theta = h/\sigma_D$ for the standardized local violation. Then
\begin{equation}
  \sqrt{n}\,(\hat{\tau}_{\mathrm{hybrid}} - \tau) \xrightarrow{d} H(h; c) = G_P - w(V/\sigma_D)\,V,
  \label{eq:limit-H}
\end{equation}
where $(G_P, V)$ is bivariate normal with $G_P \sim \mathcal{N}(0, \sigma_P^2)$, $V \sim \mathcal{N}(h, \sigma_D^2)$, and $\Cov(G_P, V) = \varrho\,\sigma_D^2$. The asymptotic risk reduces to the one-dimensional Gaussian integral
\begin{equation}
  r(h; c) = \E[H(h; c)^2] = \sigma_P^2 - \varrho^2\sigma_D^2 + \E_V\!\left[\bigl(\varrho(V - h) - w(V/\sigma_D)\,V\bigr)^2\right],
  \label{eq:limit-risk}
\end{equation}
which is continuous in $(h, c)$ and satisfies $r(h; c) \to \sigma_P^2$ as $|h| \to \infty$ for fixed $c$.
\end{proposition}

Proposition~\ref{prop:limit-risk} makes the transition region computable. Calibrating $(\sigma_S, \sigma_P, \sigma_{SP})$ to the featured cell ($n = 10{,}000$, $\pi_L = 0.20$: $\sigma_S = 2.53$, $\sigma_P = 5.21$, correlation $0.52$), the risk curves of Figure~\ref{fig:limit-risk} show the characteristic pretest shape \citep{judge1978pretest}: near-SI risk at $h = 0$, a raw-risk hump at intermediate violations (peaking at $\theta$ between $2.8$ and $2.9$, just above the detection-damage band, which spans $\theta \approx 0.4$--$2.0$ at this calibration), and reversion to PPI++ as the violation becomes plainly detectable. For $c = 1.5$ the worst-case excess over the oracle envelope occurs at $h = 0$, where the hybrid concedes most to an oracle that would pick SI. Because $h$ is the limiting estimator gap and the calibration is taken from the simulated replications, the attenuation factor $\kappa$ of Proposition~\ref{prop:si-bias} is already absorbed into every quantity plotted and tabulated here.

\subsection{Tuning the Constant \texorpdfstring{$c$}{c}}
\label{app:hybrid-tuning}

Define the excess risk $R_{\mathrm{excess}}(c, \rho) = \mathrm{MSE}(\hat{\tau}_{\mathrm{hybrid}}; c, \rho) - \min(\mathrm{MSE}(\hat{\tau}_{\mathrm{SI}}),\, \mathrm{MSE}(\hat{\tau}_{\mathrm{PPI++}}))$ and choose $c^* = \arg\min_c \max_\rho R_{\mathrm{excess}}(c, \rho)$. On DGP~1 parameters, sweeping 50 values of $c$ in $[0.1, 5.0]$ against 17 values of $\rho$ in $[0, 0.8]$, the finite-sample minimax $c^*$ is $1.7$ under default prediction quality ($R^2 = 0.56$) and $1.4$ under low prediction quality ($R^2 = 0.25$) (Figure~\ref{fig:minimax}). The objective is flat near its optimum: over $\pm 0.2$ around $c^*$ the worst-case excess varies by 11.5\% and 13.4\%, and the five best values of $c$ are $\{1.4, \ldots, 1.8\}$ and $\{1.1, \ldots, 1.5\}$, so $c = 1.5$ is in both windows. It lies between the two optima, and the smaller value shades the weight toward the PPI++ endpoint, the safer direction in the undetectable-violation regime. At $c = 1.5$ and default prediction quality the maximum excess occurs at $\rho = 0$ and is about $1.3$ times the smaller component MSE. The Cauchy kernel has the functional form of a Focused Information Criterion weight \citep{hjort2003frequentist}. The weight is smooth in $T_n$, satisfies $0 < w \leq 1$, and equals 1 only at $T_n = 0$, avoiding the discontinuities of hard selection.

\begin{figure}[!htbp]
\centering
\includegraphics[width=0.95\textwidth]{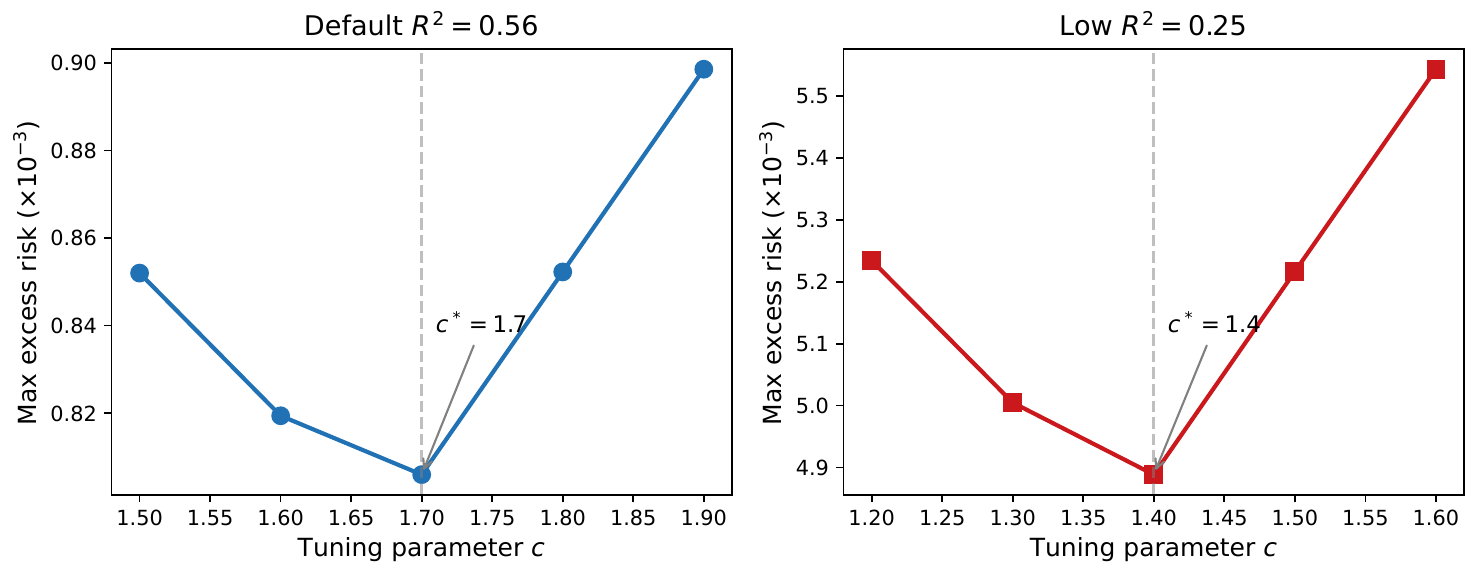}
\caption{Minimax analysis for the hybrid tuning parameter $c$.}
\label{fig:minimax}
\smallskip

{\scriptsize \noindent Notes: $n = 10{,}000$, $\pi_L = 0.20$, $R = 500$. Each panel plots $\max_\rho R_{\mathrm{excess}}(c, \rho)$ over the $\pm 0.2$ neighborhood of the sweep's minimizer; the sweep itself covers 50 values of $c$ in $[0.1, 5.0]$. Left panel: default prediction quality ($R^2 = 0.56$), minimax-optimal $c^* = 1.7$, shown on $\{1.5, 1.6, 1.7, 1.8, 1.9\}$. Right panel: low prediction quality ($R^2 = 0.25$), $c^* = 1.4$, shown on $\{1.2, 1.3, 1.4, 1.5, 1.6\}$. The fixed default $c = 1.5$ lies between the two optima and is conservative in the default-quality regime.}
\end{figure}

The limit experiment complements the grids. Figure~\ref{fig:limit-risk} plots $r(h; c)$ against the standardized violation $\theta$ for $c \in \{0.5, 1.5, 1.6, 3.0\}$, with the SI risk $\sigma_S^2 + h^2$, the PPI++ risk $\sigma_P^2$, the oracle envelope $\min(\sigma_P^2, \sigma_S^2 + h^2)$, and finite-sample checks at $n = 10{,}000$, which fall below the limit formula by at most 8.4\% (Appendix~\ref{app:hybrid-proofs}). Larger $c$ reduces risk under the null but increases peak risk at intermediate violations. The limit-minimax optimum over $c \in [0.3, 4]$ is $c^* = 1.6$, and $c = 1.5$ attains a worst-case excess of $7.880$ against $7.663$, a ratio of $1.028$ (Table~\ref{tab:limit-minimax-c}); these limit-experiment excesses are in units of $n\,\mathrm{MSE}(\hat{\tau})$ over the local parameter $h$ and are not comparable in magnitude with the finite-sample grid, although both place the worst case at the null. The recommended $c = 1.5$ thus sits between the finite-sample optima $1.4$ and $1.7$ and within 2.8\% of the limit-minimax $1.6$.

\begin{table}[!htbp]
\centering
\caption{Worst-case excess risk of the fixed-$c$ hybrid in the limit experiment.}
\label{tab:limit-minimax-c}
\footnotesize
\begin{tabular}{@{}rrr@{}}
\toprule
$\boldsymbol{c}$ & \textbf{Worst-case excess} & \textbf{Ratio to} $\boldsymbol{c^*}$ \\
\midrule
0.5 & $13.120$ & $1.712$ \\
1.0 & $9.866$ & $1.288$ \\
1.5 & $7.880$ & $1.028$ \\
1.6 ($c^*$) & $7.663$ & $1.000$ \\
1.7 & $8.098$ & $1.057$ \\
2.0 & $9.398$ & $1.226$ \\
3.0 & $13.825$ & $1.804$ \\
\bottomrule
\end{tabular}
\smallskip

{\scriptsize \noindent Notes: excess risk is $\max_h [r(h;c) - \min(\sigma_P^2, \sigma_S^2 + h^2)]$, the largest gap between the Cauchy-kernel hybrid's asymptotic risk of Proposition~\ref{prop:limit-risk} and the oracle envelope that switches between PPI++ and SI. Calibrated at DGP~2, $n = 10{,}000$, $\pi_L = 0.20$ ($\sigma_S = 2.531$, $\sigma_P = 5.214$, $\sigma_D = 4.465$, from $R = 600$ replications of that cell), on an 81-point $h$ grid spanning $[0, 35.72]$. Where the worst case falls depends on $c$: for $c \leq 1.5$ it is at $h = 0$, where the hybrid concedes most to an oracle that would pick SI; from $c = 1.6$ upward it moves onto the raw-risk hump, at $\theta = h/\sigma_D \approx 2.9$ for $c = 1.6$ and $1.7$, $3.0$ for $c = 2.0$, and $3.1$ for $c = 3.0$. The limit-minimax $c^* = 1.6$ is the crossover: its local maximum at $h = 0$ is $7.569$ and its hump maximum is $7.663$, while at $c = 1.5$ the ordering is reversed, $7.880$ at $h = 0$ against $7.228$ at the hump. Units are those of $n\,\mathrm{MSE}(\hat{\tau})$ and are not comparable with the finite-sample excess-risk grids earlier in this appendix.}
\end{table}

Because $c^*$ depends on $(\sigma_S, \sigma_P, \sigma_{SP})$, the rule can be made covariance-adaptive: estimate them from the SI and PPI++ variance estimates and their covariance, and set $c$ to the limit-minimax value at those estimates, a grid search over one-dimensional integrals. Across six DGP cells (Table~\ref{tab:adaptive-c}) the estimated $c^*$ concentrates tightly (per-cell means $1.52$--$1.55$), and the adaptive hybrid reproduces the fixed-$c$ hybrid to two decimals in bias, RMSE, and RE, so the fixed default already sits where the data-driven choice lands in the settings tested.

\begin{table}[!htbp]
\centering
\caption{Covariance-adaptive $c$ versus fixed $c = 1.5$ ($n = 10{,}000$, $R = 500$ per cell).}
\label{tab:adaptive-c}
\footnotesize
\setlength{\tabcolsep}{4pt}
\begin{tabular}{@{}llrrrrr@{}}
\toprule
\textbf{Cell} & \textbf{Rule} & \textbf{Mean $c$} & \textbf{Bias} & \textbf{RMSE} & \textbf{Coverage} & \textbf{RE} \\
\midrule
DGP 1, $\pi_L = 0.05$ & fixed & 1.50 & $+0.005$ & $0.057$ & $0.958$ & $2.47$ \\
 & adaptive & 1.53 & $+0.005$ & $0.057$ & $0.958$ & $2.48$ \\
DGP 1, $\pi_L = 0.20$ & fixed & 1.50 & $+0.001$ & $0.036$ & $0.962$ & $1.91$ \\
 & adaptive & 1.53 & $+0.001$ & $0.036$ & $0.962$ & $1.91$ \\
DGP 1, $\pi_L = 0.50$ & fixed & 1.50 & $+0.001$ & $0.029$ & $0.946$ & $1.49$ \\
 & adaptive & 1.53 & $+0.001$ & $0.029$ & $0.946$ & $1.49$ \\
DGP 2, $\rho = 0.2$ & fixed & 1.50 & $-0.012$ & $0.042$ & $0.944$ & $1.64$ \\
 & adaptive & 1.52 & $-0.012$ & $0.042$ & $0.942$ & $1.64$ \\
DGP 2, $\rho = 0.4$ & fixed & 1.50 & $-0.019$ & $0.055$ & $0.940$ & $1.25$ \\
 & adaptive & 1.52 & $-0.019$ & $0.055$ & $0.940$ & $1.25$ \\
DGP 9 & fixed & 1.50 & $-0.004$ & $0.049$ & $0.938$ & $1.21$ \\
 & adaptive & 1.55 & $-0.004$ & $0.049$ & $0.940$ & $1.21$ \\
\bottomrule
\end{tabular}
\smallskip

{\scriptsize \noindent Notes: DGP 2 cells at $\pi_L = 0.20$; $R = 500$ per cell, all-units five-fold cross-fitting. Adaptive $c$ minimizes the worst-case excess of the Proposition~\ref{prop:limit-risk} risk over the oracle envelope, on a grid $c \in [0.3, 4]$, at the replication's estimated covariance. RE is the RMSE ratio versus labeled-only. The Monte Carlo band at $R = 500$ is $[0.931, 0.969]$ and every entry is inside it. The PPI++ input uses the exact variance. The run is independent of the $R = 2{,}000$ hybrid evaluation of Section~\ref{sec:results-hybrid}, which gives RE $2.22$ at DGP~1, $\pi_L = 0.05$, against $2.47$ here.}
\end{table}

\begin{figure}[!ht]
\centering
\includegraphics[width=0.82\textwidth]{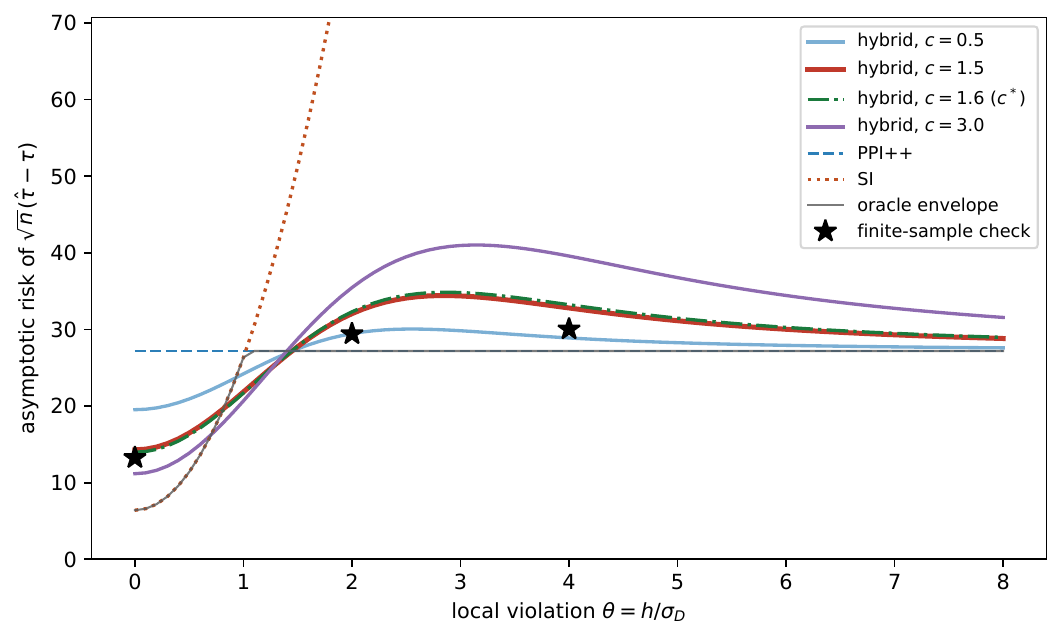}
\caption{Limit-experiment risk of the fixed-$c$ hybrid (Proposition~\ref{prop:limit-risk}), calibrated to the featured simulation cell.}
\label{fig:limit-risk}
\smallskip

{\scriptsize \noindent Notes: calibration $\sigma_S = 2.53$, $\sigma_P = 5.21$, $\mathrm{corr}_{SP} = 0.52$. Risk curves are plotted for $c \in \{0.5, 1.5, 1.6, 3.0\}$; the $c = 1.6$ curve is the limit-minimax optimum of Table~\ref{tab:limit-minimax-c} and lies essentially on top of the recommended $c = 1.5$ curve, marginally below it at $h = 0$ and marginally above it at the hump (worst-case excess ratio $1.028$). Stars are finite-sample $n \times \mathrm{MSE}$ checks at $n = 10{,}000$ with $c = 1.5$.}
\end{figure}

\subsection{Hybrid Simulation Results}
\label{app:hybrid-sims}

The canonical hybrid run is the $R = 2{,}000$ evaluation on DGPs~1, 2, and~9 (replication package). Under valid surrogacy (DGP~1) the hybrid delivers RE $2.0$--$2.2\times$ over labeled-only at $\pi_L = 0.05$--$0.20$, falling to $1.5\times$ at $\pi_L = 0.50$, with 94--96\% coverage; DGP~2 at $\rho = 0$, the same model on independent draws, gives $1.9$--$2.3\times$. The average weight is about $0.7$. Adaptivity has a finite-sample cost: the hybrid does not reach the surrogate index's own efficiency ($5.54\times$ at $\pi_L = 0.05$ and $2.91\times$ at $\pi_L = 0.20$, Table~\ref{tab:dgp1}), because it hedges against a violation.

Under moderate violation (DGP~2 at $\pi_L = 0.20$) the average weight falls from 0.63 at $\rho = 0.2$ to 0.30 at $\rho = 0.4$ as $T_n$ grows. The hybrid retains part of the SI bias (relative bias $-8.4\%$ and $-9.0\%$) with coverage of 93.1--93.4\%, far above SI's but below nominal. At $\rho = 0.2$ it still has the lowest RMSE of the three estimators; by $\rho = 0.4$ PPI++ dominates it (RMSE 0.050 against 0.056), the raw-risk hump of Proposition~\ref{prop:limit-risk} appearing in finite samples. Under antagonistic surrogacy (DGP~9) the weight is near zero in every replication (mean $0.012$ at $\pi_L = 0.20$), so the hybrid reverts to PPI++ and avoids SI's wrong-sign estimate. These results describe the detectable-violation regime; Hillstrom illustrates the undetectable one (Section~\ref{sec:hillstrom}). The coverage figures are reported, not guaranteed (Remark~\ref{rem:ci-validity}).

Against a nonparametric alternative, the simulation-calibrated interval does at least as well at a fraction of the cost. On DGP~1 at $\pi_L = 0.20$ it covers 96.2\% of replications and a percentile bootstrap with $B = 200$ covers 95.0\%, with mean widths $0.163$ and $0.160$; on DGP~2 at $\rho = 0.2$ the two cover 93.2\% and 92.4\%, with widths $0.170$ and $0.167$ ($R = 500$, $M = 10{,}000$, $c = 1.5$). Both under-cover in the violation cell, and the comparison does not isolate the cause.

The featured violation cell (DGP~2, $\rho = 0.2$, $\pi_L = 0.20$) appears in several runs. Hybrid coverage is $0.934$ in the canonical evaluation and in the hard-switch comparison (Table~\ref{tab:hard-switch}), $0.952$ on the $R = 1{,}000$ sensitivity grid (Table~\ref{tab:hybrid-sensitivity}), and $0.932$ in the $R = 500$ interval validation. SI coverage is $0.682$ in Table~\ref{tab:dgp2}, $0.680$ in the paired run of Section~\ref{sec:detection-damage}, $0.654$ in Table~\ref{tab:hard-switch}, and $0.638$ in Table~\ref{tab:cuped-full}. Within each estimator, no two values differ by more than about two Monte Carlo standard errors of the difference, and the text quotes the canonical runs.

\subsection{Sensitivity in \texorpdfstring{$c$}{c} and the Hard-Switch Comparison}
\label{app:hybrid-sensitivity}

The sensitivity analysis (Table~\ref{tab:hybrid-sensitivity}) confirms that $c = 1.5$ balances efficiency under valid surrogacy ($1.95\times$ RE) against coverage under violations ($95.2\%$ at $\rho = 0.2$ and $92.9\%$ at $\rho = 0.4$ on the sensitivity grid).

A comparison with the naive hard-switch strategy (Table~\ref{tab:hard-switch}) shows the switch suffers from pre-test bias: its coverage drops to 69.9\% at $\rho = 0.2$ and 70.6\% at $\rho = 0.4$, while the hybrid holds 93.4\% and 93.1\% in the same cells. At $\rho = 0.2$ the switch has higher RMSE than either of the estimators it selects between, $0.0515$ against SI's $0.0452$ and PPI++'s $0.0505$, and its coverage of 69.9\% sits barely above SI's 65.4\% and far below PPI++'s 95.9\%: conditioning on the pre-test mixes the two sampling distributions without inheriting the good properties of either.

\begin{table}[!ht]
\centering
\caption{Hybrid estimator sensitivity to $c$.}
\label{tab:hybrid-sensitivity}
\footnotesize
\begin{tabular}{@{}llrrr@{}}
\toprule
\textbf{DGP} & $\boldsymbol{c}$ & \textbf{Coverage} & \textbf{RE vs LO} & \textbf{Mean $w$} \\
\midrule
DGP 1 ($\rho=0$) & 0.5 & 0.959 & $1.66\times$ & 0.55 \\
DGP 1 ($\rho=0$) & 1.0 & 0.953 & $1.83\times$ & 0.66 \\
DGP 1 ($\rho=0$) & 1.5 & 0.952 & $1.95\times$ & 0.72 \\
DGP 1 ($\rho=0$) & 2.0 & 0.954 & $2.06\times$ & 0.76 \\
DGP 1 ($\rho=0$) & 3.0 & 0.954 & $2.21\times$ & 0.81 \\
\midrule
DGP 2 ($\rho=0.2$) & 0.5 & 0.961 & $1.57\times$ & 0.45 \\
DGP 2 ($\rho=0.2$) & 1.0 & 0.957 & $1.64\times$ & 0.56 \\
DGP 2 ($\rho=0.2$) & 1.5 & 0.952 & $1.68\times$ & 0.62 \\
DGP 2 ($\rho=0.2$) & 2.0 & 0.946 & $1.72\times$ & 0.67 \\
DGP 2 ($\rho=0.2$) & 3.0 & $\mathbf{0.933}$ & $1.75\times$ & 0.73 \\
\midrule
DGP 2 ($\rho=0.4$) & 0.5 & 0.939 & $1.31\times$ & 0.16 \\
DGP 2 ($\rho=0.4$) & 1.0 & $\mathbf{0.931}$ & $1.27\times$ & 0.25 \\
DGP 2 ($\rho=0.4$) & 1.5 & $\mathbf{0.929}$ & $1.24\times$ & 0.31 \\
DGP 2 ($\rho=0.4$) & 2.0 & $\mathbf{0.918}$ & $1.21\times$ & 0.36 \\
DGP 2 ($\rho=0.4$) & 3.0 & $\mathbf{0.893}$ & $1.16\times$ & 0.44 \\
\bottomrule
\end{tabular}
\smallskip

{\scriptsize \noindent Notes: $n = 10{,}000$, $\pi_L = 0.20$, $R = 1{,}000$, all-units five-fold cross-fitting. The SI--PPI++ estimator-disagreement diagnostic uses $\alpha = 0.05$, two-sided, with $\mathrm{SE}(\hat{D})$ from the joint sandwich; the PPI++ input is the primary configuration. \textbf{Bold} coverage entries fall outside the $R = 1{,}000$ Monte Carlo band $[0.936, 0.964]$. RE vs LO is the RMSE ratio, consistent with the convention of Section~\ref{sec:eval-metrics}. RE is reported for every row so that the high-$c$ rows remain comparable with the rest; an RE entry alongside a double-digit relative bias is a ratio of RMSEs against a biased estimator, not an efficiency gain. Three rows are in that category: DGP~2 at $\rho = 0.2$ with $c = 3.0$ ($-10.4\%$ relative bias), and DGP~2 at $\rho = 0.4$ with $c = 2.0$ ($-10.8\%$) and $c = 3.0$ ($-13.8\%$).}
\end{table}

\begin{table}[!ht]
\centering
\caption{Hard-switch vs.\ hybrid estimator.}
\label{tab:hard-switch}
\footnotesize
\setlength{\tabcolsep}{4pt}
\begin{tabular}{@{}llrrrr@{}}
\toprule
\textbf{Setting} & \textbf{Estimator} & \textbf{Bias} & \textbf{RMSE} & \textbf{Coverage} & \textbf{Reject} \\
\midrule
DGP 1, $\pi_L = 0.05$ & SI & $-0.0002$ & $0.0254$ & $0.950$ & \multirow{4}{*}{$0.106$} \\
 & PPI++ & $+0.0008$ & $0.0965$ & $0.941$ & \\
 & Hard Switch & $+0.0011$ & $0.0684$ & $\mathbf{0.898}$ & \\
 & Hybrid ($c = 1.5$) & $+0.0004$ & $0.0641$ & $0.941$ & \\
\midrule
DGP 1, $\pi_L = 0.20$ & SI & $+0.0004$ & $0.0246$ & $0.949$ & \multirow{4}{*}{$0.091$} \\
 & PPI++ & $+0.0003$ & $0.0501$ & $0.955$ & \\
 & Hard Switch & $+0.0005$ & $0.0377$ & $\mathbf{0.921}$ & \\
 & Hybrid ($c = 1.5$) & $+0.0002$ & $0.0363$ & $0.960$ & \\
\midrule
DGP 2, $\rho = 0.0$ & SI & $+0.0006$ & $0.0254$ & $0.942$ & \multirow{4}{*}{$0.096$} \\
 & PPI++ & $+0.0009$ & $0.0519$ & $0.944$ & \\
 & Hard Switch & $+0.0004$ & $0.0392$ & $\mathbf{0.908}$ & \\
 & Hybrid ($c = 1.5$) & $+0.0007$ & $0.0379$ & $0.951$ & \\
\midrule
DGP 2, $\rho = 0.2$ & SI & $-0.0380$ & $0.0452$ & $\mathbf{0.654}$ & \multirow{4}{*}{$0.203$} \\
 & PPI++ & $-0.0022$ & $0.0505$ & $0.959$ & \\
 & Hard Switch & $-0.0188$ & $0.0515$ & $\mathbf{0.699}$ & \\
 & Hybrid ($c = 1.5$) & $-0.0157$ & $0.0435$ & $\mathbf{0.934}$ & \\
\midrule
DGP 2, $\rho = 0.4$ & SI & $-0.0996$ & $0.1027$ & $\mathbf{0.022}$ & \multirow{4}{*}{$0.716$} \\
 & PPI++ & $-0.0009$ & $0.0501$ & $\mathbf{0.962}$ & \\
 & Hard Switch & $-0.0143$ & $0.0661$ & $\mathbf{0.706}$ & \\
 & Hybrid ($c = 1.5$) & $-0.0224$ & $0.0561$ & $\mathbf{0.931}$ & \\
\midrule
DGP 9 & SI & $-0.4973$ & $0.4975$ & $\mathbf{0.000}$ & \multirow{4}{*}{$1.000$} \\
 & PPI++ & $+0.0026$ & $0.0477$ & $0.952$ & \\
 & Hard Switch & $+0.0026$ & $0.0477$ & $0.952$ & \\
 & Hybrid ($c = 1.5$) & $-0.0033$ & $0.0482$ & $0.949$ & \\
\bottomrule
\end{tabular}
\smallskip

{\scriptsize \noindent Notes: $n = 10{,}000$, $R = 2{,}000$, all-units five-fold cross-fitting. The Monte Carlo band around the nominal 95\% is $[0.940, 0.960]$ and \textbf{bold} coverage entries fall outside it. All rows use $\pi_L = 0.20$ except the first block. PPI++ is the primary configuration and SI uses the joint sandwich variance. The hard switch uses SI when the two-sided diagnostic fails to reject at $\alpha = 0.10$ and PPI++ otherwise; the ``Reject'' column is the diagnostic's rejection rate in that setting, which equals the fraction of replications on which the switch selects PPI++. The hybrid uses $c = 1.5$. No labeled-only arm is run in this comparison, so no RE column is reported. True $\tau$ is $0.15$ for DGP~1 and DGP~2 at $\rho = 0$, $0.1875$ at $\rho = 0.2$, $0.25$ at $\rho = 0.4$, and $0.41$ for DGP~9. This run and Table~\ref{tab:dgp2} use separate seeds; the PPI++ bias at DGP~2, $\rho = 0.2$, is $-0.0022$ here and $+0.002$ there, each within two Monte Carlo standard errors of zero and about $2.5$ standard errors of the difference apart.}
\end{table}

\subsection{The Diverging-\texorpdfstring{$c_n$}{cn} Variant}
\label{app:diverging-c}

This subsection states and proves the two asymptotic properties of the diverging-$c_n$ hybrid summarized in Section~\ref{app:hybrid-construction}. Neither applies to the fixed-$c$ implementation used in the paper's results; they are recorded because they delimit what an adaptive weight can and cannot achieve. Throughout, $w_n = c_n/(c_n + T_n^2)$ with $c_n \to \infty$.

\begin{proposition}[Efficiency under $H_0$, diverging $c_n$]
\label{prop:hybrid-efficiency}
Under the conditions of Proposition~\ref{prop:surrogacy-test} and $H_0$, if $c_n \to \infty$ then $w_n \xrightarrow{p} 1$ and $\sqrt{n}(\hat{\tau}_{\mathrm{hybrid}} - \hat{\tau}_{\mathrm{SI}}) \xrightarrow{p} 0$, so the two estimators share a limiting distribution; under a uniform-integrability condition on $n(\hat{\tau}_{\mathrm{hybrid}} - \tau)^2$ the variance ratio $\Var(\hat{\tau}_{\mathrm{hybrid}})/\Var(\hat{\tau}_{\mathrm{SI}}) \to 1$ as well. The choice $c_n = (\log n)^2$ is one admissible sequence. The same conclusion holds under local alternatives $\delta_n = h/\sqrt{n}$: since $T_n = O_p(1)$ there too, $w_n \xrightarrow{p} 1$ and the diverging-$c_n$ hybrid converges to the surrogate index and inherits SI's local bias in full. A nondegenerate weight limit under local alternatives is a property of the fixed-$c$ hybrid (Proposition~\ref{prop:limit-risk}), not of this variant.
\end{proposition}

\begin{proof}
Under $H_0$, $T_n \xrightarrow{d} \mathcal{N}(0,1)$, so $T_n^2 = O_p(1)$ and $w_n = c_n/(c_n + T_n^2) = 1 - O_p(1/c_n) \xrightarrow{p} 1$. Since $\hat{\tau}_{\mathrm{hybrid}} - \hat{\tau}_{\mathrm{SI}} = (1 - w_n)(\hat{\tau}_{\mathrm{PPI++}} - \hat{\tau}_{\mathrm{SI}}) = O_p(c_n^{-1} n^{-1/2})$, the two are equivalent at the $\sqrt{n}$ scale, and the stated uniform-integrability condition upgrades that distributional equivalence to convergence of the second moments, hence of the variance ratio. Under $\delta_n = h/\sqrt{n}$, Proposition~\ref{prop:limit-risk}'s argument gives $T_n \xrightarrow{d} \mathcal{N}(h/\sigma_D, 1)$, still $O_p(1)$, so the same computation yields $w_n \xrightarrow{p} 1$.
\end{proof}

\begin{proposition}[Root-$n$ safety under a fixed alternative, diverging $c_n$]
\label{prop:hybrid-safety}
Fix a violation and let $\delta = \plim \hat{D} \neq 0$ be the \emph{estimator gap} it induces, which in the setting of Proposition~\ref{prop:si-bias} is $\kappa\,\delta_{\mathrm{direct}}$, the direct effect times the attenuation factor, rather than the direct effect itself. Suppose $\mathrm{SE}(\hat{D}) = \sigma_D n^{-1/2}(1 + o_p(1))$ with $\sigma_D > 0$, so that $T_n^2 = n\delta^2/\sigma_D^2\,(1 + o_p(1))$. Then, with $\hat{D} = \hat{\tau}_{\mathrm{PPI++}} - \hat{\tau}_{\mathrm{SI}}$ as defined in Section~\ref{sec:surrogacy-test}:
\begin{enumerate}[leftmargin=*,itemsep=2pt]
  \item $w_n \xrightarrow{p} 0$ if and only if $c_n/n \to 0$. Detectability alone, in the sense $|T_n| \to \infty$, is not sufficient: it leaves $w_n$ bounded away from zero whenever $c_n$ grows at rate $n$ or faster.
  \item $\sqrt{n}(\hat{\tau}_{\mathrm{hybrid}} - \hat{\tau}_{\mathrm{PPI++}}) = -\,w_n \sqrt{n}\,\hat{D} \xrightarrow{p} 0$, so the hybrid is root-$n$ equivalent to PPI++ and shares its limiting distribution, if and only if $c_n = o(\sqrt{n})$.
\end{enumerate}
Part~2 is the condition under which the hybrid inherits PPI++'s inference, and it is strictly stronger than part~1.
\end{proposition}

\begin{proof}
Write $w_n = c_n/(c_n + T_n^2)$ with $T_n^2 = n\delta^2/\sigma_D^2(1 + o_p(1))$. Then $w_n = (1 + n\delta^2/(c_n\sigma_D^2))^{-1}(1 + o_p(1))$, which tends to zero exactly when $n/c_n \to \infty$, proving part~1. For part~2, the hybrid satisfies $\hat{\tau}_{\mathrm{hybrid}} - \hat{\tau}_{\mathrm{PPI++}} = w_n(\hat{\tau}_{\mathrm{SI}} - \hat{\tau}_{\mathrm{PPI++}}) = -w_n\hat{D}$, and $\sqrt{n}\hat{D} = \sqrt{n}\delta + O_p(1)$, so
\[
  \bigl|w_n\sqrt{n}\hat{D}\bigr| = \frac{c_n\sigma_D^2}{n\delta^2}\bigl|\sqrt{n}\delta + O_p(1)\bigr|(1 + o_p(1)) = \frac{c_n \sigma_D^2}{|\delta|\sqrt{n}}(1 + o_p(1)),
\]
which tends to zero if and only if $c_n/\sqrt{n} \to 0$. Combining with Proposition~\ref{prop:hybrid-efficiency}, a sequence satisfying both requirements, for instance $c_n = (\log n)^2$, gives SI efficiency under $H_0$ and root-$n$ PPI++ equivalence under a fixed violation, while conceding the local-alternative regime entirely.
\end{proof}

\section{Proofs of the Structural Properties}
\label{app:proofs}

\subsection{Proof of Proposition~\ref{prop:si-bias} (Surrogate Index Bias under Partial Mediation)}

\begin{proof}
Write the surrogate index estimator as
$$\hat{\tau}^{(2)} = \frac{1}{n_1}\sum_{i: T_i=1} \hat{f}(S_i, X_i) - \frac{1}{n_0}\sum_{i: T_i=0} \hat{f}(S_i, X_i),$$
where $\hat{f}$ converges in probability to the population linear projection of $Y$ on $(1, S, X)$, which we write $f^*(S, X) = a_0 + a_S S + a_X X$. The projection omits $T$, and since $S$ responds to $T$ the omitted term $\delta T$ is reallocated to $(a_S, a_X)$; the first step is to compute that reallocation.

\emph{Step 1: the population OLS coefficients.} By the omitted-variable formula applied to the population projection,
$$(a_S, a_X) = (\beta_{YS}, \beta_{YX}) + \delta \cdot (a, b),$$
where $(a, b)$ are the coefficients in the population linear projection of $T$ on $(1, S, X)$. Write $q = p(1-p) = \Var(T_i)$, $\sigma_X^2 = \Var(X_i)$, and center all variables. The surrogate equation gives $\Cov(T, S) = \gamma_S q$, $\Var(S) = \beta_{SX}^2\sigma_X^2 + \gamma_S^2 q + \sigma_{\varepsilon_S}^2$, and $\Cov(S, X) = \beta_{SX}\sigma_X^2$, while randomization gives $\Cov(T, X) = 0$. The normal equations
$$\begin{pmatrix} \Var(S) & \Cov(S,X) \\ \Cov(S,X) & \sigma_X^2 \end{pmatrix}\begin{pmatrix} a \\ b\end{pmatrix} = \begin{pmatrix} \gamma_S q \\ 0 \end{pmatrix}$$
have determinant $\sigma_X^2(\gamma_S^2 q + \sigma_{\varepsilon_S}^2)$, so by Cramer's rule
$$a = \frac{\gamma_S q\,\sigma_X^2}{\sigma_X^2(\sigma_{\varepsilon_S}^2 + \gamma_S^2 q)} = \frac{\gamma_S\, p(1-p)}{\sigma_{\varepsilon_S}^2 + \gamma_S^2 p(1-p)},
\qquad b = -\beta_{SX}\,a.$$
The covariate drops out of $a$: because $\Cov(T, X) = 0$, adding $X$ to the projection changes neither the numerator nor the determinant's dependence on $\beta_{SX}$, so the partial coefficient of $T$ on $S$ is the same as in the projection of $T$ on $S$ alone after partialling out $X$. This gives $a_S = \beta_{YS} + \delta a$, which is \eqref{eq:pooled-slope}, and $a_X = \beta_{YX} - \beta_{SX}\delta a$.

\emph{Step 2: the surrogate index limit.} By the law of large numbers and the continuous mapping theorem,
$$\hat{\tau}^{(2)} \xrightarrow{p} a_S\bigl(\E[S_i | T_i = 1] - \E[S_i | T_i = 0]\bigr) + a_X\bigl(\E[X_i | T_i = 1] - \E[X_i | T_i = 0]\bigr).$$
By randomization the second difference is zero, and the surrogate equation gives $\E[S_i | T_i = 1] - \E[S_i | T_i = 0] = \gamma_S$. Therefore
$$\hat{\tau}^{(2)} \xrightarrow{p} \gamma_S\bigl(\beta_{YS} + \delta a\bigr) = \beta_{YS}\gamma_S + \delta\,\frac{\gamma_S^2 q}{\sigma_{\varepsilon_S}^2 + \gamma_S^2 q} = \tau - \delta\left(1 - \frac{\gamma_S^2 q}{\sigma_{\varepsilon_S}^2 + \gamma_S^2 q}\right) = \tau - \delta\kappa$$
with $\kappa = \sigma_{\varepsilon_S}^2/(\sigma_{\varepsilon_S}^2 + \gamma_S^2 q)$ as claimed.

\emph{Step 3: the oracle prediction function.} If $\hat{f}$ is replaced by $m(s,x) = \beta_{YS}s + \beta_{YX}x$, Step 1 is bypassed: $a_S = \beta_{YS}$ and the same argument gives $\plim \hat{\tau}^{(2)} = \beta_{YS}\gamma_S = \tau - \delta$, so the bias is exactly $-\delta$. The gap between the two statements is entirely the reallocation of $\delta T$ onto the pooled surrogate coefficient, and it vanishes as $\gamma_S^2 p(1-p)/\sigma_{\varepsilon_S}^2 \to 0$.
\end{proof}

\subsection{Proof of Proposition~\ref{prop:ppi-validity} (Unbiasedness of PPI at a Fixed Predictor)}

\begin{proof}
Condition on the treatment vector and on the labeled counts $n_{L,t} \geq 1$. Because $T$ is independent of the potential outcomes and the units are i.i.d., the units of arm $t$ are an i.i.d.\ sample from the distribution of $(X, S(t), Y(t))$, and $\hat{Y}_i = f(S_i, X_i)$ is a fixed function of each unit's own data. Under Assumption~\ref{asm:mcar} the labeled units of arm $t$ are a simple random subset of that arm of size $n_{L,t}$, drawn independently of the units' data. A simple random subset has the same expected mean as the set it is drawn from, so
$$\E\bigl[\bar{Y}_{L,t}\bigr] = \E[Y(t)], \qquad \E\bigl[\bar{\hat{Y}}_{L,t}\bigr] = \E\bigl[\bar{\hat{Y}}_t\bigr].$$
The first identity gives $\E[\hat{\tau}_Y^{\calL}] = \E[Y(1)] - \E[Y(0)] = \tau$, and the second, differenced across arms, gives $\E[\hat{\tau}_{\hat{Y}}^{\mathrm{all}} - \hat{\tau}_{\hat{Y}}^{\calL}] = 0$. Hence $\E[\hat{\tau}^{\mathrm{PPI}(\lambda)}] = \tau + \lambda \cdot 0 = \tau$ at every fixed $\lambda$, and the unconditional statement follows by the tower property.

For the fold-level statement, fix fold $k$ and condition on the fold assignment, the treatment vector, $\hat\beta_{(-k)}$, and the labeled count $n_{L,tk}$ of fold $k$ in each arm. The fit $\hat\beta_{(-k)}$ is a function of the data and labels of units outside fold $k$ only. The units of fold $k$ are independent of those, and under Assumption~\ref{asm:mcar} their labels, given $n_{L,tk}$, form a simple random subset of the fold's units in arm $t$ whatever the labels outside the fold. Every unit of fold $k$ is scored by the coefficient $\hat\beta_{(-k)}$, which is fixed under this conditioning, so the argument above applies within the fold and the fold's labeled mean of $\hat{Y}$ in arm $t$ has the conditional expectation of its all-units mean.
\end{proof}

\subsection{Proof of Proposition~\ref{prop:joint-si-ppi} (Joint Distribution with a Learned Linear Index)}

\begin{proof}
Throughout, $\hat{Q}_{(-k)}$ and $\hat\beta_{(-k)}$ are the out-of-fold Gram matrix and OLS coefficient of Equation~\eqref{eq:fold-beta}, $n_L = \sum_i L_i$, and $\mathcal{L}$ is the labeled set. All limits are as $n \to \infty$ with $K$, $p$, and $\pi_L$ fixed.

\emph{Step (i): decompose the surrogate index.} By the identity~\eqref{eq:si-fold-average},
$$\hat{\tau}_{\mathrm{SI}} = \sum_{k} \frac{n_k}{n}\,\hat{d}_k'\hat\beta_{(-k)} = \sum_k \frac{n_k}{n}\Bigl[\hat{d}_k'\beta_0 + d_0'(\hat\beta_{(-k)} - \beta_0) + (\hat{d}_k - d_0)'(\hat\beta_{(-k)} - \beta_0)\Bigr].$$
The first term is $\beta_0'\hat d$ with $\hat d = \sum_k (n_k/n)\hat d_k = \bar g_1 - \bar g_0$, an average of i.i.d.\ terms. The second is $d_0'$ times the fold average of the first-stage errors. The third is a sum of products of two $O_p(n^{-1/2})$ factors, because $\hat d_k - d_0 = O_p(n^{-1/2})$ (each fold holds $n/K$ units) and $\hat\beta_{(-k)} - \beta_0 = O_p(n^{-1/2})$ under the assumed moments and positive definiteness of $Q$ on the column space; it is therefore $O_p(1/n)$ and does not enter the limit.

\emph{Step (ii): the fold average of the first-stage errors.} A standard expansion of the out-of-fold OLS coefficient gives
$$\hat\beta_{(-k)} - \beta_0 = Q^{-1}\,\frac{1}{n_{L,(-k)}}\sum_{j \in \calL \setminus \mathcal{I}_k} g_j \varepsilon_j + O_p(1/n),$$
where $n_{L,(-k)} = |\calL \setminus \mathcal{I}_k|$ and the remainder collects $(\hat{Q}_{(-k)}^{-1} - Q^{-1})$ times a mean-zero average. Averaging over folds with weights $n_k/n$, each labeled unit $j$ appears in the $K - 1$ out-of-fold training sets that exclude its own fold, and it enters each with weight $1/n_{L,(-k)}$ against a fold weight $n_k/n$. Since the folds are equal blocks, $n_k/n = 1/K$ exactly, and $n_{L,(-k)} = n_L(K-1)/K + O_p(n^{1/2})$, so the total weight on $g_j\varepsilon_j$ is
$$\sum_{k \neq k(j)} \frac{n_k}{n}\,\frac{1}{n_{L,(-k)}} = (K-1)\cdot\frac{1}{K}\cdot\frac{K}{(K-1)n_L}\,\bigl(1 + O_p(n^{-1/2})\bigr) = \frac{1}{n_L}\bigl(1 + O_p(n^{-1/2})\bigr),$$
so the fold weights and the training-set sizes cancel exactly at leading order and
$$\sum_k \frac{n_k}{n}\bigl(\hat\beta_{(-k)} - \beta_0\bigr) = Q^{-1}\,\frac{1}{n_L}\sum_{j \in \calL} g_j\varepsilon_j + O_p(1/n).$$
Because $n_L/n \to \pi_L$ and $L$ is independent of $(T, X, S(0), S(1), Y(0), Y(1))$, hence of $(T, g, \varepsilon)$, $n_L^{-1}\sum_{j \in \calL} g_j \varepsilon_j = n^{-1}\sum_i (L_i/\pi_L) g_i\varepsilon_i + o_p(n^{-1/2})$. Combining with Step~(i),
$$\hat{\tau}_{\mathrm{SI}} - \tau_{\mathrm{SI}}^{0} = \frac{1}{n}\sum_i \left\{A_i\bigl(g_i'\beta_0 - \E[g'\beta_0\mid T_i]\bigr) + \frac{L_i}{\pi_L}\,d_0'Q^{-1}g_i\varepsilon_i\right\} + o_p(n^{-1/2}),$$
where the first summand is the influence function of the all-units difference in means of $g'\beta_0$ with arm shares estimated, which is why the arm-conditional mean is subtracted.

\emph{Step (iii): PPI++ is first-order insensitive to the first stage.} PPI++ is $\hat{\tau}_Y^{\calL} + \lambda(\hat{\tau}_{\hat{Y}}^{\mathrm{all}} - \hat{\tau}_{\hat{Y}}^{\calL})$ with one coefficient $\lambda$ common to both arms, and each unit of fold $k$ is scored by $\hat\beta_{(-k)}$. Write $w^{\mathrm{all}}_k$ for fold $k$'s contribution to the between-arm difference of mean design vectors over all units, $\sum_{i \in \mathcal{I}_k} g_i\{T_i/n_1 - (1 - T_i)/n_0\}$, and $w^{\calL}_k$ for the same sum over the labeled units of fold $k$ with $n_{L,1}$ and $n_{L,0}$ in place of $n_1$ and $n_0$, so that $\sum_k w^{\mathrm{all}}_k = \hat{d}_{\mathrm{all}}$ and $\sum_k w^{\calL}_k = \hat{d}_{\calL}$. Then, exactly,
$$\hat{\tau}_{\hat{Y}}^{\mathrm{all}} - \hat{\tau}_{\hat{Y}}^{\calL} = \sum_k (w^{\mathrm{all}}_k - w^{\calL}_k)'\hat\beta_{(-k)} = (\hat{d}_{\mathrm{all}} - \hat{d}_{\calL})'\beta_0 + \sum_k (w^{\mathrm{all}}_k - w^{\calL}_k)'(\hat\beta_{(-k)} - \beta_0).$$
Under MCAR the labeled units are a simple random subset, so within each fold the labeled and all-units weighted design sums differ by $O_p(n^{-1/2})$, that is, $w^{\mathrm{all}}_k - w^{\calL}_k = O_p(n^{-1/2})$; with $\hat\beta_{(-k)} - \beta_0 = O_p(n^{-1/2})$ and $K$ fixed, the last sum is $O_p(1/n)$. Hence $\hat\tau_{\mathrm{PPI}}$ has the same first-order expansion as the estimator that uses the population index $g'\beta_0$, which is the statement that PPI++ is orthogonal to the first stage at first order. Writing $Z_i = Y_i - \lambda g_i'\beta_0$, that estimator is the labeled difference in means of $Z$ plus $\lambda$ times the all-units difference in means of $g'\beta_0$, whose influence function is Equation~\eqref{eq:phi-ppi}. Its probability limit is $\tau$: randomization gives $\E[Y \mid T = 1] - \E[Y \mid T = 0] = \E[Y(1)] - \E[Y(0)]$, and under MCAR the labeled and all-units arm means of $g'\beta_0$ share their limit, so the $\lambda$ terms cancel. This step uses the independence of $T$ from the baseline covariates and potential outcomes, not from the observed post-treatment $S$ and $Y$.

\emph{Step (iv): the central limit theorem.} Stack the moment vector
$$\theta = \bigl(\beta,\; \mu_{Y,\calL,1},\; \mu_{Y,\calL,0},\; m_{g,\mathrm{all},1},\; m_{g,\mathrm{all},0},\; m_{g,\calL,1},\; m_{g,\calL,0}\bigr),$$
whose components are respectively the labeled-population OLS coefficient, the labeled arm means of $Y$, the all-units arm means of $g$, and the labeled arm means of $g$. Each is an i.i.d.\ average or a smooth function of one, and
\begin{align*}
  \tau_{\mathrm{SI}} &= (m_{g,\mathrm{all},1} - m_{g,\mathrm{all},0})'\beta, \\
  \tau_{\mathrm{PPI}} &= \mu_{Y,\calL,1} - \mu_{Y,\calL,0} + \lambda\bigl[(m_{g,\mathrm{all},1} - m_{g,\mathrm{all},0}) - (m_{g,\calL,1} - m_{g,\calL,0})\bigr]'\beta
\end{align*}
are smooth functions of $\theta$. Steps (i) to (iii) show that the cross-fitting leaves the delta-method expansion of these two functions unchanged up to $o_p(n^{-1/2})$, so $\sqrt{n}$ times the pair of centered estimators equals $n^{-1/2}\sum_i \phi_i + o_p(1)$ with $\phi_i$ given by Equations~\eqref{eq:phi-si} and~\eqref{eq:phi-ppi}. The summands are i.i.d.\ with mean zero and finite second moments under $\E\|g\|^4 < \infty$, $\E Y^4 < \infty$, and positive definiteness of $Q$ on the column space, so the Lindeberg--L\'evy theorem gives Equation~\eqref{eq:joint-clt} with $\Omega = \E[\phi_i\phi_i']$. When $\calL$ is instead a uniformly random subset of fixed size, the summands are exchangeable rather than independent; conditionally on the data the labeled averages are means under sampling without replacement, whose central limit theorem gives the same limit, because for a centered $h$ with variance $\sigma_h^2$ the difference between its labeled and all-units means has variance $(1 - \pi_L)\sigma_h^2/(n\pi_L)$ under either scheme. Consistency of $\hat\Omega$ follows from the law of large numbers applied to the plug-ins together with the continuous mapping theorem.
\end{proof}

\subsection{Proof of Proposition~\ref{prop:efficiency} (Efficiency Ordering)}

\begin{proof}
We work with the exact variance $V(\lambda)$ of Equation~\eqref{eq:exact-variance}, a quadratic in $\lambda$ with leading coefficient $\sum_t a_t v_t$, $a_t = 1/n_{L,t} - 1/n_t \geq 0$ and $v_t = \sigma^2_{\hat{Y},t}$. The labeled-only estimator corresponds to $\lambda = 0$, and $V(0) = \sum_t \sigma_{Y,t}^2/n_{L,t} = \Var(\hat{\tau}^{(0)})$. Since $0 \in [0,1] \subset \R$, the minimum over $[0,1]$ is at most $V(0)$ and the minimum over $\R$ is at most the minimum over $[0,1]$, which gives Equation~\eqref{eq:efficiency-ordering}. If some arm has $v_t > 0$ and $n_{L,t} < n_t$, the leading coefficient is positive and $V$ is strictly convex with unique minimizer
$$\lambda^* = \frac{\sum_t a_t\,\gamma_t}{\sum_t a_t\,v_t},$$
which is Equation~\eqref{eq:lambda-exact} at population moments; the clipped minimizer is $\lambda^*_c = \clip(\lambda^*, 0, 1)$. Strict convexity makes the first inequality strict when $\lambda^* \notin [0,1]$, and the second strict when $\lambda^*_c > 0$, that is, when $\lambda^* > 0$. At $\pi_L = 1$ every $a_t$ is zero, $V$ is constant, and the three variances coincide. When the arms share $(\gamma_t, v_t)$, $\lambda^* = \gamma/v$ does not depend on $\pi_L$, whereas the plug-in rule of Equation~\eqref{eq:lambda-opt} does; the ordering concerns only the minimizers of $V$ over the two feasible sets.
\end{proof}

\section{Proofs for the Variance Correction, the Diagnostic, and the Fixed-\texorpdfstring{$c$}{c} Hybrid}
\label{app:hybrid-proofs}

\subsection{Proof of Proposition~\ref{prop:corrected-variance} (Overlap-Corrected PPI++ Variance)}

\begin{proof}
The PPI++ estimator decomposes within each treatment arm $t$ as $\hat{\tau}_t = \bar{Y}_{L,t} + \lambda(\bar{\hat{Y}}_t - \bar{\hat{Y}}_{L,t})$.
The all-data mean decomposes as $\bar{\hat{Y}}_t = (n_{L,t}/n_t)\bar{\hat{Y}}_{L,t} + (n_{U,t}/n_t)\bar{\hat{Y}}_{U,t}$, so
$\bar{\hat{Y}}_t - \bar{\hat{Y}}_{L,t} = (n_{U,t}/n_t)(\bar{\hat{Y}}_{U,t} - \bar{\hat{Y}}_{L,t})$ and
$\hat{\tau}_t = \bar{Y}_{L,t} + \lambda (n_{U,t}/n_t)(\bar{\hat{Y}}_{U,t} - \bar{\hat{Y}}_{L,t})$.

Condition on $n_t$ and $n_{L,t}$ and hold the prediction function and $\lambda$ fixed. Since labeled and unlabeled units are disjoint and i.i.d.\ within the arm, $\bar{Y}_{L,t}$ and $\bar{\hat{Y}}_{U,t}$ are independent, while $\bar{Y}_{L,t}$ and $\bar{\hat{Y}}_{L,t}$ are correlated. The exact variance is
$$\Var(\hat{\tau}_t) = \frac{\sigma_{Y,t}^2}{n_{L,t}} + \lambda^2 \frac{n_{U,t}^2}{n_t^2}\sigma_{\hat{Y},t}^2\!\left(\frac{1}{n_{U,t}} + \frac{1}{n_{L,t}}\right) - \frac{2\lambda n_{U,t}}{n_t}\frac{\gamma_t}{n_{L,t}}.$$

The plug-in variance formula treats the all-data and labeled-data terms as independent, yielding:
$$V_{\mathrm{plug},t} = \frac{\sigma_{Y,t}^2}{n_{L,t}} + \lambda^2 \sigma_{\hat{Y},t}^2\!\left(\frac{1}{n_t} + \frac{1}{n_{L,t}}\right) - \frac{2\lambda \gamma_t}{n_{L,t}}.$$

The two covariances that this independence assumption discards are exactly the ones the $n_{L,t}$ shared units create, $\Cov(\bar{Y}_{L,t}, \bar{\hat{Y}}_t) = \gamma_t/n_t$ and $\Cov(\bar{\hat{Y}}_t, \bar{\hat{Y}}_{L,t}) = \sigma_{\hat{Y},t}^2/n_t$, entering the variance with weights $2\lambda$ and $-2\lambda^2$ respectively. The difference $C_t = \Var(\hat{\tau}_t) - V_{\mathrm{plug},t}$ therefore simplifies to $C_t = 2\lambda(\gamma_t - \lambda\sigma_{\hat{Y},t}^2)/n_t$
after algebraic cancellation using $n_{U,t} = n_t - n_{L,t}$, which is Equation~\eqref{eq:overlap-intuition}; substituting the plug-in rule $\lambda = \lambda_t^*/(1 + \pi_L)$ with $\lambda_t^* = \gamma_t/\sigma_{\hat{Y},t}^2$ gives $\gamma_t - \lambda\sigma_{\hat{Y},t}^2 = \gamma_t\pi_L/(1+\pi_L)$ and hence Equation~\eqref{eq:overlap-piL}. Summing over arms, and using $(n_{U,t}/n_t)^2(1/n_{U,t} + 1/n_{L,t}) = (n_{U,t}/n_t)/n_{L,t} = a_t$, gives Equation~\eqref{eq:exact-variance}.

For the limit, both $V(\lambda)$ and $\hat{V}_{\mathrm{corrected}}$ tend to zero, so we work on the $n$ scale. Write $nV(\lambda) = \sum_t [\sigma^2_{Y,t}\,n/n_{L,t} + n a_t(\lambda^2 v_t - 2\lambda\gamma_t)]$, and $n\hat{V}_{\mathrm{corrected}}$ as the same expression with $(\hat{\sigma}^2_{Y,t}, \hat{v}_t, \hat{\gamma}_t)$ in place of the population moments. The coefficients $n/n_{L,t}$ and $n a_t$ converge to finite limits, and the sample moments are consistent, so $n\hat{V}_{\mathrm{corrected}} - nV(\lambda) \xrightarrow{p} 0$ by the continuous mapping theorem; with an estimate $\hat\lambda \xrightarrow{p} \lambda$ the same holds because the expression is continuous in $\lambda$. Since $nV(\lambda)$ converges to a positive limit, the ratio $\hat{V}_{\mathrm{corrected}}/V(\lambda)$ tends to one in probability. Combined with asymptotic normality of $\sqrt{n}(\hat{\tau}^{\mathrm{PPI}(\lambda)} - \tau)$, Slutsky's theorem gives asymptotically exact Wald coverage.
\end{proof}

\subsection{Proof of Proposition~\ref{prop:surrogacy-test} (Diagnostic Test)}

\begin{proof}
Under $H_0$, $\tau_{\mathrm{SI}}^{0} = d_0'\beta_0 = \tau$, so both coordinates of Equation~\eqref{eq:joint-clt} are centered at $\tau$. Since $\hat{D} = \hat{\tau}_{\mathrm{PPI++}} - \hat{\tau}_{\mathrm{SI}}$ is the contrast $(-1, 1)$ of the pair, Proposition~\ref{prop:joint-si-ppi} gives $\sqrt{n}\,\hat{D} \xrightarrow{d} \mathcal{N}(0, \sigma_D^2)$ with $\sigma_D^2 = (1,-1)\Omega(1,-1)' = \Omega_{11} + \Omega_{22} - 2\Omega_{12}$. Consistency of $\hat\Omega$ gives $n\widehat{\Var}(\hat D) = (1,-1)\hat\Omega(1,-1)' \xrightarrow{p} \sigma_D^2$, and because $\sigma_D^2 > 0$ the square root is continuous there, so $\sqrt{n}\,\mathrm{SE}(\hat D) \xrightarrow{p} \sigma_D$ and Slutsky's theorem gives $T_n = \sqrt{n}\hat{D}/\{\sqrt{n}\,\mathrm{SE}(\hat{D})\} \xrightarrow{d} \mathcal{N}(0,1)$.
\end{proof}

\paragraph{The covariance under a fixed prediction function.} For comparison, and because it is the form the surrogate-index literature uses, we record the closed form of $\Cov(\hat{\tau}_{\mathrm{SI}}, \hat{\tau}_{\mathrm{PPI++}})$ when $f$ is fixed rather than learned (Remark~\ref{rem:fixed-f-cov}). There $\hat{\tau}_{\mathrm{SI}} = \hat{\tau}_{\hat{Y}}^{\mathrm{all}}$ and $\hat{\tau}_{\mathrm{PPI++}} = \hat{\tau}_Y^{\mathcal{L}} + \lambda^*(\hat{\tau}_{\hat{Y}}^{\mathrm{all}} - \hat{\tau}_{\hat{Y}}^{\mathcal{L}})$, so by linearity $\Cov(\hat{\tau}_{\mathrm{SI}}, \hat{\tau}_{\mathrm{PPI++}}) = \Cov(\hat{\tau}_{\hat{Y}}^{\mathrm{all}}, \hat{\tau}_Y^{\mathcal{L}}) + \lambda^* \Var(\hat{\tau}_{\hat{Y}}^{\mathrm{all}}) - \lambda^* \Cov(\hat{\tau}_{\hat{Y}}^{\mathrm{all}}, \hat{\tau}_{\hat{Y}}^{\mathcal{L}})$. Within each arm $t$, $\Cov(\bar{\hat{Y}}_t, \bar{Y}_{L,t}) = \gamma_t/n_t$ and $\Cov(\bar{\hat{Y}}_t, \bar{\hat{Y}}_{L,t}) = \sigma_{\hat{Y},t}^2/n_t$ from the $n_{L,t}$ shared units, the $\lambda^*$-dependent terms cancel, and $\Cov(\hat{\tau}_{\mathrm{SI}}, \hat{\tau}_{\mathrm{PPI++}}) = \sum_t \gamma_t/n_t$. This identity does not extend to the learned index: it omits the first-stage term of Equation~\eqref{eq:phi-si}, which enters both $\Var(\hat{\tau}_{\mathrm{SI}})$ and the covariance.

\subsection{Hybrid Risk by Regime}

We claim neither a uniform-oracle MSE bound for the hybrid nor consistency of its plug-in interval (Remark~\ref{rem:ci-validity}). Proposition~\ref{prop:limit-risk} is the operative result; the two paragraphs below record the endpoints it interpolates between.

\emph{Fixed detectable alternatives.} If $|\delta|/\mathrm{SE}(\hat{D}) \to \infty$ then $T_n^2 \to \infty$ in probability and $w \xrightarrow{p} 0$, so the hybrid converges to PPI++ and the interval collapses onto the exact-variance PPI++ interval. The interval inherits PPI validity in this regime, under Assumption~\ref{asm:mcar} and the regularity conditions of Proposition~\ref{prop:joint-si-ppi}.

\emph{Local alternatives.} Along $\delta_n$ with $\sqrt{n}\,\E[\hat{D}] \to h$, the statistic $T_n$ converges in distribution to $\mathcal{N}(h/\sigma_D, 1)$, so $w$ converges in distribution to a non-degenerate random variable in $(0, 1)$ and the hybrid is a random convex combination of SI and PPI++. The expected MSE involves the distribution of $w$ times the SI bias, which is non-zero and does not vanish; Proposition~\ref{prop:limit-risk} characterizes this regime exactly. This is the regime in which the Hillstrom hybrid carries 10.7\% relative bias,
and it is also the regime in which the plug-in interval has no coverage guarantee: the standardized gap $h/\sigma_D$ is estimated with error that stays $O_p(1)$ at exactly the scale on which the weight varies, so the simulated percentile distribution does not converge to the hybrid's sampling distribution \citep{leeb2005model}. We therefore present the fixed-$c$ hybrid as an exploratory shrinkage estimator with an exact limit-experiment risk characterization and an empirically reported interval, not as a theorem-backed inferential procedure. The fixed $c = 1.5$ is justified by the limit-experiment and finite-sample minimax analyses of Appendix~\ref{app:hybrid}, not by an MSE oracle bound.

\subsection{Proof of Proposition~\ref{prop:limit-risk} (Limit-Experiment Risk)}

\begin{proof}
Under the stated local sequence and the asymptotic linearity maintained in Proposition~\ref{prop:surrogacy-test}, the pair $\bigl(\sqrt{n}(\hat{\tau}_{\mathrm{SI}} - \tau),\, \sqrt{n}(\hat{\tau}_{\mathrm{PPI++}} - \tau)\bigr)$ converges jointly to a bivariate normal vector $(G_S, G_P)$ with means $(-h, 0)$, variances $(\sigma_S^2, \sigma_P^2)$, and covariance $\sigma_{SP}$: the local shift moves the SI mean by $-h$, while PPI++ stays centered at $\tau$ at first order under MCAR (Proposition~\ref{prop:joint-si-ppi}). Here $h$ is by definition the limit of $\sqrt{n}\,\E[\hat{D}]$, so no separate accounting for the attenuation factor of Proposition~\ref{prop:si-bias} is needed: under DGP~2 with $\delta_n = d/\sqrt{n}$, Proposition~\ref{prop:si-bias} gives $h = \kappa d$, and every quantity in the risk formula is expressed in terms of $h$ rather than $d$. Define $V = G_P - G_S \sim \mathcal{N}(h, \sigma_D^2)$; then $\sqrt{n}\hat{D} \xrightarrow{d} V$ jointly, and since $\sqrt{n}\,\mathrm{SE}(\hat{D}) \xrightarrow{p} \sigma_D$, Slutsky gives $T_n \xrightarrow{d} V/\sigma_D$ jointly with $(G_S, G_P)$. Writing
\[
  \sqrt{n}(\hat{\tau}_{\mathrm{hybrid}} - \tau)
  = \sqrt{n}(\hat{\tau}_{\mathrm{PPI++}} - \tau) - w(T_n)\,\sqrt{n}\hat{D},
\]
the map $(g, v) \mapsto g - w(v/\sigma_D)v$ is continuous and $w$ is bounded, so the continuous mapping theorem yields \eqref{eq:limit-H}.

For the risk, $\Cov(G_P, V) = \Cov(G_P, G_P - G_S) = \sigma_P^2 - \sigma_{SP} = \varrho\,\sigma_D^2$. Decompose $G_P = \varrho(V - h) + U$ with $U \perp V$ and $\Var(U) = \sigma_P^2 - \varrho^2\sigma_D^2$. Then $H = U + \varrho(V - h) - w(V/\sigma_D)V$, and since $U$ is independent of $V$,
\[
  \E[H^2] = \Var(U) + \E_V\!\left[\bigl(\varrho(V - h) - w(V/\sigma_D)V\bigr)^2\right],
\]
which is \eqref{eq:limit-risk}. The function $v \mapsto w(v/\sigma_D)v = c\,v/(c + v^2/\sigma_D^2)$ is bounded (by $\sigma_D\sqrt{c}/2$ in absolute value), so dominated convergence gives continuity in $(h, c)$ and, since $w(V/\sigma_D) \xrightarrow{p} 0$ as $|h| \to \infty$ while $\E[\varrho^2(V-h)^2] = \varrho^2\sigma_D^2$, the limit $r(h; c) \to \sigma_P^2 - \varrho^2\sigma_D^2 + \varrho^2\sigma_D^2 = \sigma_P^2$.
\end{proof}

Numerically, at the calibrated covariance ($\sigma_S = 2.53$, $\sigma_P = 5.21$, $\mathrm{corr}_{SP} = 0.52$; estimated from $R = 600$ replications of the featured cell), the finite-sample check at $n = 10{,}000$ gives $n \times \mathrm{MSE}$ of $13.3$, $29.4$, and $30.0$ at $\theta \in \{0, 2, 4\}$ against limit-formula values of $14.3$, $32.1$, and $32.8$: the finite-sample values fall $7.2\%$, $8.3\%$, and $8.4\%$ below the limit values.

\section{Extended Power Curves for the SI--PPI++ Estimator-Disagreement Diagnostic}
\label{app:power-curves}

Table~\ref{tab:power-curves} reports the diagnostic's power across $n \in \{1{,}000,\; 10{,}000,\; 100{,}000\}$ and $\pi_L \in \{0.05, 0.20, 0.50\}$, one- and two-sided at $\alpha = 0.05$. At $\rho = 0$ the size is near nominal at every labeled fraction and sample size, 3.4--7.6\% one-sided and 3.6--9.0\% two-sided; with $\mathrm{SE}(\hat{D})$ carrying the first-stage term, the $\pi_L = 0.50$ sizes run 4.0--5.6\% one-sided. Power grows with $n$, $\pi_L$, and $\rho$, and the one-sided test dominates the two-sided one throughout. Table~\ref{tab:frontier-coverage} reports SI and PPI++ coverage from an independent run on a $\rho$ grid in steps of $0.05$; its damage-edge readings enter the grid-resolution bracket of Section~\ref{sec:detection-damage}, while the band edges themselves come from the paired run behind Figure~\ref{fig:detection-damage}.

\begin{table}[!ht]
\centering
\begin{minipage}[t]{0.52\textwidth}
\centering
\caption{Power of the SI--PPI++ estimator-disagreement diagnostic.}
\label{tab:power-curves}
\scriptsize
\setlength{\tabcolsep}{3pt}
\begin{tabular}{@{}rrlrrr@{}}
\toprule
$\boldsymbol{n}$ & $\boldsymbol{\pi_L}$ & $\boldsymbol{\rho}$ & $\boldsymbol{\delta}$ & \textbf{1-sided} & \textbf{2-sided} \\
\midrule
1,000 & 0.05 & 0.0 & 0.000 & 0.076 & 0.090 \\
      &      & 0.4 & 0.100 & 0.110 & 0.088 \\
      &      & 0.8 & 0.600 & 0.612 & 0.520 \\
\midrule
1,000 & 0.20 & 0.0 & 0.000 & 0.034 & 0.040 \\
      &      & 0.4 & 0.100 & 0.126 & 0.084 \\
      &      & 0.8 & 0.600 & 0.994 & 0.982 \\
\midrule
1,000 & 0.50 & 0.0 & 0.000 & 0.056 & 0.052 \\
      &      & 0.4 & 0.100 & 0.320 & 0.194 \\
      &      & 0.8 & 0.600 & 1.000 & 1.000 \\
\midrule
10,000 & 0.05 & 0.0 & 0.000 & 0.042 & 0.036 \\
       &      & 0.4 & 0.100 & 0.296 & 0.202 \\
       &      & 0.8 & 0.600 & 1.000 & 1.000 \\
\midrule
10,000 & 0.20 & 0.0 & 0.000 & 0.034 & 0.040 \\
       &      & 0.1 & 0.017 & 0.084 & 0.062 \\
       &      & 0.2 & 0.038 & 0.220 & 0.132 \\
       &      & 0.3 & 0.064 & 0.382 & 0.294 \\
       &      & 0.4 & 0.100 & 0.738 & 0.630 \\
       &      & 0.5 & 0.150 & 0.946 & 0.924 \\
       &      & 0.8 & 0.600 & 1.000 & 1.000 \\
\midrule
10,000 & 0.50 & 0.0 & 0.000 & 0.052 & 0.048 \\
       &      & 0.4 & 0.100 & 0.982 & 0.946 \\
       &      & 0.8 & 0.600 & 1.000 & 1.000 \\
\midrule
100,000 & 0.05 & 0.0 & 0.000 & 0.056 & 0.054 \\
        &      & 0.4 & 0.100 & 0.980 & 0.952 \\
        &      & 0.8 & 0.600 & 1.000 & 1.000 \\
\midrule
100,000 & 0.20 & 0.0 & 0.000 & 0.050 & 0.048 \\
        &      & 0.1 & 0.017 & 0.322 & 0.224 \\
        &      & 0.2 & 0.038 & 0.866 & 0.786 \\
        &      & 0.3 & 0.064 & 1.000 & 1.000 \\
        &      & 0.4 & 0.100 & 1.000 & 1.000 \\
\midrule
100,000 & 0.50 & 0.0 & 0.000 & 0.040 & 0.046 \\
        &      & 0.2 & 0.038 & 1.000 & 1.000 \\
        &      & 0.4 & 0.100 & 1.000 & 1.000 \\
\bottomrule
\end{tabular}
\smallskip

{\scriptsize \noindent Notes: DGP~2, $R = 500$ per cell, all-units five-fold cross-fitting, $\alpha = 0.05$, $\mathrm{SE}(\hat{D})$ from the joint sandwich. The grid is independent of the paired run of Section~\ref{sec:detection-damage}, which puts the $\rho = 0.2$, $n = 10{,}000$ cell at $0.136$ two-sided against $0.132$ here.}
\end{minipage}\hfill
\begin{minipage}[t]{0.44\textwidth}
\centering
\caption{Coverage on the $\rho$ grid behind the detection-damage band.}
\label{tab:frontier-coverage}
\scriptsize
\setlength{\tabcolsep}{4pt}
\begin{tabular}{@{}rrrrr@{}}
\toprule
& \multicolumn{2}{c}{$\boldsymbol{n = 10{,}000}$} & \multicolumn{2}{c}{$\boldsymbol{n = 100{,}000}$} \\
\cmidrule(lr){2-3}\cmidrule(lr){4-5}
$\boldsymbol{\rho}$ & \textbf{SI} & \textbf{PPI++} & \textbf{SI} & \textbf{PPI++} \\
\midrule
0.00 & $0.950$ & $0.946$ & $0.952$ & $0.960$ \\
0.05 & $\mathbf{0.924}$ & $0.954$ & $\mathbf{0.804}$ & $0.952$ \\
0.10 & $\mathbf{0.884}$ & $0.942$ & $\mathbf{0.454}$ & $0.958$ \\
0.15 & $\mathbf{0.806}$ & $\mathbf{0.928}$ & $\mathbf{0.088}$ & $0.952$ \\
0.20 & $\mathbf{0.662}$ & $0.956$ & $\mathbf{0.008}$ & $0.932$ \\
0.25 & $\mathbf{0.514}$ & $0.950$ & $\mathbf{0.000}$ & $0.960$ \\
0.30 & $\mathbf{0.308}$ & $0.944$ & $\mathbf{0.000}$ & $0.946$ \\
0.40 & $\mathbf{0.036}$ & $0.952$ & $\mathbf{0.000}$ & $0.946$ \\
0.50 & $\mathbf{0.000}$ & $0.936$ & $\mathbf{0.000}$ & $0.944$ \\
0.60 & $\mathbf{0.000}$ & $\mathbf{0.970}$ & $\mathbf{0.000}$ & $\mathbf{0.970}$ \\
\bottomrule
\end{tabular}
\smallskip

{\scriptsize \noindent Notes: DGP~2, $\pi_L = 0.20$, $R = 500$ per cell, all-units five-fold cross-fitting. \textbf{Bold} coverage entries fall outside the band $[0.931, 0.969]$, in either direction. PPI++ is the primary configuration and the surrogate index uses the joint sandwich variance. The band edges of Section~\ref{sec:detection-damage} come instead from the paired run and its refined local grid (Table~\ref{tab:edge-ci}).}
\end{minipage}
\end{table}

\section{Cross-PPI Comparison}
\label{app:cross-ppi}

Cross-PPI \citep{zrnic2024cross} applies cross-fitting to both the prediction model and the PPI correction. Across DGPs 1 and 2 (Table~\ref{tab:cross-ppi}), Cross-PPI and PPI++ under our all-units $K = 5$ protocol produce nearly identical results. RMSE values differ by at most 1.2\%, and the near-equivalence is expected: with every unit already scored out of fold, the additional cross-fitting in Cross-PPI provides only marginal improvement.

\begin{table}[!htbp]
\centering
\caption{Cross-PPI versus PPI++.}
\label{tab:cross-ppi}
\footnotesize
\setlength{\tabcolsep}{3pt}
\begin{tabular}{@{}llrrrrr@{}}
\toprule
\textbf{DGP} & \textbf{Method} & \textbf{Bias} & \textbf{RMSE} & \textbf{Coverage} & \textbf{RE vs LO} & \textbf{ESS mult.} \\
\midrule
DGP 1 & LO & $-0.0002$ & $0.066$ & $0.960$ & $1.000$ & $1.00$ \\
DGP 1 & PPI++ & $-0.003$ & $0.050$ & $\mathbf{0.970}$ & $1.326$ & $1.76$ \\
DGP 1 & Cross-PPI & $-0.003$ & $0.050$ & $0.948$ & $1.330$ & $1.77$ \\
\midrule
DGP 2 ($\rho{=}0$) & LO & $+0.004$ & $0.072$ & $0.948$ & $1.000$ & $1.00$ \\
DGP 2 ($\rho{=}0$) & PPI++ & $+0.002$ & $0.050$ & $0.952$ & $1.443$ & $2.08$ \\
DGP 2 ($\rho{=}0$) & Cross-PPI & $+0.002$ & $0.051$ & $\mathbf{0.926}$ & $1.426$ & $2.03$ \\
\midrule
DGP 2 ($\rho{=}0.2$) & LO & $+0.002$ & $0.069$ & $0.950$ & $1.000$ & $1.00$ \\
DGP 2 ($\rho{=}0.2$) & PPI++ & $+0.002$ & $0.046$ & $\mathbf{0.972}$ & $1.481$ & $2.19$ \\
DGP 2 ($\rho{=}0.2$) & Cross-PPI & $+0.002$ & $0.047$ & $0.960$ & $1.468$ & $2.16$ \\
\midrule
DGP 2 ($\rho{=}0.4$) & LO & $-0.003$ & $0.068$ & $0.950$ & $1.000$ & $1.00$ \\
DGP 2 ($\rho{=}0.4$) & PPI++ & $-0.004$ & $0.049$ & $0.960$ & $1.372$ & $1.88$ \\
DGP 2 ($\rho{=}0.4$) & Cross-PPI & $-0.004$ & $0.050$ & $0.942$ & $1.367$ & $1.87$ \\
\bottomrule
\end{tabular}
\smallskip

{\scriptsize \noindent Notes: $n = 10{,}000$, $\pi_L = 0.20$, $R = 500$, all-units five-fold cross-fitting. The Monte Carlo band around the nominal 95\% is $[0.931, 0.969]$ and \textbf{bold} coverage entries fall outside it, in either direction. RE is the RMSE ratio $\mathrm{RMSE}_{\mathrm{LO}}/\mathrm{RMSE}_{\mathrm{method}}$ and the ESS-multiplier column is its square. PPI++ is the primary configuration; Cross-PPI uses its own cross-fitted correction, so the two rows differ in both the fitting scheme and the variance construction.}
\end{table}

\section{Additional Supplementary Results}
\label{app:additional}

\subsection{Timing, Robustness Checks, and Full Comparison Tables}

\paragraph{Timing benchmarks.}

Table~\ref{tab:timing} reports the median wall-clock cost of each component at two sample sizes.

\begin{table}[!ht]
\centering
\caption{Computational timing for key components.}
\label{tab:timing}
\footnotesize
\begin{tabular}{@{}lrr@{}}
\toprule
\textbf{Component} & \textbf{$n = 10{,}000$} & \textbf{$n = 100{,}000$} \\
\midrule
SI estimation & 1.5ms & 15.0ms \\
PPI++ estimation & 749$\mu$s & 6.3ms \\
Hybrid CI ($M = 10{,}000$ draws) & 566$\mu$s & 835$\mu$s \\
PPI++ Bootstrap ($B = 200$) & 77ms & 604ms \\
OLS prediction & 650$\mu$s & 5.7ms \\
GBT prediction & 698ms & 9.2s \\
\bottomrule
\end{tabular}
\smallskip

{\scriptsize \noindent Notes: median wall-clock time per operation over $R = 20$ replications on DGP~1 at $\pi_L = 0.20$, on the reference machine (8 cores, Apple silicon). The prediction rows are the full five-fold cross-fitted fit-and-predict step. SI estimation includes the joint sandwich of Proposition~\ref{prop:joint-si-ppi}, which is why it costs about twice PPI++ rather than a third of it. The hybrid CI cost is nearly independent of $n$ because it operates on pre-computed summary statistics (the $2 \times 2$ covariance matrix) via Monte Carlo sampling.}
\end{table}

\paragraph{Multi-surrogate DGP~7 results.}

Table~\ref{tab:dgp7} reports every method on the three-surrogate design across labeled fractions.

\begin{table}[!ht]
\centering
\caption{DGP 7 (Multi-Surrogate, $J = 3$): method performance across labeled fractions.}
\label{tab:dgp7}
\footnotesize
\begin{tabular}{@{}llrrrrr@{}}
\toprule
$\boldsymbol{\pi_L}$ & \textbf{Method} & \textbf{Bias} & \textbf{RMSE} & \textbf{Coverage} & \textbf{RE} & \textbf{ESS mult.} \\
\midrule
\multirow{4}{*}{0.05} & Labeled-Only & $-0.001$ & 0.136 & 0.950 & 1.00 & 1.00 \\
& Surrogate Index & $+0.001$ & 0.025 & 0.952 & 5.36 & 28.7 \\
& PPI++ & $-0.001$ & 0.097 & 0.944 & 1.40 & 1.96 \\
& Composite Proxy & $+0.003$ & 0.022 & 0.936 & 6.30 & 39.7 \\
\midrule
\multirow{4}{*}{0.20} & Labeled-Only & $-0.003$ & 0.069 & 0.940 & 1.00 & 1.00 \\
& Surrogate Index & $+0.001$ & 0.023 & 0.956 & 3.04 & 9.22 \\
& PPI++ & $+0.001$ & 0.052 & 0.940 & 1.34 & 1.79 \\
& Composite Proxy & $+0.004$ & 0.021 & 0.954 & 3.28 & 10.7 \\
\midrule
\multirow{4}{*}{0.50} & Labeled-Only & $-0.0002$ & 0.041 & 0.968 & 1.00 & 1.00 \\
& Surrogate Index & $+0.0001$ & 0.023 & 0.960 & 1.82 & 3.31 \\
& PPI++ & $+0.0003$ & 0.033 & 0.964 & 1.23 & 1.50 \\
& Composite Proxy & $+0.003$ & 0.021 & 0.958 & 1.97 & 3.89 \\
\bottomrule
\end{tabular}
\smallskip

{\scriptsize \noindent Notes: $n = 10{,}000$, $K_{\text{hist}} = 50$, $R = 500$, all-units five-fold cross-fitting. The Monte Carlo band around the nominal 95\% is $[0.931, 0.969]$ and every entry is inside it. PPI++ is the primary configuration and the surrogate index uses the joint sandwich variance. RE is the RMSE ratio $\mathrm{RMSE}_{\mathrm{LO}}/\mathrm{RMSE}_{\mathrm{method}}$ and the ESS-multiplier column is its square. RE values should be read alongside coverage.}
\end{table}

\paragraph{GBT against OLS.}

Table~\ref{tab:gbt} sets the gradient-boosted-tree prediction model against OLS on the representative cells.

\begin{table}[!ht]
\centering
\caption{GBT versus OLS prediction models across representative DGPs.}
\label{tab:gbt}
\footnotesize
\begin{tabular}{@{}llllrrrr@{}}
\toprule
\textbf{DGP} & $\boldsymbol{\pi_L}$ & \textbf{Model} & \textbf{Method} & \textbf{Bias} & \textbf{RMSE} & \textbf{Coverage} & \textbf{RE} \\
\midrule
1 & 0.05 & OLS & Surrogate Index & $-0.001$ & 0.027 & 0.950 & 5.59 \\
1 & 0.05 & GBT & Surrogate Index & $-0.003$ & 0.027 & 0.950 & 5.59 \\
1 & 0.05 & OLS & PPI++ & $+0.007$ & 0.093 & 0.955 & 1.60 \\
1 & 0.05 & GBT & PPI++ & $+0.004$ & 0.098 & 0.950 & 1.51 \\
\midrule
2 ($\rho{=}0.4$) & 0.20 & OLS & Surrogate Index & $-0.099$ & 0.102 & \textbf{0.030} & --- \\
2 ($\rho{=}0.4$) & 0.20 & GBT & Surrogate Index & $-0.099$ & 0.103 & \textbf{0.010} & --- \\
2 ($\rho{=}0.4$) & 0.20 & OLS & PPI++ & $-0.003$ & 0.048 & 0.980 & 1.44 \\
2 ($\rho{=}0.4$) & 0.20 & GBT & PPI++ & $-0.002$ & 0.050 & 0.980 & 1.36 \\
\midrule
4 ($\Delta{=}0.2$) & 0.20 & OLS & PPI++ & $-0.001$ & 0.055 & 0.935 & 1.60 \\
4 ($\Delta{=}0.2$) & 0.20 & GBT & PPI++ & $-0.0004$ & 0.056 & 0.955 & 1.57 \\
\midrule
8 & 0.05 & OLS & Surrogate Index & $+0.0004$ & 0.025 & 0.950 & 5.47 \\
8 & 0.05 & GBT & Surrogate Index & $-0.002$ & 0.026 & 0.945 & 5.18 \\
\bottomrule
\end{tabular}
\smallskip

{\scriptsize \noindent Notes: $n = 10{,}000$, $R = 200$ replications per cell, all-units five-fold cross-fitting. The Monte Carlo band around the nominal 95\% is $[0.920, 0.980]$ and \textbf{bold} coverage entries fall outside it. PPI++ is the primary configuration. The surrogate index uses the joint sandwich variance in the OLS rows and the fixed-predictor plug-in variance in the GBT rows, which ignores first-stage uncertainty, so GBT surrogate-index coverage is heuristic; on two cells the plug-in variance is within Monte Carlo error of the Monte Carlo variance (Appendix~\ref{app:ablations}), and the comparison's evidential weight is on bias, RMSE, and RE. RE is the RMSE ratio $\mathrm{RMSE}_{\mathrm{LO}}/\mathrm{RMSE}_{\mathrm{method}}$ within the same cell and prediction model, and is suppressed (---) for methods with absolute relative bias above 10\%.}
\end{table}

\paragraph{Criteo-calibrated semi-synthetic testbed.}

Table~\ref{tab:semisynthetic-full} gives the full semi-synthetic sweep behind the main-text summary.

\begin{table}[!ht]
\centering
\caption{Full per-method results for the Criteo-calibrated semi-synthetic sweep ($n = 64{,}000$, $\pi_L = 0.20$, $R = 500$ per dial). Bias in units of $10^{-4}$.}
\label{tab:semisynthetic-full}
\footnotesize
\setlength{\tabcolsep}{4pt}
\begin{tabular}{@{}rlrrrr@{}}
\toprule
\textbf{Dial} & \textbf{Method} & \textbf{Bias} & \textbf{Rel.\ bias (\%)} & \textbf{Coverage} & \textbf{RE} \\
\midrule
\multirow{3}{*}{0.0} & Labeled-Only & $+0.00$ & $+0.1$ & 0.958 & 1.00 \\
 & Surrogate Index & $+0.09$ & $+3.5$ & \textbf{0.976} & 5.65 \\
 & PPI++ & $+0.16$ & $+6.1$ & 0.942 & 1.06 \\
\midrule
\multirow{3}{*}{0.5} & Labeled-Only & $-0.11$ & $-1.9$ & 0.958 & 1.00 \\
 & Surrogate Index & $-2.94$ & $-51.6$ & \textbf{0.550} & --- \\
 & PPI++ & $+0.07$ & $+1.2$ & 0.944 & 1.07 \\
\midrule
\multirow{3}{*}{1.0} & Labeled-Only & $-0.57$ & $-6.3$ & 0.952 & 1.00 \\
 & Surrogate Index & $-6.14$ & $-67.5$ & \textbf{0.084} & --- \\
 & PPI++ & $-0.37$ & $-4.1$ & 0.948 & 1.07 \\
\midrule
\multirow{3}{*}{1.5} & Labeled-Only & $+0.36$ & $+2.8$ & 0.948 & 1.00 \\
 & Surrogate Index & $-9.62$ & $-74.6$ & \textbf{0.004} & --- \\
 & PPI++ & $+0.35$ & $+2.7$ & 0.942 & 1.06 \\
\midrule
\multirow{3}{*}{2.0} & Labeled-Only & $-0.43$ & $-2.5$ & 0.942 & 1.00 \\
 & Surrogate Index & $-13.56$ & $-79.2$ & \textbf{0.000} & --- \\
 & PPI++ & $-0.35$ & $-2.1$ & 0.940 & 1.07 \\
\bottomrule
\end{tabular}
\smallskip

{\scriptsize \noindent Notes: all-units five-fold cross-fitting. PPI++ is the primary configuration and the surrogate index uses the joint sandwich variance. \textbf{Bold} coverage entries fall outside the $R = 500$ Monte Carlo band $[0.931, 0.969]$, in either direction. RE is suppressed (---) for methods with absolute relative bias above 10\%. SI RE at dial 0 is $5.65$; PPI++ RE is $1.06$--$1.07$ at every dial. The dial-to-$\tau$ mapping and the diagnostic rejection rates are in Table~\ref{tab:semisynthetic} and are not repeated as columns here; the mean $T_n$ across dials is 0.02, 0.39, 0.71, 1.17, and 1.51. Calibration details: logistic visit and conversion-given-visit models fit on the real Criteo data (visit-stage treatment coefficient 0.23, conversion-stage $\hat{\kappa} = 0.31$, fit on 656{,}929 visitors); covariates resampled from a pool of 500{,}000 real rows.}
\end{table}

\paragraph{CUPED comparison.}

Table~\ref{tab:cuped-full} compares every method against the CUPED-adjusted baseline by labeled fraction.

\begin{table}[!ht]
\centering
\caption{CUPED-adjusted baseline comparison.}
\label{tab:cuped-full}
\footnotesize
\setlength{\tabcolsep}{4pt}
\begin{tabular}{@{}llrrrrrr@{}}
\toprule
& & \multicolumn{3}{c}{\textbf{Coverage}} & \multicolumn{3}{c}{\textbf{RE vs CUPED}} \\
\cmidrule(lr){3-5}\cmidrule(lr){6-8}
\textbf{DGP} & \textbf{Method} & $\pi_L = 0.05$ & $0.20$ & $0.50$ & $0.05$ & $0.20$ & $0.50$ \\
\midrule
DGP 1 & Surrogate Index & $0.936$ & $\mathbf{0.930}$ & $0.960$ & $4.77$ & $2.55$ & $1.59$ \\
 & PPI++ & $0.950$ & $0.958$ & $0.968$ & $1.35$ & $1.30$ & $1.10$ \\
 & CUPED & $0.954$ & $0.954$ & $0.956$ & $1.00$ & $1.00$ & $1.00$ \\
\midrule
DGP 2 ($\rho{=}0.0$) & Surrogate Index & $\mathbf{0.930}$ & $\mathbf{0.970}$ & $0.950$ & $4.76$ & $2.87$ & $1.70$ \\
 & PPI++ & $0.952$ & $0.936$ & $\mathbf{0.972}$ & $1.31$ & $1.23$ & $1.10$ \\
 & CUPED & $0.952$ & $\mathbf{0.924}$ & $0.944$ & $1.00$ & $1.00$ & $1.00$ \\
\midrule
DGP 2 ($\rho{=}0.2$) & Surrogate Index & $\mathbf{0.678}$ & $\mathbf{0.638}$ & $\mathbf{0.672}$ & --- & --- & --- \\
 & PPI++ & $0.960$ & $0.942$ & $0.950$ & $1.41$ & $1.23$ & $1.08$ \\
 & CUPED & $0.950$ & $0.952$ & $0.950$ & $1.00$ & $1.00$ & $1.00$ \\
\midrule
DGP 2 ($\rho{=}0.4$) & Surrogate Index & $\mathbf{0.032}$ & $\mathbf{0.020}$ & $\mathbf{0.008}$ & --- & --- & --- \\
 & PPI++ & $0.954$ & $0.958$ & $0.960$ & $1.37$ & $1.24$ & $1.06$ \\
 & CUPED & $0.962$ & $0.960$ & $0.950$ & $1.00$ & $1.00$ & $1.00$ \\
\midrule
DGP 8 & Surrogate Index & $0.934$ & $0.946$ & $0.946$ & $4.88$ & $2.53$ & $1.62$ \\
 & PPI++ & $0.946$ & $0.958$ & $0.950$ & $1.34$ & $1.25$ & $1.08$ \\
 & CUPED & $0.958$ & $0.946$ & $0.950$ & $1.00$ & $1.00$ & $1.00$ \\
\bottomrule
\end{tabular}
\smallskip

{\scriptsize \noindent Notes: $n = 10{,}000$, $R = 500$, all-units five-fold cross-fitting. \textbf{Bold} coverage entries fall outside the band $[0.931, 0.969]$, in either direction. PPI++ is the primary configuration and the surrogate index uses the joint sandwich variance. RE is $\mathrm{RMSE}_{\mathrm{CUPED}}/\mathrm{RMSE}_{\mathrm{method}}$, computed from the run's mean-squared-error ratios, and its square is the ESS multiplier; it is suppressed (---) where the surrogate index carries more than 10\% absolute relative bias. Bias, RMSE, and the omitted labeled-only rows, which run $0.88$--$0.94$ times the CUPED baseline in every cell, are in the replication package.}
\end{table}

\subsection{Interval, Portfolio, and Threshold Sensitivity}
\label{app:threshold-sensitivity}

\paragraph{Portfolio loss-function sensitivity.}

Table~\ref{tab:portfolio-sensitivity} varies the two features of the DGP~6 objective flagged in Section~\ref{sec:results-decisions}. A launched experiment pays $\tau_k - c$, so the oracle launches when $\tau_k > c$ and regret is $\sum_k [\max(\tau_k - c, 0) - d_k(\tau_k - c)]$, the paper's objective at $c = 0$. The antagonistic fraction $f_{\mathrm{ant}}$ is the share of experiments in which the surrogate responds positively to treatment while the outcome loads negatively on it, as in DGP~9. The decision rule stays cost-blind: reject $H_0: \tau_k = 0$ and launch if the estimate is positive.

\begin{table}[!htbp]
\centering
\caption{Portfolio regret under launch costs and sign-reversing experiments.}
\label{tab:portfolio-sensitivity}
\footnotesize
\setlength{\tabcolsep}{5pt}
\begin{tabular}{@{}l rrr rrr r@{}}
\toprule
& \multicolumn{3}{c}{$\boldsymbol{f_{\mathrm{ant}} = 0}$} & \multicolumn{3}{c}{$\boldsymbol{f_{\mathrm{ant}} = 0.10}$} & \textbf{Harmful launches} \\
\cmidrule(lr){2-4}\cmidrule(lr){5-7}
\textbf{Method} & $c = 0$ & $c = 0.01$ & $c = 0.02$ & $c = 0$ & $c = 0.01$ & $c = 0.02$ & $f_{\mathrm{ant}}: 0 \to 0.10$ \\
\midrule
Labeled-Only    & $82.6$ & $82.6$ & $82.9$ & $81.8$ & $81.9$ & $82.1$ & $0.53 \to 0.45$ \\
Naive Surrogate & $36.0$ & $33.9$ & $32.0$ & $51.1$ & $51.3$ & $52.2$ & $0.23 \to 3.16$ \\
Surrogate Index & $40.1$ & $38.1$ & $36.2$ & $54.9$ & $55.3$ & $56.4$ & $0.23 \to 3.25$ \\
PPI++           & $72.8$ & $72.2$ & $71.8$ & $72.4$ & $71.9$ & $71.5$ & $0.41 \to 0.35$ \\
\bottomrule
\end{tabular}
\smallskip

{\scriptsize \noindent Notes: $K = 200$ experiments per portfolio, $\pi_L = 0.20$, $R = 200$ portfolios, launch test at $\alpha = 0.05$ under the cost-blind rule, all-units five-fold cross-fitting. Entries in the first six columns are relative regret, mean regret as a percentage of the oracle gain, which is $4.010$ at $c = 0$ with $f_{\mathrm{ant}} = 0$ and falls to $3.528$ and $3.085$ as the launch cost rises, and $3.992$, $3.510$, $3.065$ at $f_{\mathrm{ant}} = 0.10$. The last column is the mean number of launched experiments with a negative true effect, per portfolio, at $c = 0$. PPI++ is the primary configuration; the plug-in ablation reproduces it to the reported precision in every cell. Both tables compute relative regret as the ratio of the mean regret to the mean oracle gain over their portfolios. At $c = 0$, $f_{\mathrm{ant}} = 0$ the entries are $82.6$, $36.0$, $40.1$, and $72.8$ here ($R = 200$) and $83.5$, $36.0$, $40.5$, and $73.6$ in Table~\ref{tab:portfolio} ($R = 500$); the Monte Carlo standard errors in this table are about $1$ point. This objective reports no coverage, so the bolding convention does not apply.}
\end{table}

\paragraph{Detection-damage band under alternative thresholds.}

Table~\ref{tab:detection-damage-sensitivity} locates the two edges under alternative thresholds, at both sample sizes and both levels, from the paired run behind Figure~\ref{fig:detection-damage}. The damage edge is the first $\rho$ at which SI coverage falls below the coverage threshold and the detection edge the first $\rho$ at which two-sided power exceeds the power threshold, both by linear interpolation on the grid $\rho \in \{0, 0.02, \ldots, 0.20, 0.25, 0.30, 0.35, 0.40, 0.50, 0.60\}$. The band is nonempty in every row with a coverage threshold of $0.85$ or $0.90$. At $0.95$ the damage edge is undefined rather than the band empty, because SI coverage at $\rho = 0$ is already $0.948$ under the joint sandwich. Raising the level from $0.05$ to $0.10$ pulls the detection edge in by $0.02$ to $0.05$ without closing the band in any row, and replacing the sandwich by the delta-method or fixed-predictor construction shifts either edge by at most $0.002$ at the featured thresholds and $0.003$ across the grid.

\begin{table}[!htbp]
\centering
\caption{Detection-damage band edges under alternative thresholds and test levels.}
\label{tab:detection-damage-sensitivity}
\footnotesize
\setlength{\tabcolsep}{5pt}
\begin{tabular}{@{}rrrrrrr@{}}
\toprule
$\boldsymbol{n}$ & \textbf{Cov.\ thr.} & \textbf{Power thr.} & $\boldsymbol{\alpha}$ & \textbf{Damage edge} & \textbf{Detection edge} & \textbf{Width} \\
\midrule
10{,}000 & 0.85 & 0.50 & 0.05 & $0.134$ & $0.372$ & $0.237$ \\
10{,}000 & 0.85 & 0.80 & 0.05 & $0.134$ & $0.460$ & $0.326$ \\
10{,}000 & 0.85 & 0.50 & 0.10 & $0.134$ & $0.323$ & $0.189$ \\
10{,}000 & 0.85 & 0.80 & 0.10 & $0.134$ & $0.432$ & $0.298$ \\
10{,}000 & 0.90 & 0.50 & 0.05 & $0.098$ & $0.372$ & $0.273$ \\
10{,}000 & 0.90 & 0.80 & 0.05 & $0.098$ & $0.460$ & $0.362$ \\
10{,}000 & 0.90 & 0.50 & 0.10 & $0.098$ & $0.323$ & $0.225$ \\
10{,}000 & 0.90 & 0.80 & 0.10 & $0.098$ & $0.432$ & $0.334$ \\
10{,}000 & 0.95 & 0.50 & 0.05 & undefined & $0.372$ & --- \\
\midrule
100{,}000 & 0.85 & 0.50 & 0.05 & $0.046$ & $0.160$ & $0.114$ \\
100{,}000 & 0.85 & 0.80 & 0.05 & $0.046$ & $0.217$ & $0.171$ \\
100{,}000 & 0.85 & 0.50 & 0.10 & $0.046$ & $0.135$ & $0.088$ \\
100{,}000 & 0.85 & 0.80 & 0.10 & $0.046$ & $0.190$ & $0.144$ \\
100{,}000 & 0.90 & 0.50 & 0.05 & $0.035$ & $0.160$ & $0.125$ \\
100{,}000 & 0.90 & 0.80 & 0.05 & $0.035$ & $0.217$ & $0.182$ \\
100{,}000 & 0.90 & 0.50 & 0.10 & $0.035$ & $0.135$ & $0.100$ \\
100{,}000 & 0.90 & 0.80 & 0.10 & $0.035$ & $0.190$ & $0.155$ \\
100{,}000 & 0.95 & 0.50 & 0.05 & undefined & $0.160$ & --- \\
\bottomrule
\end{tabular}
\smallskip

{\scriptsize \noindent Notes: DGP~2, $\pi_L = 0.20$, from the paired run behind Figure~\ref{fig:detection-damage}, $R = 500$ per grid point at $n = 10{,}000$ and $R = 300$ at $n = 100{,}000$; edges are on the $\rho$ scale and are located by linear interpolation. Rows use the joint sandwich of Proposition~\ref{prop:joint-si-ppi} for both the SI coverage curve and $\mathrm{SE}(\hat{D})$, so each row is internally consistent. ``Undefined'' marks a coverage threshold the SI curve never crosses from above on the grid because SI coverage at $\rho = 0$ is already below it; the $0.95$ rows at the $0.80$ power threshold and at $\alpha = 0.10$ are undefined for the same reason and are omitted. Damage edges and widths here are on the coarse paired grid; at the featured thresholds the refined local grid moves the damage edges to $0.094$ and $0.030$ and the widths from $0.273$ and $0.125$ to $0.277$ and $0.130$ (Table~\ref{tab:edge-ci}).}
\end{table}

\paragraph{Monte Carlo intervals for the band edges.}

Table~\ref{tab:edge-ci} attaches 95\% percentile intervals to the edges at the featured thresholds, coverage $0.90$ and power $0.50$. Each $(n, \rho)$ cell of the paired run is resampled over its replications $B = 1{,}000$ times, both curves are recomputed, and both edges are interpolated again; $\delta$-scale edges, widths, and shrink factors are formed draw by draw. The intervals are conditional on linear interpolation between the fixed grid points and omit grid-resolution error. The frontier-grid rows come from a run that stores only per-cell coverage proportions, so their intervals come from a parametric binomial bootstrap and are not term-by-term comparable with the others.

The refined run adds points around each crossing ($\rho$ in steps of $0.005$, $R = 1{,}000$ per point and $R = 2{,}000$ at the null); at a $\rho$ shared with the paired grid its cell contains the paired run's draws as its first replications, the fine curve is the refined points plus the paired points outside the refined range, and each cell is resampled to its own $R$. Grid resolution matters for the damage edge. On the coarse frontier grid of Table~\ref{tab:frontier-coverage}, SI coverage at $n = 100{,}000$ is $0.952$ at $\rho = 0$ and $0.804$ at $\rho = 0.05$, so linear interpolation places the $0.90$ crossing at $0.05\,(0.952 - 0.900)/(0.952 - 0.804) = 0.018$ whatever the curve does between those points, and a monotone cubic (pchip) interpolant at $0.024$. Part of that grid's disagreement with the paired run is Monte Carlo variation: at $\rho = 0.05$ the refined grid gives $0.849$ ($R = 1{,}000$), about two combined standard errors from $0.804$. At $n = 100{,}000$ the refined null cell gives SI coverage $0.942$ (Monte Carlo standard error $0.005$) and diagnostic size $0.046$ ($0.005$); SI coverage is $0.900$ at $\rho = 0.030$ and $0.889$ at $0.035$, and the edge is $0.030$ under both interpolants, against $0.035$ and $0.036$ on the coarse paired grid. The four paired-run estimates lie in $[0.030, 0.036]$, and adding the frontier grid's $0.018$ and $0.024$ gives the grid-resolution bracket $[0.018, 0.036]$, or $[0.0027, 0.0056]$ in $\delta$. At $n = 10{,}000$ the fine grid gives $0.094$ under both interpolants, with conditional interval $[0.077, 0.108]$, paired-run range $[0.094, 0.098]$, and bracket $[0.080, 0.098]$ with the frontier grid's $0.080$ and $0.083$; the curve is flat there, with SI coverage $0.904$, $0.903$, $0.914$, $0.898$, $0.906$, and $0.913$ from $\rho = 0.080$ to $0.105$ and a first clear drop at $\rho = 0.110$ ($0.883$), and the bootstrap median is $0.083$. The detection edges, $0.372$ and $0.160$ at $\alpha = 0.05$, are unchanged by the refinement.

\begin{table}[!htbp]
\centering
\caption{Monte Carlo intervals for the detection-damage band edges.}
\label{tab:edge-ci}
\footnotesize
\setlength{\tabcolsep}{4pt}
\begin{tabular}{@{}llrlll@{}}
\toprule
\textbf{Scale} & \textbf{Quantity} & $\boldsymbol{\alpha}$ & $\boldsymbol{n = 10{,}000}$ & $\boldsymbol{n = 100{,}000}$ & \textbf{Shrink factor} \\
\midrule
\multicolumn{6}{l}{\textit{Paired run}} \\
$\rho$ & Damage edge & --- & $0.098$ $[0.076, 0.118]$ & $0.035$ $[0.000, 0.042]$ & $2.79$ $[2.08, \infty)$ \\
$\rho$ & Detection edge & 0.05 & $0.372$ $[0.361, 0.380]$ & $0.160$ $[0.150, 0.165]$ & $2.32$ $[2.22, 2.48]$ \\
$\rho$ & Detection edge & 0.10 & $0.323$ $[0.307, 0.339]$ & $0.135$ $[0.128, 0.144]$ & $2.40$ $[2.21, 2.57]$ \\
$\rho$ & Band width & 0.05 & $0.273$ $[0.252, 0.297]$ & $0.125$ $[0.113, 0.163]$ & $2.19$ $[1.63, 2.52]$ \\
$\rho$ & Band width & 0.10 & $0.225$ $[0.200, 0.253]$ & $0.100$ $[0.089, 0.141]$ & $2.26$ $[1.56, 2.66]$ \\
$\delta$ & Damage edge & --- & $0.0163$ $[0.0124, 0.0200]$ & $0.0055$ $[0.0000, 0.0066]$ & $2.98$ $[2.18, \infty)$ \\
$\delta$ & Detection edge & 0.05 & $0.0887$ $[0.0849, 0.0919]$ & $0.0286$ $[0.0265, 0.0297]$ & $3.10$ $[2.92, 3.36]$ \\
$\delta$ & Detection edge & 0.10 & $0.0715$ $[0.0664, 0.0771]$ & $0.0234$ $[0.0220, 0.0252]$ & $3.06$ $[2.77, 3.36]$ \\
$\delta$ & Band width & 0.05 & $0.0723$ $[0.0672, 0.0775]$ & $0.0231$ $[0.0208, 0.0293]$ & $3.13$ $[2.41, 3.57]$ \\
$\delta$ & Band width & 0.10 & $0.0552$ $[0.0492, 0.0628]$ & $0.0179$ $[0.0160, 0.0245]$ & $3.09$ $[2.16, 3.63]$ \\
\midrule
\multicolumn{6}{l}{\textit{Paired run, refined local grid for the damage edge (linear interpolation)}} \\
$\rho$ & Damage edge & --- & $0.094$ $[0.077, 0.108]$ & $0.030$ $[0.028, 0.036]$ & $3.15$ $[2.15, 3.70]$ \\
$\rho$ & Band width & 0.05 & $0.277$ $[0.260, 0.301]$ & $0.130$ $[0.118, 0.135]$ & $2.13$ $[1.98, 2.49]$ \\
$\delta$ & Damage edge & --- & $0.0156$ $[0.0125, 0.0182]$ & $0.0046$ $[0.0043, 0.0057]$ & $3.37$ $[2.26, 4.02]$ \\
$\delta$ & Band width & 0.05 & $0.0730$ $[0.0687, 0.0791]$ & $0.0239$ $[0.0217, 0.0250]$ & $3.05$ $[2.82, 3.51]$ \\
\midrule
\multicolumn{6}{l}{\textit{Frontier-coverage grid (parametric binomial bootstrap)}} \\
$\rho$ & Damage edge & --- & $0.080$ $[0.050, 0.105]$ & $0.018$ $[0.012, 0.024]$ & $4.55$ $[2.65, 7.39]$ \\
$\rho$ & Damage edge (pchip) & --- & $0.083$ $[0.050, 0.111]$ & $0.024$ $[0.017, 0.032]$ & --- \\
$\delta$ & Damage edge & --- & $0.0130$ $[0.0079, 0.0177]$ & $0.0027$ $[0.0019, 0.0037]$ & $4.86$ $[2.74, 8.10]$ \\
\bottomrule
\end{tabular}
\smallskip

{\scriptsize \noindent Notes: DGP~2, $\pi_L = 0.20$; paired run at $R = 500$ ($n = 10{,}000$) and $R = 300$ ($n = 100{,}000$), frontier grid at $R = 500$ on a coarser $\rho$ grid. Each entry is the point estimate read off the observed curve followed by the 95\% percentile interval over $B = 1{,}000$ bootstrap draws, conditional on linear interpolation between grid points. The damage edge is where SI coverage under the joint sandwich falls below $0.90$ and does not depend on $\alpha$; the detection edge is where two-sided power exceeds $0.50$; $\delta = 0.15\,\rho/(1-\rho)$. The shrink factor is the $n = 10{,}000$ value over the $n = 100{,}000$ value. On the coarse paired grid at $n = 100{,}000$, 22\% of the resamples put the damage edge at $\rho = 0$, which makes that shrink factor unbounded in those draws and leaves its interval without a finite upper limit; the refined local grid ($\rho$ in steps of $0.005$, $R = 1{,}000$ per point and $2{,}000$ at the null) removes that. Refined-grid band widths combine the refined damage edge with the coarse-grid detection edge. The pchip row uses a monotone piecewise-cubic interpolant; under pchip the paired-run damage edges are $0.098$ and $0.036$ on the coarse grid and $0.094$ and $0.030$ on the refined grid, and the text above reconciles the grid-resolution brackets.}
\end{table}

\subsection{Ablations and Variance Checks}
\label{app:ablations}

Every ablation configuration runs inside the same replication as the primary one, on the same draws and fitted predictions, so the tables below are paired cell by cell with the main-text tables and their row differences carry much less Monte Carlo error than a comparison of independent runs, though pairing does not remove it. The PPI++ ablations use the plug-in tuning rule, the plug-in variance, a clipped coefficient, a per-arm coefficient, or the plug-in rule and plug-in variance together (the plug-in-rule estimator end to end); the surrogate-index ablations use the plug-in and the delta-method first-stage variances. This subsection also holds the variance checks behind Section~\ref{sec:surrogacy-test}.

\paragraph{The PPI++ tuning and variance rules.}

Table~\ref{tab:ppi-ablation} varies one axis at a time around the primary configuration. In the linear designs only the plug-in tuning rule moves anything material: it sits below the per-arm optimum by the factor $1/(1 + \pi_L)$, so RE falls from $1.524$ to $1.519$ at $\pi_L = 0.05$, from $1.395$ to $1.376$ at $\pi_L = 0.20$, and from $1.185$ to $1.158$ at $\pi_L = 0.50$, a loss of $0.3\%$, $1.4\%$, and $2.3\%$, with intervals $0.1\%$ to $2.3\%$ wider. At the primary coefficient the overlap correction cancels numerically in these designs, so the plug-in and exact variances agree to four decimals in width and exactly in coverage, and clipping and a per-arm coefficient are immaterial where the coefficient lies in $[0,1]$ and the arms share $(\gamma_t, \sigma^2_{\hat{Y},t})$. The rules separate under mix drift, where the labeled cohort is the biased one and the plug-in rule shrinks toward it: at $d = 2$ the primary configuration carries $-5.0\%$ relative bias at $0.939$ coverage and RE $1.81$, the plug-in rule $-11.9\%$ at $0.923$ and RE $1.71$, and the plug-in rule with the plug-in variance covers $0.918$. At $d = 1$ the same ordering holds at a tenth of the size ($-0.4\%$ against $-4.1\%$). Under effect drift, which biases every estimator, the rules differ by at most $0.1$ percentage points of relative bias. The table has no bias column; the DGP~11 relative biases in this paragraph are reported in the replication package.

\begin{table}[!htbp]
\centering
\caption{PPI++ tuning and variance rules: the primary configuration against its one-axis ablations.}
\label{tab:ppi-ablation}
\footnotesize
\setlength{\tabcolsep}{5pt}
\begin{tabular}{@{}llrrrr@{}}
\toprule
\textbf{Cell} & \textbf{Configuration} & \textbf{RMSE} & \textbf{Coverage} & \textbf{RE} & \textbf{Mean width} \\
\midrule
DGP 1, $\pi_L = 0.05$ & PPI++ (primary) & $0.0929$ & $0.953$ & $1.524$ & $0.3667$ \\
 & PPI++ (plug-in $\lambda$) & $0.0932$ & $0.949$ & $1.519$ & $0.3671$ \\
 & PPI++ (plug-in variance) & $0.0929$ & $0.953$ & $1.524$ & $0.3667$ \\
 & PPI++ (clipped) & $0.0927$ & $0.950$ & $1.526$ & $0.3661$ \\
 & PPI++ (per-arm $\lambda$) & $0.0934$ & $0.952$ & $1.515$ & $0.3674$ \\
\midrule
DGP 1, $\pi_L = 0.20$ & PPI++ (primary) & $0.0507$ & $0.950$ & $1.395$ & $0.1999$ \\
 & PPI++ (plug-in $\lambda$) & $0.0514$ & $0.949$ & $1.376$ & $0.2024$ \\
 & PPI++ (plug-in variance) & $0.0507$ & $0.950$ & $1.395$ & $0.1999$ \\
 & PPI++ (clipped) & $0.0507$ & $0.950$ & $1.396$ & $0.1999$ \\
 & PPI++ (per-arm $\lambda$) & $0.0507$ & $0.951$ & $1.394$ & $0.2000$ \\
\midrule
DGP 1, $\pi_L = 0.50$ & PPI++ (primary) & $0.0376$ & $0.953$ & $1.185$ & $0.1465$ \\
 & PPI++ (plug-in $\lambda$) & $0.0385$ & $0.954$ & $1.158$ & $0.1499$ \\
 & PPI++ (plug-in variance) & $0.0376$ & $0.953$ & $1.185$ & $0.1465$ \\
 & PPI++ (clipped) & $0.0376$ & $0.953$ & $1.185$ & $0.1464$ \\
 & PPI++ (per-arm $\lambda$) & $0.0376$ & $0.953$ & $1.185$ & $0.1465$ \\
\midrule
DGP 2, $\rho = 0.2$ & PPI++ (primary) & $0.0508$ & $0.955$ & $1.378$ & $0.2001$ \\
 & PPI++ (plug-in $\lambda$) & $0.0515$ & $0.948$ & $1.360$ & $0.2025$ \\
 & PPI++ (plug-in variance) & $0.0508$ & $0.955$ & $1.378$ & $0.2001$ \\
 & PPI++ (clipped) & $0.0508$ & $0.955$ & $1.378$ & $0.2000$ \\
 & PPI++ (per-arm $\lambda$) & $0.0509$ & $0.954$ & $1.376$ & $0.2002$ \\
\midrule
DGP 9, $\pi_L = 0.05$ & PPI++ (primary) & $0.0947$ & $0.948$ & $1.195$ & $0.3585$ \\
 & PPI++ (clipped) & $0.0947$ & $0.948$ & $1.195$ & $0.3582$ \\
\midrule
DGP 11, mix drift $d = 1$ & PPI++ (primary) & $0.0497$ & $0.937$ & $1.385$ & --- \\
 & PPI++ (plug-in $\lambda$) & $0.0505$ & $0.942$ & $1.364$ & --- \\
 & PPI++ (plug-in $\lambda$, plug-in var.) & $0.0505$ & $0.939$ & $1.364$ & --- \\
\midrule
DGP 11, mix drift $d = 2$ & PPI++ (primary) & $0.0499$ & $0.939$ & $1.805$ & --- \\
 & PPI++ (plug-in $\lambda$) & $0.0525$ & $\mathbf{0.923}$ & $1.715$ & --- \\
 & PPI++ (plug-in $\lambda$, plug-in var.) & $0.0525$ & $\mathbf{0.918}$ & $1.715$ & --- \\
\midrule
DGP 11, effect drift $g_e = 0.2$ & PPI++ (primary) & $0.0973$ & $\mathbf{0.568}$ & $1.078$ & --- \\
 & PPI++ (plug-in $\lambda$) & $0.0978$ & $\mathbf{0.587}$ & $1.073$ & --- \\
 & PPI++ (plug-in $\lambda$, plug-in var.) & $0.0978$ & $\mathbf{0.571}$ & $1.073$ & --- \\
\bottomrule
\end{tabular}
\smallskip

{\scriptsize \noindent Notes: $n = 10{,}000$, all-units five-fold cross-fitting; within a cell every configuration runs on the same draws and the same fitted prediction model. DGP~1 and DGP~2 cells are at $R = 2{,}000$ with the band $[0.940, 0.960]$, and the DGP~2 cell is at $\pi_L = 0.20$; DGP~9 is at $R = 500$ with the band $[0.931, 0.969]$; the DGP~11 cells are at $\pi_L = 0.20$ and $R = 1{,}000$ with the band $[0.936, 0.964]$. \textbf{Bold} coverage entries fall outside the band for their row. ``PPI++ (plug-in $\lambda$, plug-in var.)'' is the plug-in-rule estimator end to end: the plug-in tuning rule of Equation~\eqref{eq:lambda-opt}, the coefficient clipped to $[0,1]$, and the plug-in variance. RMSE and mean width are printed to four decimals because the configurations differ in the third or fourth, and RE to three for the same reason; RE is the RMSE ratio $\mathrm{RMSE}_{\mathrm{LO}}/\mathrm{RMSE}_{\mathrm{method}}$. The DGP~11 run records no mean interval width, so that column reads --- there, and each block prints only the configurations run in it.}
\end{table}

\paragraph{The surrogate-index variance.}

Table~\ref{tab:si-variance-ablation} holds the surrogate-index point estimate fixed and varies only the interval: ``plug-in'' treats the fitted index as fixed, ``delta-method'' adds $d'V_\beta d$, and the sandwich adds the cross term as well. Compared with Section~\ref{sec:results-gbt}, this ablation additionally shows that on the rich index, where the plug-in interval is $54\%$ of the sandwich width, the delta-method construction recovers almost all of the gap, so the omitted cross term is small on that design, not small in general.

\begin{table}[!htbp]
\centering
\caption{Surrogate-index coverage under the three interval constructions.}
\label{tab:si-variance-ablation}
\footnotesize
\setlength{\tabcolsep}{5pt}
\begin{tabular}{@{}lrrrrr@{}}
\toprule
& \multicolumn{3}{c}{\textbf{Coverage}} & & \\
\cmidrule(lr){2-4}
\textbf{Cell} & plug-in & delta & sandwich & \textbf{Width ratio} & \textbf{1st-stage share} \\
\midrule
DGP 1, $\pi_L = 0.05$ & $0.941$ & $0.948$ & $0.949$ & $1.038$ & $0.072$ \\
DGP 1, $\pi_L = 0.20$ & $0.952$ & $0.956$ & $0.956$ & $1.009$ & $0.019$ \\
DGP 1, $\pi_L = 0.50$ & $0.947$ & $0.950$ & $0.949$ & $1.004$ & $0.007$ \\
DGP 1, $\pi_L = 1.00$ & $0.946$ & $0.946$ & $0.946$ & $1.002$ & $0.004$ \\
DGP 2, $\rho = 0.0$ & $0.946$ & $0.948$ & $0.948$ & $1.009$ & $0.019$ \\
DGP 2, $\rho = 0.2$ & $\mathbf{0.678}$ & $\mathbf{0.681}$ & $\mathbf{0.682}$ & $1.011$ & $0.021$ \\
DGP 2, $\rho = 0.4$ & $\mathbf{0.019}$ & $\mathbf{0.021}$ & $\mathbf{0.021}$ & $1.012$ & $0.023$ \\
\midrule
Multi-surr.\ (rich), $m = 1$ & $\mathbf{0.665}$ & $0.960$ & $0.960$ & $1.852$ & $0.709$ \\
\bottomrule
\end{tabular}
\smallskip

{\scriptsize \noindent Notes: the point estimate is identical across the three coverage columns, which differ only in the variance. Width ratio is the sandwich width over the plug-in width; ``1st-stage share'' is $1$ minus the squared inverse of that ratio, the fraction of the sandwich variance the plug-in construction omits. DGP~1 and DGP~2 rows: $n = 10{,}000$, $R = 2{,}000$, all-units five-fold cross-fitting, band $[0.940, 0.960]$; DGP~2 rows at $\pi_L = 0.20$. The multi-surrogate row is the Criteo-calibrated rich index at $n = 64{,}000$, $\pi_L = 0.20$, $R = 200$, from Table~\ref{tab:joint-cov-check}, where the band is $[0.920, 0.980]$; its width ratio and share are computed from the mean variances rather than the mean widths. \textbf{Bold} coverage entries fall outside the relevant band.}
\end{table}

\paragraph{Variance checks outside Proposition~\ref{prop:joint-si-ppi}.}
Two of the three designs the proposition does not cover (Table~\ref{tab:implementation}) are checked on two cells each. For the gradient-boosted-tree rows, which use the fixed-predictor plug-in variance and covariance, the check runs DGP~1 with $\pi_L = 0.20$ and DGP~2 with $\rho = 0.2$ on the GBT comparison's own draws ($R = 200$, with a paired bootstrap of $B = 100$ draws that refit the trees). The plug-in variance is $0.95$ (Monte Carlo standard error $0.11$) and $1.13$ ($0.11$) of the Monte Carlo variance of $\hat\tau_{\mathrm{SI}}$, within Monte Carlo error of one on these two cells and silent about other designs; surrogate-index coverage is $0.940$ at DGP~1 and $0.665$ at DGP~2, where the shortfall is the $-0.039$ bias. The fixed-predictor $\Var(\hat D)$ is $1.07$ ($0.11$) and $0.94$ ($0.10$) of the Monte Carlo variance, and the diagnostic's size at DGP~1 is $0.065$ ($0.017$) at $\alpha = 0.05$ and $0.105$ ($0.022$) at $\alpha = 0.10$. A bootstrap that assigns folds to resampled rows lets copies of one unit fall on both sides of a fold, where trees partly memorize them, and understates $\Var(\hat\tau_{\mathrm{PPI++}})$ at $0.74$ and $0.79$ and $\Var(\hat D)$ at $0.67$ and $0.60$ of the Monte Carlo variance. Keeping every copy in its unit's fold removes the shortfall: on the DGP~1 cell, with $B = 50$, the clustered-fold bootstrap is $1.045$ ($0.116$), $1.026$ ($0.109$), and $1.168$ ($0.124$) of the Monte Carlo variances of $\hat\tau_{\mathrm{SI}}$, $\hat\tau_{\mathrm{PPI++}}$, and $\hat D$, its intervals cover $0.955$ for SI and $0.945$ for PPI++, and the diagnostic's size is $0.045$ ($0.015$) at $\alpha = 0.05$. With clustered folds the tree bootstrap is a validation reference on this one cell, and it agrees with the direct check.
For the MAR and enrollment-drift rows, which apply the MCAR sandwich with the labeled-design $\hat Q$, the check runs DGP~5 MAR with $q = 0.10$ and DGP~11 mix drift with $d = 2$ at $R = 500$, with a paired bootstrap of $B = 200$ draws that resamples units within arm with their labels. Surrogacy holds and the index is correctly specified in both cells. The sandwich is $1.075$ (Monte Carlo standard error $0.074$) and $0.936$ ($0.064$) of the Monte Carlo variance of $\hat\tau_{\mathrm{SI}}$, against $1.083$ and $0.937$ for the bootstrap, and $0.971$ and $0.969$ of that of $\hat{D}$; surrogate-index coverage is $0.946$ and $0.944$. The diagnostic rejects on $0.046$ and $0.042$ of replications at $\alpha = 0.05$, but these are rejection rates, not sizes: unadjusted PPI++ can be miscentered, as under mix drift, so they do not validate the diagnostic under nonrandom labeling. PPI++'s exact variance is $0.967$ of the Monte Carlo variance under MAR but $0.873$ ($0.051$) under mix drift, where the bootstrap also underestimates it, at $0.900$, and coverage is $0.932$, so part of the PPI++ undercoverage in Table~\ref{tab:enrollment} is a variance shortfall. DGP~4 is not checked separately: its interval treats the external fit as fixed, and its coverage, $0.950$ at $\Delta_\beta = 0$ and $\pi_L = 0.20$ over redrawn calibration samples, is already unconditional.

\paragraph{Bootstrap sensitivity.}

Table~\ref{tab:bootstrap-sensitivity} runs the bootstrap variant of the PPI++ variance under both tuning rules on the same draws. Under the plug-in rule the pattern of Section~\ref{sec:overlap-correction} reappears at $R = 500$: at $\pi_L = 0.80$ the plug-in variance covers $90.0\%$, the exact variance $94.6\%$, and the bootstrap $90.6\%$, $93.6\%$, $94.2\%$, and $94.8\%$ at $B = 50$, $200$, $500$, and $1{,}000$. Under the primary rule the two analytic columns coincide and are inside the band at every labeled fraction. A practitioner who bootstraps should use at least $B = 500$.

\begin{table}[!htbp]
\centering
\caption{Bootstrap sensitivity analysis for the PPI++ variance estimator.}
\label{tab:bootstrap-sensitivity}
\footnotesize
\setlength{\tabcolsep}{3pt}
\begin{tabular}{@{}lrrrrrrr@{}}
\toprule
$\boldsymbol{\pi_L}$ & \textbf{Plug-in} & \textbf{Exact} & $\boldsymbol{B=50}$ & $\boldsymbol{B=100}$ & $\boldsymbol{B=200}$ & $\boldsymbol{B=500}$ & $\boldsymbol{B=1000}$ \\
\midrule
\multicolumn{8}{l}{\textit{Block A: the plug-in-rule estimator (plug-in tuning rule, clipped coefficient)}} \\
0.05 & 0.956 (0.365) & 0.956 (0.367) & \textbf{0.920} (0.337) & 0.942 (0.351) & 0.942 (0.358) & 0.950 (0.364) & 0.954 (0.366) \\
0.10 & 0.948 (0.265) & 0.956 (0.268) & \textbf{0.918} (0.248) & 0.934 (0.257) & 0.948 (0.263) & 0.946 (0.266) & 0.948 (0.267) \\
0.20 & 0.934 (0.196) & 0.940 (0.202) & \textbf{0.900} (0.186) & \textbf{0.922} (0.194) & \textbf{0.924} (0.197) & 0.942 (0.200) & 0.940 (0.201) \\
0.50 & 0.942 (0.136) & 0.964 (0.150) & 0.940 (0.139) & 0.958 (0.144) & 0.954 (0.147) & 0.956 (0.149) & 0.962 (0.149) \\
0.80 & \textbf{0.900} (0.113) & 0.946 (0.131) & \textbf{0.906} (0.122) & \textbf{0.928} (0.126) & 0.936 (0.128) & 0.942 (0.130) & 0.948 (0.131) \\
1.00 & 0.946 (0.124) & 0.946 (0.124) & \textbf{0.924} (0.114) & \textbf{0.930} (0.119) & 0.938 (0.121) & 0.940 (0.123) & 0.944 (0.123) \\
\midrule
\multicolumn{8}{l}{\textit{Block B: the primary PPI++ (exact tuning rule, common unclipped coefficient)}} \\
0.05 & 0.956 (0.366) & 0.956 (0.366) & --- & --- & --- & --- & --- \\
0.10 & 0.950 (0.267) & 0.950 (0.267) & --- & --- & --- & --- & --- \\
0.20 & 0.944 (0.200) & 0.944 (0.200) & --- & --- & --- & --- & --- \\
0.50 & 0.966 (0.147) & 0.966 (0.147) & --- & --- & --- & --- & --- \\
0.80 & 0.948 (0.130) & 0.948 (0.130) & --- & --- & --- & --- & --- \\
1.00 & 0.946 (0.124) & 0.946 (0.124) & --- & --- & --- & --- & --- \\
\bottomrule
\end{tabular}
\smallskip

{\scriptsize \noindent Notes: DGP 1, $n = 10{,}000$, $R = 500$, all-units five-fold cross-fitting. Each cell shows coverage with the mean CI width in parentheses. The two blocks run on the same draws. \textbf{Bold} coverage entries fall outside the $R = 500$ Monte Carlo band $[0.931, 0.969]$. ``Plug-in'' and ``Exact'' are the two analytic variances of Proposition~\ref{prop:corrected-variance}; in block~B the tuning rule is the minimizer of the exact variance, at which the correction cancels numerically in this design, and the two columns coincide at every labeled fraction (Section~\ref{sec:overlap-correction}). Block~B has no bootstrap columns because the bootstrap loop re-estimates $\lambda$ under the plug-in rule and clips it, so its interval belongs to block~A. An independent $R = 2{,}000$ run of the same comparison is in Table~\ref{tab:coverage-correction}.}
\end{table}

\paragraph{The joint covariance.} Tables~\ref{tab:joint-cov-check} and~\ref{tab:joint-cov-R2000} are the calibration checks summarized in Section~\ref{sec:surrogacy-test}. The joint paired bootstrap resamples units with replacement stratified by arm, keeps every copy of a resampled unit in its unit's fold, refits the index on the resampled labeled units under the all-units protocol, and recomputes both estimators, so it carries every source of first-stage uncertainty. Against it the sandwich lies slightly below for $\Var(\hat\tau_{\mathrm{PPI++}})$ ($0.985$, $0.987$, $0.983$, $0.954$) and $\Var(\hat D)$ ($0.979$, $0.976$, $0.983$, $0.942$) in every bootstrapped cell. Against the Monte Carlo variance the ratios run from $0.79$ to $1.31$ at $R = 200$, a spread the roughly $10\%$ Monte Carlo standard error of a variance ratio at that $R$ allows, and from $0.948$ to $1.040$ at $R = 2{,}000$. On the single-surrogate cells the ``SI (plug-in variance)'' ablation is within about $10\%$ of the Monte Carlo variance (ratios $0.923$, $1.012$, and $0.906$ against the sandwich's $0.941$, $1.032$, and $0.926$), and the two intervals cover within $0.010$ of each other.

\paragraph{One cell across runs.} Surrogate-index coverage at DGP~2, $\rho = 0$, $n = 10^5$ is $0.930$ in Table~\ref{tab:joint-cov-check} ($R = 200$), $0.910$ in the paired run of Section~\ref{sec:detection-damage} ($R = 300$), $0.952$ on the frontier-coverage grid ($R = 500$), $0.958$ at $R = 2{,}000$ (Table~\ref{tab:joint-cov-R2000}, variance ratio $1.022$), and $0.942$ at $R = 2{,}000$ on the refined local grid. PPI++ covers $0.915$ in Table~\ref{tab:joint-cov-check}, $1.8$ Monte Carlo standard errors below nominal, and $0.945$ at $R = 2{,}000$. The three $R = 2{,}000$ values lie inside the band $[0.940, 0.960]$.

\begin{table}[!htbp]
\centering
\caption{Joint paired-bootstrap validation of the learned-index sandwich.}
\label{tab:joint-cov-check}
\footnotesize
\setlength{\tabcolsep}{3pt}
\begin{tabular}{@{}llrrrrrrrr@{}}
\toprule
& & \multicolumn{4}{c}{\textbf{Sandwich / bootstrap}} & \multicolumn{2}{c}{\textbf{Coverage}} & \multicolumn{2}{c}{\textbf{Rejection}} \\
\cmidrule(lr){3-6}\cmidrule(lr){7-8}\cmidrule(lr){9-10}
\textbf{Design} & \textbf{Cell} & $\Var_{\mathrm{SI}}$ & $\Var_{\mathrm{PPI}}$ & $\Cov$ & $\Var(\hat{D})$ & SI & PPI++ & $\alpha = .05$ & $\alpha = .10$ \\
\midrule
DGP 1 & $\pi_L = 0.20$, $n = 10^4$ & $0.998$ & $0.985$ & $1.001$ & $0.979$ & $0.930$ & $0.925$ & $0.080$ & $0.145$ \\
DGP 2 & $\rho = 0$, $n = 10^4$ & $1.008$ & $0.987$ & $1.016$ & $0.976$ & $0.960$ & $0.960$ & $0.060$ & $0.100$ \\
DGP 2 & $\rho = 0.2$, $n = 10^4$ & $0.998$ & $0.983$ & $0.990$ & $0.983$ & $\mathbf{0.665}$ & $0.975$ & $0.115$ & $0.210$ \\
Multi-surr.\ (rich) & $m = 1$, $n = 64{,}000$ & $0.988$ & $0.954$ & $1.018$ & $0.942$ & $0.960$ & $0.920$ & $0.075$ & $0.130$ \\
DGP 2 & $\rho = 0$, $n = 10^5$ & --- & --- & --- & --- & $0.930$ & $\mathbf{0.915}$ & $0.055$ & $0.100$ \\
\bottomrule
\end{tabular}
\smallskip

{\scriptsize \noindent Notes: $R = 200$ replications per row, with $B = 200$ bootstrap draws each except in the $n = 10^5$ row, which is run without the bootstrap and whose ratio columns therefore read ---. The Monte Carlo band around the nominal 95\% at $R = 200$ is $[0.920, 0.980]$ and \textbf{bold} coverage entries fall outside it; Monte Carlo standard errors on the coverage columns run from $0.011$ to $0.033$ and on the rejection columns from $0.016$ to $0.025$. The ratio columns are the mean sandwich variance over the mean paired-bootstrap variance, with every copy of a resampled unit kept in its unit's fold, so $1.000$ means the two agree exactly at that cell; a bootstrap that assigns folds to resampled rows, letting copies of a unit fall on both sides of a fold, gives PPI++ and $\hat D$ variances 1 to 2\% lower (4 to 5\% on the rich index) and ratios within $2.1\%$ of one. The rejection columns are two-sided. They are a \emph{size} in the DGP~1 and DGP~2 $\rho = 0$ rows, where the joint null holds exactly; the rich index at $m = 1$ satisfies surrogacy by construction but not correct specification of the OLS basis, so its row is a near-null stress check (Section~\ref{sec:semisynthetic}); the $\rho = 0.2$ row is a \emph{power}. The bootstrap $\mathrm{SE}(\hat{D})$ gives $0.080$ and $0.160$ on DGP~1, $0.050$ and $0.100$ on DGP~2 at $\rho = 0$, and $0.065$ and $0.120$ on the rich index, matching the sandwich within Monte Carlo error. The two $n = 10^5$ coverage entries are compared with the other runs of that cell in Appendix~\ref{app:ablations}.
The replication package also reports the plug-in, delta-method, and sandwich surrogate-index variances cell by cell.}
\end{table}

\begin{table}[!htbp]
\centering
\caption{The learned-index sandwich against the Monte Carlo variance at $R = 2{,}000$.}
\label{tab:joint-cov-R2000}
\footnotesize
\setlength{\tabcolsep}{3pt}
\begin{tabular}{@{}llrrrrrrrr@{}}
\toprule
& & \multicolumn{4}{c}{\textbf{Sandwich / Monte Carlo}} & \multicolumn{2}{c}{\textbf{Coverage}} & \multicolumn{2}{c}{\textbf{Rejection}} \\
\cmidrule(lr){3-6}\cmidrule(lr){7-8}\cmidrule(lr){9-10}
\textbf{Design} & \textbf{Cell} & $\Var_{\mathrm{SI}}$ & $\Var_{\mathrm{PPI}}$ & $\Cov$ & $\Var(\hat{D})$ & SI & PPI++ & $\alpha = .05$ & $\alpha = .10$ \\
\midrule
DGP 1 & $\pi_L = 0.20$, $n = 10^4$ & $1.007$ & $0.995$ & $0.990$ & $1.002$ & $0.954$ & $0.944$ & $0.051$ & $0.100$ \\
DGP 2 & $\rho = 0$, $n = 10^4$ & $1.003$ & $0.977$ & $0.995$ & $0.973$ & $0.955$ & $0.946$ & $0.056$ & $0.107$ \\
DGP 2 & $\rho = 0.2$, $n = 10^4$ & $1.033$ & $1.011$ & $1.008$ & $1.019$ & $\mathbf{0.676}$ & $0.954$ & $0.135$ & $0.215$ \\
Multi-surr.\ (rich) & $m = 1$, $n = 64{,}000$ & $1.018$ & $0.960$ & $1.035$ & $0.948$ & $0.951$ & $0.944$ & $0.058$ & $0.105$ \\
DGP 2 & $\rho = 0$, $n = 10^5$ & $1.022$ & $1.018$ & $1.040$ & $1.006$ & $0.958$ & $0.945$ & $0.054$ & $0.107$ \\
\bottomrule
\end{tabular}
\smallskip

{\scriptsize \noindent Notes: $R = 2{,}000$ replications per row, without the bootstrap; same cells, seeds, and code path as Table~\ref{tab:joint-cov-check}, whose 200 replications are the first 200 here. The ratio columns are the mean sandwich (co)variance over the Monte Carlo (co)variance of the estimators across replications, so $1.000$ means the sandwich is exact at that cell; their delta-method Monte Carlo standard errors run from $0.030$ to $0.069$, and every ratio is within two of them of one. The Monte Carlo band around the nominal 95\% at $R = 2{,}000$ is $[0.940, 0.960]$ and \textbf{bold} coverage entries fall outside it; coverage standard errors are $0.005$ ($0.010$ for SI at $\rho = 0.2$). The rejection columns are two-sided, with standard errors of $0.005$ and $0.007$ where the null holds; they are a \emph{size} in the three exact-null rows, a near-null stress result on the rich index, and a \emph{power} at $\rho = 0.2$. On the rich index the ``SI (plug-in variance)'' interval covers $0.701$ against the sandwich's $0.951$.}
\end{table}

\subsection{Full Per-DGP and Real-Data Results}
\label{app:full-results}

The replication package (Code and data, Section~\ref{sec:discussion}) reports the complete per-DGP results for every method, configuration, and labeled fraction, including DGP~1 at $\pi_L = 1.00$, DGP~2 at every $\rho$ and labeled fraction, the DGP~5 MAR rows, the DGP~6 portfolios at $K \in \{50, 100, 200\}$ and $\pi_L \in \{0.20, 1.00\}$, DGP~10 at every labeled fraction, and the eight DGP~11 configurations. The tables below are the ones the text cites for numbers reported nowhere else: the Criteo estimators at full scale, the Hillstrom hybrid across labeled fractions, the $\pi_L = 0.20$ slices of DGPs~4 and~10, DGPs~7 and~8, and DGP~9 with its diagnostic and hybrid blocks.

Every simulation table in this subsection runs at $n = 10{,}000$ under the all-units five-fold cross-fitting protocol of Section~\ref{sec:eightmethods}, with two exceptions: the DGP~4 surrogate index, fit once on external historical units, and the composite proxy, whose weights are calibrated on estimated effects from a library of past experiments rather than on unit-level predictions (Method~4 and Table~\ref{tab:implementation}). Every row is a primary configuration except ``PPI++ (clipped)'' in DGP~9, where the clip binds. Unlike the main-text tables, these report the raw RMSE ratio for every method, so an RE entry beside a large bias or poor coverage is not an efficiency gain. \textbf{Bold} coverage entries fall outside the Monte Carlo band at the table's $R$.

\begin{table}[!ht]
\centering
\caption{Criteo full estimator results at the complete sample size.}
\label{tab:criteo-full}
\footnotesize
\begin{tabular}{@{}llrrrrr@{}}
\toprule
$\boldsymbol{\pi_L}$ & \textbf{Method} & \textbf{Bias} & \textbf{Rel.\ bias (\%)} & \textbf{RMSE} & \textbf{Coverage} & \textbf{RE} \\
\midrule
\multirow{3}{*}{0.05} & Labeled-Only & $+0.01$ & $0.1$ & $1.52$ & $0.970$ & $1.00$ \\
 & Surrogate Index & $-5.63$ & $-48.9$ & $5.63$ & \textbf{0.000} & --- \\
 & PPI++ & $+0.02$ & $0.2$ & $1.45$ & $0.980$ & $1.05$ \\
\midrule
\multirow{3}{*}{0.20} & Labeled-Only & $+0.09$ & $0.8$ & $0.66$ & $0.990$ & $1.00$ \\
 & Surrogate Index & $-5.62$ & $-48.8$ & $5.62$ & \textbf{0.000} & --- \\
 & PPI++ & $+0.07$ & $0.6$ & $0.62$ & \textbf{1.000} & $1.06$ \\
\bottomrule
\end{tabular}
\smallskip

{\scriptsize \noindent Notes: $R = 100$ random splits at the full $n = 13{,}979{,}592$, all-units five-fold cross-fitting. Bias and RMSE are in units of $10^{-4}$ and the masking target, the full-sample difference in means, is $11.5 \times 10^{-4}$. PPI++ is the primary configuration and the surrogate index uses the joint sandwich variance; the naive surrogate, composite proxy, and AIPW are not run at this scale. RE is the RMSE ratio $\mathrm{RMSE}_{\mathrm{LO}}/\mathrm{RMSE}_{\mathrm{method}}$, suppressed (---) for methods with absolute relative bias above 10\%. At $R = 100$ the Monte Carlo band around the nominal 95\% is $[0.907, 0.993]$; \textbf{bold} coverage entries fall outside it, in either direction. The SI--PPI++ estimator-disagreement diagnostic (two-sided, $\alpha = 0.05$) rejects on 99\% of splits at $\pi_L = 0.05$ (mean $T_n = 4.0$) and 100\% at $\pi_L = 0.20$ (mean $T_n = 8.1$). Coverage is measured against the full-sample difference in means, which the labeled subsample partially overlaps, making LO and PPI++ coverage conservative.}
\end{table}

\begin{table}[!htbp]
\centering
\caption{Hillstrom hybrid estimator performance across labeled fractions.}
\label{tab:hillstrom-hybrid}
\footnotesize
\setlength{\tabcolsep}{4pt}
\begin{tabular}{@{}lcccccc@{}}
\toprule
$\pi_L$ & \makecell[c]{\textbf{Hybrid bias}\\\textbf{(MC SE)}} & \makecell[c]{\textbf{Abs.\ Rel.}\\\textbf{Bias (\%)}} & \textbf{RMSE} & \makecell[c]{\textbf{Coverage}\\\textbf{(MC SE)}} & \makecell[c]{\textbf{Mean}\\\textbf{weight}} & \makecell[c]{\textbf{Diagnostic}\\\textbf{reject}} \\
\midrule
$0.05$ & $-0.00041$ (0.00009) & $8.4$  & $0.00201$ & $0.940$ (0.011) & $0.69$ & $0.094$ \\
$0.10$ & $-0.00045$ (0.00007) & $9.1$  & $0.00155$ & $0.938$ (0.011) & $0.66$ & $0.110$ \\
$0.20$ & $-0.00053$ (0.00005) & $10.7$ & $0.00115$ & $\mathbf{0.930}$ (0.011) & $0.64$ & $0.114$ \\
$0.50$ & $-0.00051$ (0.00003) & $10.3$ & $0.00079$ & $\mathbf{0.984}$ (0.006) & $0.53$ & $0.168$ \\
\bottomrule
\end{tabular}
\smallskip

{\scriptsize \noindent Notes: $R = 500$ random splits, $c = 1.5$, all-units five-fold cross-fitting. Masking target $0.004955$. Mean weight is the expected Cauchy-kernel weight on SI. Absolute relative bias is $|\text{Bias}/\tau| \times 100\%$, with Monte Carlo standard errors in parentheses. \textbf{Bold} coverage entries fall outside the band $[0.931, 0.969]$. Because the diagnostic rejects on only 9--17\% of splits, the mean weight stays above $0.5$ at every $\pi_L$ although SI covers $0.072$ at $\pi_L = 0.20$ (Table~\ref{tab:hillstrom}); the 98\% coverage at $\pi_L = 0.50$ reflects the simulation-calibrated interval widening, not nominal coverage of an unbiased estimator. Every row builds $\hat{\Sigma}$ and $\mathrm{SE}(\hat{D})$ from the joint sandwich. Re-running the $\pi_L = 0.20$ row with the plug-in SI variance, on identical seeds, masks, and point estimates, gives mean weight $0.629$, absolute relative bias $10.4\%$, coverage $0.924$, and rejection $0.128$, against $0.638$, $10.7\%$, $0.930$, and $0.114$ here, although the mean $\Var(\hat{\tau}_{\mathrm{SI}})$ differs by a factor of $4.8$ ($3.01 \times 10^{-8}$ against $1.45 \times 10^{-7}$); the delta-method first-stage variance reproduces the sandwich row to three decimals.}
\end{table}
\FloatBarrier  

\begin{table}[!htbp]
\centering
\begin{minipage}[t]{0.49\textwidth}
\centering
\caption{DGP 4 (Stale Historical Index) at $\pi_L = 0.20$.}
\label{tab:appendix-dgp4}
\scriptsize
\setlength{\tabcolsep}{2pt}
\begin{tabular}{@{}llrrrr@{}}
\toprule
\textbf{Config} & \textbf{Method} & \textbf{Bias} & \textbf{RMSE} & \textbf{Cov.} & \textbf{RE} \\
\midrule
$\Delta_\beta{=}0.0$ & LO & $-0.001$ & $0.070$ & $0.950$ & $1.00$ \\
$\tau{=}0.15$ & NS & $+0.012$ & $0.027$ & $\mathbf{0.923}$ & $2.63$ \\
 & SI & $+0.0000$ & $0.024$ & $0.950$ & $2.90$ \\
 & PPI++ & $-0.0002$ & $0.051$ & $0.948$ & $1.38$ \\
 & CP & $+0.002$ & $0.023$ & $0.949$ & $3.11$ \\
 & AIPW & $-0.0002$ & $0.051$ & $0.950$ & $1.39$ \\
\midrule
$\Delta_\beta{=}0.1$ & LO & $-0.001$ & $0.078$ & $0.950$ & $1.00$ \\
$\tau{=}0.18$ & NS & $+0.013$ & $0.032$ & $\mathbf{0.917}$ & $2.41$ \\
 & SI & $-0.029$ & $0.039$ & $\mathbf{0.762}$ & $2.02$ \\
 & PPI++ & $-0.0002$ & $0.054$ & $0.949$ & $1.45$ \\
 & CP & $-0.027$ & $0.036$ & $\mathbf{0.773}$ & $2.18$ \\
 & AIPW & $-0.0003$ & $0.054$ & $0.949$ & $1.45$ \\
\midrule
$\Delta_\beta{=}0.2$ & LO & $+0.001$ & $0.087$ & $0.952$ & $1.00$ \\
$\tau{=}0.21$ & NS & $+0.012$ & $0.035$ & $\mathbf{0.934}$ & $2.45$ \\
 & SI & $-0.060$ & $0.065$ & $\mathbf{0.309}$ & $1.34$ \\
 & PPI++ & $-0.001$ & $0.055$ & $0.952$ & $1.56$ \\
 & CP & $-0.058$ & $0.062$ & $\mathbf{0.284}$ & $1.40$ \\
 & AIPW & $-0.001$ & $0.055$ & $0.951$ & $1.56$ \\
\midrule
$\Delta_\beta{=}0.4$ & LO & $+0.003$ & $0.106$ & $\mathbf{0.937}$ & $1.00$ \\
$\tau{=}0.27$ & NS & $+0.013$ & $0.045$ & $\mathbf{0.938}$ & $2.38$ \\
 & SI & $-0.119$ & $0.122$ & $\mathbf{0.003}$ & $0.87$ \\
 & PPI++ & $+0.003$ & $0.061$ & $0.953$ & $1.73$ \\
 & CP & $-0.117$ & $0.119$ & $\mathbf{0.001}$ & $0.89$ \\
 & AIPW & $+0.003$ & $0.061$ & $0.955$ & $1.73$ \\
\bottomrule
\end{tabular}
\smallskip

{\scriptsize \noindent Notes: $R = 2{,}000$, band $[0.940, 0.960]$. The surrogate index is the historical fit with its plug-in variance (Table~\ref{tab:implementation}). Labeled-only at $\Delta_\beta = 0.4$ is the one labeled-only, PPI++, or AIPW entry outside the band; that configuration is not in Figure~\ref{fig:coverage}.}
\end{minipage}\hfill
\begin{minipage}[t]{0.49\textwidth}
\centering
\caption{DGP 10 (Nonlinear + Partial Mediation) at $\pi_L = 0.20$.}
\label{tab:appendix-dgp10}
\scriptsize
\setlength{\tabcolsep}{2pt}
\begin{tabular}{@{}llrrrr@{}}
\toprule
\textbf{Config} & \textbf{Method} & \textbf{Bias} & \textbf{RMSE} & \textbf{Cov.} & \textbf{RE} \\
\midrule
$\rho{=}0.0$ & LO & $+0.002$ & $0.071$ & $0.948$ & $1.00$ \\
$\tau{=}0.1473$ & NS & $+0.011$ & $0.026$ & $\mathbf{0.928}$ & $2.77$ \\
 & SI & $-0.001$ & $0.024$ & $\mathbf{0.970}$ & $2.98$ \\
 & PPI++ & $+0.002$ & $0.051$ & $0.950$ & $1.40$ \\
 & CP & $-0.001$ & $0.021$ & $0.966$ & $3.35$ \\
 & AIPW & $+0.002$ & $0.051$ & $0.948$ & $1.40$ \\
\midrule
$\rho{=}0.2$ & LO & $+0.003$ & $0.064$ & $\mathbf{0.976}$ & $1.00$ \\
$\tau{=}0.1841$ & NS & $-0.024$ & $0.034$ & $\mathbf{0.816}$ & $1.88$ \\
 & SI & $-0.036$ & $0.044$ & $\mathbf{0.674}$ & $1.48$ \\
 & PPI++ & $+0.001$ & $0.047$ & $0.968$ & $1.38$ \\
 & CP & $-0.036$ & $0.042$ & $\mathbf{0.614}$ & $1.53$ \\
 & AIPW & $+0.001$ & $0.047$ & $0.968$ & $1.38$ \\
\midrule
$\rho{=}0.4$ & LO & $+0.003$ & $0.072$ & $0.948$ & $1.00$ \\
$\tau{=}0.2455$ & NS & $-0.086$ & $0.089$ & $\mathbf{0.070}$ & $0.81$ \\
 & SI & $-0.097$ & $0.101$ & $\mathbf{0.036}$ & $0.72$ \\
 & PPI++ & $+0.002$ & $0.050$ & $0.944$ & $1.45$ \\
 & CP & $-0.098$ & $0.100$ & $\mathbf{0.010}$ & $0.72$ \\
 & AIPW & $+0.002$ & $0.050$ & $0.942$ & $1.45$ \\
\bottomrule
\end{tabular}
\smallskip

{\scriptsize \noindent Notes: $R = 500$, band $[0.931, 0.969]$. SI relative bias is $-0.8\%$, $-19.4\%$, and $-39.7\%$ at $\rho = 0$, $0.2$, and $0.4$.}
\end{minipage}
\end{table}

\begin{table}[!htbp]
\centering
\begin{minipage}[t]{0.48\textwidth}
\centering
\caption{Full results for DGP 7 (Multiple Surrogates, 3 Valid).}
\label{tab:appendix-dgp7}
\scriptsize
\setlength{\tabcolsep}{2pt}
\begin{tabular}{@{}llrrrr@{}}
\toprule
\textbf{Config} & \textbf{Method} & \textbf{Bias} & \textbf{RMSE} & \textbf{Cov.} & \textbf{RE} \\
\midrule
$\pi_L{=}0.05$ & LO & $+0.0003$ & $0.145$ & $\mathbf{0.930}$ & $1.00$ \\
$\pi_L{=}0.05$ & SI & $+0.001$ & $0.027$ & $0.944$ & $5.46$ \\
$\pi_L{=}0.05$ & PPI++ & $+0.0002$ & $0.092$ & $0.946$ & $1.58$ \\
$\pi_L{=}0.05$ & CP & $-0.0003$ & $0.021$ & $0.956$ & $7.01$ \\
\midrule
$\pi_L{=}0.2$ & LO & $+0.006$ & $0.070$ & $0.942$ & $1.00$ \\
$\pi_L{=}0.2$ & SI & $+0.001$ & $0.024$ & $0.946$ & $2.92$ \\
$\pi_L{=}0.2$ & PPI++ & $+0.004$ & $0.050$ & $0.958$ & $1.39$ \\
$\pi_L{=}0.2$ & CP & $-0.0003$ & $0.021$ & $0.956$ & $3.36$ \\
\midrule
$\pi_L{=}0.5$ & LO & $+0.005$ & $0.043$ & $0.940$ & $1.00$ \\
$\pi_L{=}0.5$ & SI & $+0.001$ & $0.023$ & $0.946$ & $1.85$ \\
$\pi_L{=}0.5$ & PPI++ & $+0.002$ & $0.037$ & $0.940$ & $1.16$ \\
$\pi_L{=}0.5$ & CP & $-0.0003$ & $0.021$ & $0.956$ & $2.07$ \\
\bottomrule
\end{tabular}
\smallskip

{\scriptsize \noindent Notes: $n = 10{,}000$, $R = 500$, all-units five-fold cross-fitting, Monte Carlo band $[0.931, 0.969]$. This is an independent run of the DGP~7 design reported in Table~\ref{tab:dgp7}, with the composite proxy calibrated on a separately drawn library of $K_{\text{hist}} = 30$ rather than $50$ past experiments; the composite-proxy rows differ by library as well as by Monte Carlo variation, the other rows by Monte Carlo variation alone. AIPW is not run in this design.}
\end{minipage}\hfill
\begin{minipage}[t]{0.48\textwidth}
\centering
\caption{Full results for DGP 8 (Nonlinear Outcome).}
\label{tab:appendix-dgp8}
\scriptsize
\setlength{\tabcolsep}{2pt}
\begin{tabular}{@{}llrrrr@{}}
\toprule
\textbf{Config} & \textbf{Method} & \textbf{Bias} & \textbf{RMSE} & \textbf{Cov.} & \textbf{RE} \\
\midrule
$\pi_L{=}0.05$ & LO & $+0.007$ & $0.145$ & $0.958$ & $1.00$ \\
$\pi_L{=}0.05$ & SI & $+0.002$ & $0.026$ & $0.934$ & $5.49$ \\
$\pi_L{=}0.05$ & PPI++ & $+0.007$ & $0.096$ & $0.946$ & $1.51$ \\
$\pi_L{=}0.05$ & AIPW & $+0.007$ & $0.095$ & $0.946$ & $1.52$ \\
\midrule
$\pi_L{=}0.2$ & LO & $+0.001$ & $0.073$ & $0.942$ & $1.00$ \\
$\pi_L{=}0.2$ & SI & $+0.002$ & $0.026$ & $0.936$ & $2.84$ \\
$\pi_L{=}0.2$ & PPI++ & $+0.001$ & $0.051$ & $0.956$ & $1.42$ \\
$\pi_L{=}0.2$ & AIPW & $+0.001$ & $0.051$ & $0.954$ & $1.42$ \\
\midrule
$\pi_L{=}0.5$ & LO & $+0.003$ & $0.046$ & $0.942$ & $1.00$ \\
$\pi_L{=}0.5$ & SI & $+0.002$ & $0.025$ & $0.942$ & $1.80$ \\
$\pi_L{=}0.5$ & PPI++ & $+0.003$ & $0.038$ & $0.952$ & $1.21$ \\
$\pi_L{=}0.5$ & AIPW & $+0.003$ & $0.038$ & $0.952$ & $1.21$ \\
\bottomrule
\end{tabular}
\smallskip

{\scriptsize \noindent Notes: $n = 10{,}000$, $R = 500$, all-units five-fold cross-fitting, so the Monte Carlo band around the nominal 95\% is $[0.931, 0.969]$; every entry is inside it. RE is the RMSE ratio $\mathrm{RMSE}_{\mathrm{LO}}/\mathrm{RMSE}_{\mathrm{method}}$. The naive surrogate and the composite proxy are not run in this design.}
\end{minipage}
\end{table}

\begin{table}[!htbp]
\centering
\caption{Full results for DGP 9 (Antagonistic Surrogate).}
\label{tab:appendix-dgp9}
\scriptsize
\setlength{\tabcolsep}{2pt}
\begin{tabular}{@{}llrrrr@{}}
\toprule
\textbf{$\pi_L$} & \textbf{Method} & \textbf{Bias} & \textbf{RMSE} & \textbf{Coverage} & \textbf{RE} \\
\midrule
\multicolumn{6}{c}{\textit{Estimation Results}} \\
\midrule
0.05 & LO & $-0.003$ & $0.113$ & $0.938$ & $1.00$ \\
0.05 & NS & $-0.450$ & $0.450$ & $\mathbf{0.000}$ & $0.25$ \\
0.05 & SI & $-0.497$ & $0.497$ & $\mathbf{0.000}$ & $0.23$ \\
0.05 & PPI++ & $-0.004$ & $0.095$ & $0.948$ & $1.20$ \\
0.05 & CP & $-0.261$ & $0.262$ & $\mathbf{0.000}$ & $0.43$ \\
0.05 & AIPW & $-0.006$ & $0.095$ & $0.952$ & $1.20$ \\
0.05 & PPI++ (clipped) & $-0.004$ & $0.095$ & $0.948$ & $1.20$ \\
\midrule
0.2 & LO & $-0.001$ & $0.053$ & $0.956$ & $1.00$ \\
0.2 & NS & $-0.451$ & $0.451$ & $\mathbf{0.000}$ & $0.12$ \\
0.2 & SI & $-0.498$ & $0.498$ & $\mathbf{0.000}$ & $0.11$ \\
0.2 & PPI++ & $-0.0004$ & $0.045$ & $0.960$ & $1.18$ \\
0.2 & CP & $-0.260$ & $0.261$ & $\mathbf{0.000}$ & $0.20$ \\
0.2 & AIPW & $-0.001$ & $0.045$ & $0.964$ & $1.18$ \\
0.2 & PPI++ (clipped) & $-0.001$ & $0.045$ & $0.960$ & $1.18$ \\
\midrule
0.5 & LO & $-0.0001$ & $0.035$ & $0.962$ & $1.00$ \\
0.5 & NS & $-0.450$ & $0.450$ & $\mathbf{0.000}$ & $0.08$ \\
0.5 & SI & $-0.497$ & $0.497$ & $\mathbf{0.000}$ & $0.07$ \\
0.5 & PPI++ & $+0.0003$ & $0.032$ & $0.958$ & $1.10$ \\
0.5 & CP & $-0.261$ & $0.262$ & $\mathbf{0.000}$ & $0.13$ \\
0.5 & AIPW & $+0.0002$ & $0.032$ & $0.958$ & $1.10$ \\
0.5 & PPI++ (clipped) & $+0.0003$ & $0.032$ & $0.958$ & $1.10$ \\
\midrule
\multicolumn{6}{c}{\textit{SI--PPI++ Disagreement Diagnostic}} \\
\midrule
$\pi_L$ & Mean $\hat{D}$ & Mean $\widehat{\mathrm{SE}}(\hat{D})$ & \multicolumn{3}{l}{Rejection rate ($\alpha = 0.05$)} \\
\midrule
0.05 & $0.493$ & $0.090$ & \multicolumn{3}{l}{$1.000$} \\
0.2 & $0.498$ & $0.045$ & \multicolumn{3}{l}{$1.000$} \\
0.5 & $0.498$ & $0.028$ & \multicolumn{3}{l}{$1.000$} \\
\midrule
\multicolumn{6}{c}{\textit{Hybrid Estimator Results ($c = 1.5$)}} \\
\midrule
$\pi_L$ & Mean $\hat{\tau}_{\mathrm{hyb}}$ & Mean $w$ & Bias & RMSE & Coverage \\
\midrule
0.05 & $0.382$ & $0.052$ & $-0.028$ & $0.102$ & $\mathbf{0.918}$ \\
0.2 & $0.404$ & $0.012$ & $-0.006$ & $0.046$ & $0.960$ \\
0.5 & $0.408$ & $0.005$ & $-0.002$ & $0.032$ & $0.956$ \\
\bottomrule
\end{tabular}
\smallskip

{\scriptsize \noindent Notes: $n = 10{,}000$, $R = 500$, all-units five-fold cross-fitting, so the Monte Carlo band around the nominal 95\% is $[0.931, 0.969]$. ``PPI++ (clipped)'' is the binding-clip ablation, with the coefficient clipped to $[0, 1]$; every other estimation row is a primary configuration. The hybrid uses $c = 1.5$.}
\end{table}


\end{document}